\pdfoutput=1
\documentclass[a4paper,11pt]{article}
\usepackage[utf8]{inputenc}
\usepackage[margin=1in]{geometry}
\usepackage[toc,page]{appendix}
\usepackage{pdfpages}
\usepackage{mathrsfs}
\usepackage{amsfonts}
\usepackage{amsmath}
\usepackage{mathtools}

\usepackage{amssymb}
\usepackage{bbm}
\usepackage{amsthm}
\usepackage{graphicx}
\usepackage{centernot}
\usepackage{caption}
\usepackage{subcaption}
\usepackage{braket}
\usepackage{lastpage}
\usepackage{enumitem}
\usepackage{setspace}
\usepackage{xcolor}

\usepackage[english]{babel} 
\usepackage{dsfont}
\usepackage{upgreek}
\usepackage{mparhack}

\usepackage[square,sort,comma,numbers]{natbib}
\usepackage[colorlinks=true,allcolors=blue]{hyperref}

\usepackage{fancyhdr}

\usepackage{orcidlink}

\newcommand{\norm}[1]{\left\lVert #1 \right\rVert}
\newcommand{\abs}[1]{\left\lvert #1 \right\rvert}

\newcommand{\ceil}[1]{\left\lceil #1 \right\rceil}

\newcommand{\interior}[1]{%
	{\kern0pt#1}^{\mathrm{o}}%
}
\renewcommand{\braket}[1]{\left\langle#1\right\rangle}

\newcommand{\R}{\mathbb{R}}

\newcommand{\C}{\mathbb{C}}

\newtheorem{theorem}{Theorem}
\newtheorem{definition}[theorem]{Definition}
\newtheorem{proposition}[theorem]{Proposition}
\newtheorem{lemma}[theorem]{Lemma}
\newtheorem{corollary}[theorem]{Corollary}
\newtheorem{remark}[theorem]{Remark}
\newtheorem{assumption}[theorem]{Assumption}
\numberwithin{equation}{section}
\author{Johannes Agerskov\textsuperscript{1}\orcidlink{0000-0002-0533-3221}, Robin Reuvers\textsuperscript{2*}\orcidlink{0000-0003-2949-3614}, and Jan Philip Solovej\textsuperscript{3}\orcidlink{0000-0002-0244-1497}}

\date{
    1. NNF Quantum Computing Programme, Niels Bohr Institute, University of Copenhagen, Blegdamsvej 17, DK-2100 Copenhagen {\O}, Denmark
	\\
	2*. Universit\`{a} degli Studi Roma Tre, Dipartimento di Matematica e Fisica, L.go S.\ L.\ Murialdo 1, 00146 Roma, Italy\\
	3. QMATH, Department of Mathematics, University of Copenhagen, Universitetsparken 5, DK-2100 Copenhagen \O, Denmark\\
    \today\\\vspace{0.5cm}
    *Corresponding author. E-mail: robin.reuvers@uniroma3.it\\
    Contributing authors: johannes.agerskov@nbi.ku.dk; solovej@math.ku.dk}
\title{Ground State Energy of Dilute Fermi Gases in 1D}
\begin{document}
	\maketitle
	\begin{abstract}
        We study the one-dimensional spin-$J$ Fermi gas, interacting through a general repulsive 2-body potential, and prove asymptotics of the ground state energy in the dilute limit. The asymptotic behaviour is given in terms of the ground state energy of a spin chain, which is the Heisenberg antiferromagnet in the case of spin-1/2 fermions. We determine similar asymptotics for the Lieb--Liniger--Heisenberg model, showing that a product trial state used by Girardeau is not optimal.
	\end{abstract}
 \section{Introduction}
 \label{sec1}
A well-known feature of ground state energies of interacting quantum gases is that they show a kind of `universality' in the interaction potential when the density of the gas is low. The best known example is perhaps the validity of the Lee--Huang--Yang energy expansion for 3D dilute Bose gases with general potentials \cite{lee1957eigenvalues,dyson1957ground,lieb1998ground,yau2009second,fournais2020energy,basti2021new,fournais2021energy,basti2024upper, haberberger2023free, haberberger2024upper,fournais2024free,brooks2025third}, but other set-ups have been studied extensively in the mathematics literature in recent years, such as 2D dilute Bose gases \cite{schick1971two,lieb2001ground,andersen2002ground,mora2009ground,fournais2024ground,fournais2024lower} and 3D dilute Fermi gases \cite{lieb2005ground,lauritsen2024ground,lauritsen2024pressure,giacomelli2024huang,giacomelli2025huang,chen2025second} (see the introductions of the papers cited for further references).

In contrast, in one dimension, dilute Bose and Fermi gases interacting with general potentials have not been investigated as much.\footnote{Note that weakly-interacting gases in 1D are not well described by the dilute limit, see Section~\ref{notweakly} for a discussion.} To address this, we previously studied the 1D Bose gas \cite{agerskov2022ground}, and considered a gas of $N$ spinless bosons in the interval $[0,L]$ (with density $\rho=N/L$), interacting via the $N$-body Hamiltonian 
\begin{equation}
	\label{H_N}
	H_N=-\sum^N_{i=1}\partial^2_{x_i}+\sum_{1\leq i<j\leq N}v(x_i-x_j),
\end{equation}
with $v$ a potential with even-wave scattering length $a_e$ (defined in \eqref{dyson1} below). We showed that for a wide class of potentials $v$, the ground state energy in the \textit{low-density limit} $\rho |a_e|\ll1$ can be expanded as 
\begin{equation}
\label{someeq1276}
N\frac{\pi^2}{3}\rho^2\left(1+2\rho a_e+\dots\right), 
\end{equation}
(where error terms have been omitted to have an informal discussion; see \cite[Theorem~1]{agerskov2022ground} for precise statements). 
This expression is in some sense `universal' because only the even-wave scattering length appears to first order in the low-density ground state energy expansion, independent of the exact shape of the potential $v$. As indicated, a stronger `universality' (continuing the expansion to higher orders and finding that it depends only on the scattering length) has been demonstrated in 2 and 3D, but a crucial difference with higher dimensions is that in 1D, the leading-order term in the energy expansion is the ground state energy of the free Fermi gas. This means that the techniques used to prove \eqref{someeq1276} differ significantly from those used in 2 and 3D (see \cite[Section~1.2]{agerskov2022ground} for a summary).

To better understand \eqref{someeq1276}, we can study relevant trial wave functions. The leading-order term in \eqref{someeq1276} is the ground state energy of the free Fermi gas in 1D. For spinless bosons, it can be found by using the wave function 
\begin{equation}
\label{someeq1123}
\prod_{1\leq i<j\leq N}|\sin(\pi \frac{x_j-x_i}{L})|,
\end{equation}
which, for odd $N$, is the (unnormalized, periodic b.c.) 1D free Fermi ground state on the sector $0\leq x_1<\dots<x_N<L$, extended symmetrically to other orderings of the particles to render the wave function bosonic. This is the exact ground state of the Tonks--Girardeau gas \cite{girardeau1960relationship}, which corresponds to the low-density limit ($\rho/c\to0$) of the Lieb--Liniger gas ($v=2c\delta$ in \eqref{H_N} \cite{lieb1963exact}). This is effectively the state around which we expand to find the expansion \eqref{someeq1276}, while the expansions for the dilute Bose gas in 2 and 3D use the Bose condensate as a point of departure.

The next term in \eqref{someeq1276}---the correction involving the scattering length $a_e$---can be made intuitive by studying the trial state that was used to obtain the upper bound in \cite{agerskov2022ground}. The idea is that \eqref{someeq1123} needs to be modified in those areas of the configuration space where at least two bosons are close. In those regions, the potential should impact the wave function, and this needs to be accounted for. A 2-body energy minimization problem helps to define a good modification: for potentials with compact support contained in $[-R_0,R_0]$, we minimize the 2-body energy on some suitable scale $[-R,R]$ with $R>R_0$,
\begin{equation}
		\label{dyson1}
	\inf_{\substack{\psi\in H^1[-R,R]\\\psi(R)=\psi(-R)=1}}\int^{R}_{-R}2|\partial_x\psi|^2+v(x)|\psi(x)|^2d x=\int^{R}_{-R}2|\partial_x\psi_{\text{even}}|^2+v(x)|\psi_{\text{even}}(x)|^2d x= \frac{4}{R-a_e}. 
\end{equation}
Here, one should think of $x$ as the distance between the particles, so that the boundary condition $\psi(R)=\psi(-R)=1$ is truly a bosonic one. The minimum is attained by the \textit{even-wave scattering solution} $\psi_{\text{even}}$, while the minimal energy defines the \textit{even-wave scattering length} $a_e$. It turns out that $\psi_{\text{even}}$ is a good tool to modify \eqref{someeq1123}. If we focus on the sector $0\leq x_1<\dots<x_N< L$, the trial state producing \eqref{someeq1276} is roughly
\begin{equation}
\label{state1}
\prod_{1\leq i<j\leq N} f_{i,j}(x_i,x_j),
\end{equation}
with 
\[
f_{i,j}(x_i,x_j)=\begin{cases}
\sin(\pi \frac{x_j-x_i}{L}) & \phantom{w} j>i+1 \textnormal{ or } |x_j-x_i|>R\\
\sin(\pi \frac{R}{L})\psi_{\textnormal{even}}(x_{i+1}-x_i)& \phantom{w} j=i+1 \textnormal{ and } |x_{i+1}-x_i|\leq R.
\end{cases}
\]
That is, we introduce the scattering solution for neighbouring particles, on a suitable scale $R$. Note the wave function is continuous. It suffices to consider neighbouring particles because we restricted to  $0\leq x_1<\dots<x_N< L$; the full trial state is found by symmetrically extending to other orderings of the particles on the line. 

In \cite{agerskov2022ground}, we used a slightly more complicated trial state of the type \eqref{state1} to obtain the upper bound needed for \eqref{someeq1276}, and we also proved a matching lower bound. Intuitively speaking, the appearance of the scattering length $a_e$ in the energy expansion \eqref{someeq1276} in the dilute limit $\rho|a_e|\ll1$ comes from the importance of the problem \eqref{dyson1}---the gas is so dilute that, to first order, interactions are successfully described by 2-body scattering only.

A similar result was found for the spinless (or spin-polarized) Fermi gas with Hamiltonian \eqref{H_N} \cite{agerskov2022ground} (see also an improved upper bound \cite{lauritsen2024ground} that followed our preprint). This requires only a small modification to \eqref{dyson1}: 
\begin{equation}
		\label{dyson2}
\inf_{\substack{\psi\in H^1[-R,R]\\ \psi(R)=-\psi(-R)=1}}\int^{R}_{-R}2|\partial_x\psi|^2+v(x)|\psi(x)|^2d x=\int^{R}_{-R}2|\partial_x\psi_{\text{odd}}|^2+v(x)|\psi_{\text{odd}}(x)|^2d x= \frac{4}{R-a_o}, 
\end{equation}
where the boundary condition $\psi(R)=-\psi(-R)=1$ is now fermionic. This defines the \textit{odd-wave scattering solution} $\psi_{\text{odd}}$ and the \textit{odd-wave scattering length} $a_o$ (we have $a_e\leq a_o$, since $\abs{\psi_{\text{odd}}}$ is a trial function for \eqref{dyson1} with the same energy as $\psi_{\text{odd}}$). Replacing $\psi_{\text{even}}$ by $\psi_{\text{odd}}$ in the trial state \eqref{state1}, and extending antisymmetrically to get a fermionic wave function gives the upper bound needed to prove a ground state energy expansion
\begin{equation}
\label{energ2}
N\frac{\pi^2}{3}\rho^2\left(1+2\rho a_o+\dots\right),
\end{equation}
again for a general class of potentials. The lower bound could also be adapted to this case. Note the leading order is again the free Fermi ground state---this is independent of the exact symmetry of the particles.

The logical next question is: what happens for a dilute 1D spin-1/2 Fermi gas with Hamiltonian \eqref{H_N}? 
We can expect a free Fermi leading-order term like before, but this time, spin is present, and the wave function should be antisymmetric under the combined exchange of space and spin. If we again want to proceed by modifying the free Fermi state with 2-body scattering solutions, this means both even-wave and odd-wave scattering are needed: depending on the combined spin of the two particles involved in the scattering, the spatial wave function has to be symmetric or antisymmetric under exchange of the spatial coordinates. Indeed, in \cite[Remark~5]{agerskov2022ground}, we conjectured that the correct ground state energy is
\begin{equation}
\label{energ3}
		N\frac{\pi^2}{3}\rho^2\big(1+2\ln(2)\rho a_{e}+2(1-\ln(2))\rho a_{o}+\dots\big), 
\end{equation}
which is based on the exactly-solvable cases of the hard-core gas and the Lieb--Liniger gas (in analogy with the original conjecture for bosons \cite{astrakharchik2010low,agerskov2022ground}).
Demonstrating \eqref{energ3} rigorously is the main motivation for this article (though the results extend beyond this case, see Section~\ref{secmainresults}).

How should the trial state \eqref{state1} be modified? Again restricting to $0\leq x_1<\dots<x_N< L$, and adding a spin wave function $\phi_{\text{spin}}$, we can study
\begin{equation}
\label{newtrialstate}
\left(\prod_{1\leq i<j\leq N} f_{i,j}(x_i,x_j)\right)\phi_{\text{spin}}(\sigma_1,\dots,\sigma_N),
\end{equation}
with 
\begin{equation}
\label{someeq320990}
\begin{aligned}
&f_{i,j}(x_i,x_j)=\\
&\begin{cases}
\sin(\pi \frac{x_j-x_i}{L}) & \phantom{w} j>i+1 \text{ or } |x_j-x_i|>R\\
\sin(\pi \frac{R}{L})\left(\psi_{\text{even}}(x_{i+1}-x_i)P^{i,i+1}_A+\psi_{\text{odd}}(x_{i+1}-x_i)P^{i,i+1}_S\right)& \phantom{w} j=i+1 \text{ and } |x_{i+1}-x_i|\leq R,
\end{cases}
\end{aligned}
\end{equation}
where $P^{i,i+1}_A$ is the antisymmetric spin projection (spin zero) for particles $i$ and $i+1$ and $P^{i,i+1}_S$ is the symmetric spin projection (spin one) (again note that we only get these operators for neighbouring particles by construction). In this informal discussion, we think of at most one pair of neighbours being close at a time, so that the order of the factors in \eqref{newtrialstate} does not matter. Thus, antisymmetric spin gets coupled to even-wave scattering and vice versa, which means the trial state can be extended fermionically to other orderings of the particles. Note that the spin state $\phi_{\text{spin}}$ still needs to be chosen. The even-wave scattering function $\psi_{\text{even}}$ has a lower energy than $\psi_{\text{odd}}$ because even boundary conditions give a lower energy, which means that the spin singlet is preferred for neighbouring spins. Indeed, calculations similar to the ones that gave \eqref{someeq1276} and \eqref{energ2} for spinless particles now produce a first-order correction that is 
\[
2\rho a_{e}\left\langle\frac{1}{N}\sum_iP^{i,i+1}_A\right\rangle_{\phi_{\text{spin}}}+2\rho a_{o}\left\langle\frac{1}{N}\sum_iP^{i,i+1}_S
\right\rangle_{\phi_{\text{spin}}}.
\]
Given $a_e\leq a_{o}$ as pointed out before, to achieve the lowest energy we need to maximize the first expectation value (and hence minimize the second, given $P^{i,i+1}_S+P^{i,i+1}_A=\mathds{1}$). This is exactly the Heisenberg spin chain since
\[
2P^{i,i+1}_S=(S_i+S_{i+1})^2=\frac{3}{2}+2S_i\cdot S_{i+1},
\]
and the minimal expectation value is attained by the ground state $\phi_{\text{Heis}}$ of the Heisenberg chain \cite{bethe1931theorie,hult1938},
\[
\left\langle\frac{1}{N}\sum_iP^{i,i+1}_S\right\rangle_{\phi_{\text{Heis}}}=\frac{3}{4}+\left\langle\frac{1}{N}\sum_iS_i\cdot S_{i+1}\right\rangle_{\phi_{\text{Heis}}}=1-\ln(2).
\]
Formally, this gives the upper bound needed to prove  \eqref{energ3}. In other words, the ground state problem for the dilute, interacting spin-$1/2$ Fermi gas in 1D effectively reduces to the antiferromagnetic Heisenberg spin chain. 
More generally, ground state problems of 1D, dilute spin-$J$ Fermi gases give rise to an effective spin chain that is known as the Lai--Sutherland model \cite{lai1974lattice,sutherland1975model,volosniev2014strongly,deuretzbacher2014quantum,yang2015strongly,levinsen2015strong}, 
\begin{equation}
\label{LSintro}
H_{\textnormal{LS}}=\sum_{i=1}^{N}P_S^{i,i+1},
\end{equation}
where $P^{i,i+1}_S$ should now be understood as the projection onto all symmetric spin combinations of spins $i$ and $i+1$ (for example, if $J=5/2$, these would be combined spin $1$, $3$ and $5$, whereas spin $0$, $2$ and $4$ are antisymmetric).
This is not difficult to understand: as indicated above, in a dilute gas, the first order in the low-density energy expansion is effectively decided by 2-body scattering, which can be either even or odd wave. This means that the effective spin chain should distinguish even and odd spin combinations, and that is exactly what \eqref{LSintro} does. We again find \eqref{someeq320990} with the spin-$J$ interpretation of $P^{i,i+1}_S$ and $P^{i,i+1}_A=\mathds{1}-P^{i,i+1}_S$. Note that for spin-$J$ bosons, $a_e$ and $a_o$ switch places because the wave function needs to be symmetric upon combined exchange of space and spin coordinates. This means the correct Hamiltonian to study is now $-H_{\textnormal{LS}}$. In this article, we demonstrate the relevance of the effective spin chains rigorously, for a general class of potentials $v$. 

In the physics literature, interacting 1D gases with spin have been studied a lot (see \cite{guan2013fermi,minguzzi2022strongly,bloch2008many,guan2022yang} for reviews), particularly with the contact interaction that first appeared in the Lieb--Liniger Hamiltonian \cite{lieb1963exact}
\begin{equation}
	\label{H_Ncontact}
	H_{\textnormal{LL}}=-\sum^N_{i=1}\partial^2_{x_i}+\sum_{1\leq i<j\leq N}2c\delta(x_i-x_j).
\end{equation}
Following the introduction of the bosonic model and the exact Bethe ansatz solution by Lieb and Liniger in 1963, Gaudin and Yang independently found an exact solution for particles with spin in 1967 \cite{yang1967some,gaudin1967systeme}, with Sutherland extending the result to general symmetry types given by Young diagrams in 1968 \cite{sutherland1968further}. The Hamiltonian \eqref{H_Ncontact} for particles with spin is now known as the Yang--Gaudin model. Despite the exact solution, the connection with effective spin chains such as the Heisenberg model was not made at the time. 

This appears to have been done for the first time in the discrete set-up (the 1D Fermi--Hubbard model) by Ogata and Shiba in 1990 \cite{ogata1990bethe}. Interestingly, they were inspired by the Bethe ansatz solution by Lieb and Wu \cite{lieb1968absence}, who in turn adapted Yang and Gaudin's solution of the continuum model to the discrete case. Eventually, the spin chain observation made it back to the continuous setting \cite{guan2007ferromagnetic,matveev2008spectral} for spin-$1/2$ bosons with contact interactions. The possibility to study fermions and more general interactions was hinted at as well. For contact interactions between fermions, the phenomenon was studied in more detail with perturbation theory methods starting from 2014 \cite{volosniev2014strongly,deuretzbacher2014quantum,yang2015strongly,levinsen2015strong}. It was found that the relevant effective spin chain is the Lai--Sutherland model \eqref{LSintro}.

 After 2014, it was noticed that the effective spin chain can be engineered to differ from the Lai--Sutherland model, for example with external trapping potentials, spin-dependent interactions, or by considering several species of particles \cite{volosniev2015engineering,massignan2015magnetism,yang2016engineering, deuretzbacher2017tuning}. Effective spin chains have also been observed experimentally \cite{murmann2015antiferromagnetic}. 

Our main result is a rigorous derivation of the appearance of the Lai--Sutherland model for the ground state energy, in the absence of a trapping potential. It applies to general interaction potentials \eqref{H_N}, and therefore extends what has been done for contact interactions in the literature. 

This paper is organized as follows. Section~\ref{resultsfermi} describes our main result for spin-$J$ fermions, while Section~\ref{proofideas} describes some key ideas used in the proofs. Section~\ref{resultsLLH} contains results about the Lieb--Liniger--Heisenberg model and Section~\ref{openproblems} suggests a few open problems. It turns out the upper bounds can be formulated in slightly greater generality using matrix-valued potentials. Such potentials are introduced in Section~\ref{matrixpotentials}, while the upper bound is stated and proved in Section~\ref{secupperbound}. Finally, Section~\ref{seclowerbound} describes the lower bound and its proof, including the lower bound for the Lieb--Liniger--Heisenberg model.

\subsection{Luttinger theory and weakly-interacting 1D gases}
\label{notweakly}
This paper studies the ground state energy of dilute gases of particles with spin in 1D, and shows that the first correction term in the expansion only depends on the even- and odd-wave scattering lengths. This follows work that showed that for dilute gases in 2D and 3D, the ground state energy is even `universal' for higher terms in the expansion (see the references at the start of the introduction). 

Of course, the study of universal properties of 1D gases has a long history, in particular through Tomonaga--Luttinger theory, bosonization, and the Luttinger liquid universality class \cite{Tomonaga,luttinger1963exactly,mattis1965exact,mattis1964band,haldane1981luttinger,metzner1993conservation,giamarchi2003quantum}. An important conclusion is that the theory of excitations of a large class of 1D models has a universal shape, independent of the exact microscopic model. However, this fact does not directly seem to comment on the ground state energy of the microscopic model, and how it depends on the shape of the interaction potential, which is the topic of interest here. 

In the mathematical physics literature, many results have been obtained by Benfatto, Gallavotti, Mastropietro and others for weakly-interacting Fermi gases \cite{benfatto1990perturbation, benfatto1994beta,benfatto2011drude, benfatto2014universality,benfatto2014universality2}. Using renormalization group techniques, these papers demonstrate a notion of universality, commonly used in the physics literature, that concerns the ground state energy, but also goes a great deal beyond that. 

However, while in 2D and 3D the dilute and weakly-interacting limits are related, this is not the case in 1D. This is most easily seen in the Lieb--Liniger and Yang--Gaudin models, whose solution can entirely be expressed in terms of the parameter $\gamma=c/\rho$ (with $2c$ the parameter in front of the $\delta$-potential, see \eqref{H_Ncontact}) \cite{lieb1963exact}. This means that the case of ``fixed potential and low density'' considered in this article corresponds to $\gamma\gg1$, while the case of ``fixed density and weak interactions'' that is the focus of the references above corresponds to $\gamma\ll1$. Hence, for the models with a $\delta$-interaction, the weakly-interacting regime and the low-density regime are opposites. 

It is easy to carry out a similar scaling for potentials that are not concentrated at the origin like the $\delta$-function. In that case, the scaling that is needed to go between the two set-ups will also affect the range (or concentration) of the potential, but the conclusion remains the same: the regime considered in this article corresponds to strong interactions in the set-up used in the works by Benfatto, Gallavotti, Mastropietro and others, and so the results do not overlap.

 \section{Main results}
 \label{secmainresults}
\subsection{\texorpdfstring{The ground state energy of the dilute spin-$J$ Fermi gas}{The ground state energy of the dilute spin-J Fermi gas}}
\label{resultsfermi}
 We consider a general even, compactly-supported, repulsive (i.e.\ positive) potential $v$ that is a measure. We then study the Hamiltonian 
   \begin{equation}
  \label{EqHamiltonian}
      H=-\sum_{i=1}^{N}\partial_i^2+\sum_{i<j}v(x_i-x_j),
  \end{equation}
  with energy functional given by
 \begin{equation}
 \label{functional}
\mathcal{E}(\Psi)=\int_{[0,L]^N} \sum_{i=1}^{N}\braket{\partial_i\Psi\vert \partial_i\Psi}+\sum_{i<j}\braket{\Psi \vert v(x_i-x_j)\Psi},
 \end{equation}
and with form domain 
\begin{equation}
\label{domainfermions}
\left\{\Psi\in \wedge^NL^2([0,L];\C^{2J+1}) \big\vert \Psi\in H^1([0,L]^N; (\C^{2J+1})^{\otimes N}) \textnormal{ and } \mathcal{E}(\Psi)<\infty \right\}.
 \end{equation}
 We will not assume $v$ to be absolutely continuous with respect to the Lebesgue measure, but we abuse notation and write $v(x)dx$ when integrating with respect to $v$.
 
 We denote by $P_S$ the projection onto the symmetric subspace $\C^{2J+1}\vee \C^{2J+1}$, while $P_A$ is the projection onto the antisymmetric subspace $\C^{2J+1}\wedge \C^{2J+1}$. We will add indices $i,j$ to indicate the spin projection is acting on particles $i$ and $j$, e.g.\ $P_A^{i,j}$. The even- and odd-wave scattering lengths $a_e$ and $a_o$ were mentioned in \eqref{dyson1} and \eqref{dyson2} (see \cite[Lemma~2 and Appendix~A]{agerskov2022ground} for a precise definition, and Definition~\ref{DefinitionScatteringLengthMatrix} below for a more general definition in the case of matrix-valued potentials). 
 
 \begin{theorem}[Main result for fermions]
	\label{TheoremDiluteFermiGroundStateEnergy}
	 Let $v$ be an even, repulsive (i.e.\ positive) interaction compactly supported in $[-R_0,R_0]$, with $a_e$ and $a_o$ its even- and odd-wave scattering lengths. Consider a spin-$J$ Fermi gas with energy functional \eqref{functional} and domain \eqref{domainfermions}, with Dirichlet, Neumann or Robin boundary conditions with non-negative parameter. Let $\rho=N/L$. Then, for $\rho|a_e|$, $\rho a_o$ and $\rho R_0$ sufficiently small, the ground state energy $E_{J}(N,L)$ satisfies
	\begin{equation}
    \label{eqlower}
		E_{J}(N,L)= N\frac{\pi^2}{3}\rho^2\left(1+2\rho \left[a_e+(a_o-a_e)\epsilon^{\textnormal{LS}}_J\right]+\mathcal{O}\big((\rho \max(R_0,a_o,\abs{a_e}))^{6/5}+N^{-1}\big)\right),
	\end{equation}
	where $\epsilon^{\textnormal{LS}}_J$ is the minimal energy per site of the spin-$J$ Lai--Sutherland model,\footnote{\label{lastfootnote}We can assume periodic boundary conditions for the spin chain, but it does not matter for the theorem: other boundary conditions change $\epsilon^{\textnormal{LS}}_J$ by at most $N^{-1}$, which can be absorbed in the error term.} i.e.\ the minimal expectation value of
    \begin{equation}
    \label{spinchainmt}
h_{\textnormal{LS}}=\frac{1}{N}\sum_{i=1}^{N}P_S^{i,i+1}.
    \end{equation}
\end{theorem}

The theorem follows from an upper bound proved for Dirichlet boundary conditions (Corollary~\ref{TheoremUpperBound2}, stated and proved in Section~\ref{secupperbound}) and a lower bound proved for Neumann boundary conditions (Theorem~\ref{TheoremLowerBoundSpin} in Section~\ref{seclowerbound}): a Robin boundary condition with non-negative parameter adds a non-negative boundary term to the energy, so its ground state energy lies between these two. The upper bound is a corollary of a more general result that also covers matrix-valued potentials (Theorem~\ref{TheoremUpperBoundSpinJFermi} in Section~\ref{secupperbound}), which has the advantage that it applies to the Lieb--Liniger--Heisenberg model, where it shows that a product trial state used by Girardeau is not optimal (see Section~\ref{resultsLLH}). The lower bound is only proved for scalar potentials, and for the Lieb--Liniger--Heisenberg model. 

\begin{remark}
For bosons, we can prove a similar theorem, with an analogous proof. The only difference is that for bosons, symmetric spin ($P_S$) gets coupled to symmetric spatial behaviour ($a_e$), while $P_A$ gets coupled to $a_o$, so that $a_e$ and $a_o$ swap roles. However, since the total ground state is known to be spatially symmetric \cite{eisenberg2002polarization}, this means that it was already covered by our result for the spinless Bose gas \cite{agerskov2022ground}: a fully polarized state has $\langle h_{\textnormal{LS}}\rangle=1$, so that $a_o+(a_e-a_o)\langle h_{\textnormal{LS}}\rangle=a_e$, and this gives the regular bosonic energy expansion.
\end{remark}

\begin{remark}
    The assumption that the potential is compactly supported is not necessary: we may allow for potentials decaying fast enough to satisfy $\int_{b}^\infty v(x)x^{6+\epsilon}d x<\infty$ for some $b\geq0$ and some $\epsilon>0$. In that case, the error term in \eqref{eqlower} is slightly more complicated. For simplicity, we only discuss the case of compactly-supported potentials in this paper, and refer to \cite{agerskov2022ground} for details on how to include potentials of non-compact support.
\end{remark}

\begin{remark}
\label{remarknopotential}
Note that the notion of diluteness needed for the theorem is that $\rho|a_e|$, $\rho a_o$ and $\rho R_0$ are all small (where the absolute values are needed for $a_e$ because it can be negative, and the error term involving $R_0$ cannot be avoided, see \cite[Remarks~2 and~3]{agerskov2022ground}). Since the free case $v=0$ has $a_e=-\infty$ (see \cite[Remark~4]{agerskov2022ground}), it cannot be considered dilute and $\rho|a_e|$ in the theorem is infinite.
\end{remark}

\begin{remark}
\label{remarkrangeintro}
The theorem covers scalar potentials $v\geq0$ with even- and odd-wave scattering lengths $a_e$ and $a_o$. These satisfy $a_e\leq a_o$ and $a_o\geq0$. The full range of scattering lengths defined by those two inequalities can be attained by varying the potential (where $a_e=a_o$ requires a hard core), see Remark~\ref{remarkrange} in Appendix~\ref{AppendixA}. 
\end{remark}

    Note that the Lai--Sutherland model is the antiferromagnetic Heisenberg model for spin $1/2$. It is solved exactly by the Bethe ansatz for any spin $J$ \cite{lai1974lattice,sutherland1975model}.\footnote{This means that the Bethe ansatz gives exact eigenstates of the Hamiltonian. Whether the Bethe ansatz states form a complete set is a more delicate problem, see \cite{takahashi1971one, cmp/1104252974}. Yang and Yang proved the validity of the Bethe ansatz for the ground state for the spin-$1/2$ Heisenberg chain in \cite{yang1966one}.} The thermodynamic ground state energy per site is \cite{sutherland1975model} 
     \begin{equation}
\lim_{N\to\infty}\epsilon^{\textnormal{LS}}_J=1-\frac{1}{2J+1}\left[\psi(1)-\psi(1/(2J+1))\right],
	\end{equation}
	where $\psi(x)=\frac{d}{d x}\ln\left(\Gamma(x)\right)$ is Euler's digamma-function. Hence, the thermodynamic energy per site of the spin-$1/2$ antiferromagnetic Heisenberg chain is $\lim\limits_{N\to\infty}\epsilon^\textnormal{LS}_{1/2}=1-\ln(2)$ (also see \cite{mattis2012theory,hult1938}).\\
	For a finite chain of length $N$, Lemma~\ref{LemmaLaiSutherlandFiniteN} in Appendix~\ref{AppendixB} shows that the ground state energy per site satisfies
	\begin{equation}
    \label{some310}
		\epsilon^{\textnormal{LS}}_J=1-\frac{1}{2J+1}\left[\psi(1)-\psi(1/(2J+1))\right]+\mathcal{O}(1/N).
	\end{equation}
Noting that $a_o\geq a_e$ and combining Theorem~\ref{TheoremDiluteFermiGroundStateEnergy} and \eqref{some310} shows that the ground state energy of the dilute spin-$1/2$ Fermi gas is 
\begin{equation}
\label{energy3}
N\frac{\pi^2}{3}\rho^2\left(1+2\rho\ln(2) a_e+2\rho(1-\ln(2))a_o+\mathcal{O}\left((\rho \max(R_0,a_o,\abs{a_e}))^{6/5}+N^{-1}\right)\right),
\end{equation}
which proves a conjecture that we made in \cite[Remark~5]{agerskov2022ground}.

\begin{remark}\label{remarkexamples}
    Here are two exactly-solvable models that are consistent with \eqref{energy3}.\\
    \textbf{The hard-core potential} of radius $a>0$ has $a_e=a_o=a$, and the ground state is $(2J+1)^N$-fold degenerate, as there is one ground state for each spin state. The ground state energy in the dilute limit $\rho a\ll1$ is given by
    \[
N\frac{\pi^2}{3}\left(\frac{N}{L-Na}\right)^2=N\frac{\pi^2}{3}\rho^2\left(1-\rho a\right)^{-2}=N\frac{\pi^2}{3}\rho^2\left(1+2\rho a+\mathcal{O}\left((\rho a)^2\right)\right).
\]
\textbf{The (fermionic) Yang--Gaudin model} is a Fermi gas with point interactions, $v(x)=2c \delta(x)$. It was solved with the Bethe ansatz by Gaudin and Yang in \cite{gaudin1967systeme,yang1967some}, and the solution was generalized to any spin by Sutherland in \cite{sutherland1968further}. The ground state energy per unit length in the thermodynamic limit, for $\rho/c\ll1$, is
\begin{equation}
\label{EqYGGroundStateEnergy}
\frac{\pi^2}{3}\rho^3\left(1+2\ln(2)\rho a_e+\mathcal{O}\left((\rho/c)^2\right)\right),
\end{equation}
where $a_e=-2/c$ is the even-wave scattering length of the interaction $v$ (note $a_o=0$).
\end{remark}

\subsection{Ideas used in the proof}
\label{proofideas}
The most important ingredients in the proof are generalizations of the minimizations \eqref{dyson1} for symmetric wave functions and \eqref{dyson2} for antisymmetric wave functions. Those were used to prove the first-order ground state energy expansions for dilute spinless Bose and Fermi gases in \cite{agerskov2022ground}. For clarity, we first summarize the strategies used in the spinless case (see \cite[Section~1.2]{agerskov2022ground} for a more extensive summary), followed by the adaptations needed for the case with spin.

The upper bound uses the trial state argument outlined in Section~\ref{sec1}: for spinless bosons, the trial state departs from Girardeau's solution of the impenetrable Bose gas \eqref{someeq1123}, which is the free Fermi solution with absolute values to create a symmetric wave function. This produces the leading order in the expansion. The Girardeau wave function is then patched with the scattering solution \eqref{dyson1} on some suitable small scale, and for the closest pair of particles only, which, up to negligible error terms in the limits $\rho a_e\to0$ and $\rho R_0\to0$ (with $R_0$ the range of a compactly supported potential), gives our first-order result. 

The case of spin-$J$ fermions is similar, but it requires the spin-dependent patching discussed above in \eqref{newtrialstate} and \eqref{someeq320990}, leading to a first-order correction related to the Lai--Sutherland model \eqref{LSintro}. Given that we will actually prove the upper bound for the more general case of matrix-valued potentials, the proof requires a few additional ideas. The most important one is the matrix generalization of \eqref{dyson1} and \eqref{dyson2}: the matrix-valued scattering solution discussed in Section~\ref{matrixpotentials}. The precise trial state is introduced in Section~\ref{sectrialstate}. The upper bound is found by estimating the energy of the trial state, and the proof involves calculations that reduce to the free Fermi ground state on which the trial state is based, as well as the matrix scattering solution of Definition~\ref{DefinitionScatteringLengthMatrix}. 

The lower bound for spinless bosons in \cite{agerskov2022ground} again uses the minimization \eqref{dyson1}. An important idea (see \cite{dyson1957ground,lieb1998ground} in the 3D case) is that \eqref{dyson1} can be written as
\begin{equation}
	\label{eqidea}
	\int^{R}_{-R}2|\partial_x \psi|^2+v(x)|\psi(x)|^2dx\geq \frac{2}{R-a_e}\int (\delta_{R}(x)+\delta_{-{R}}(x))|\psi(x)|^2dx,
\end{equation}
for all $\psi\in H^1[-R,R]$ (note the boundary condition $\psi(R)=\psi(-R)=1$ in the minimization in  \eqref{dyson1}). This effectively lower bounds the kinetic and potential energy on $[-R,R]$ by a symmetric delta potential at radius $R$. To generalize this to an $N$-boson wave function, we repeatedly apply \eqref{eqidea} to obtain the symmetric delta potential for any neighbouring pairs of bosons. Inspired by \cite{lieb2004one}, we then throw away the regions where $|x_{i+1}-x_i|\leq R$ (which gives a lower bound since $v$ is positive) and as a consequence, the two delta functions at radius $|x_{i+1}-x_i|= R$ collapse into a single delta at $|x_{i+1}-x_i|= 0$, with value $4/(R-a_e)$. This corresponds to the Lieb--Liniger energy of some wave function on the reduced interval, which is of course lower bounded by the explicitly known Lieb--Liniger ground state energy \cite{lieb1963exact}. Heuristically, this corresponds to the following calculation: starting with an interval of length $L$, cut it back to length $L-(N-1)R$, so that the Lieb--Liniger ground state energy (see e.g.\ Lemma~\ref{LemmaLiebLinigerNeumannLowerBound}) with $c=2/(R-a_e)$ and new density $N/(L-(N-1)R)=\rho(1+\rho R+\dots)$ gives a lower bound
\begin{equation}
\label{heurist}
N\frac{\pi^2}{3}\rho^2(1+2\rho R+\dots)(1-2\rho(R-a_e)+\dots)=N\frac{\pi^2}{3}\rho^2(1+2\rho a_e+\dots).
\end{equation}
As is done in \cite{lieb2004one}, \cite{agerskov2022ground} also shows that the ground state wave function has little weight in the regions that get thrown out, so that \eqref{heurist} is accurate.

For spin-$J$ fermions, the strategy is the same, but some additional ideas are needed. The most important step is to combine symmetric scattering \eqref{dyson1} and antisymmetric scattering \eqref{dyson2} into a single (Dyson's) lemma involving the spin projections $P_S$ and $P_A$, see Lemma~\ref{LemmaDysonSpin1/2Fermi} in Section~\ref{secprep}. Removing the regions where neighbours are close then leads to a contact interaction between neighbouring particles whose strength depends on their spins. We handle this by moving the spin-dependent part of the interaction to alternate pairs of neighbours, where the spin projections commute. Two further steps differ from \cite{agerskov2022ground}: we replace the exact solution of the Lieb--Liniger model by a simpler lower bound (Lemma~\ref{LemmaLiebLinigerNeumannLowerBound}), and instead of showing that the removed regions carry little weight, we keep part of the energy to pay for them (Lemma~\ref{LemmaNormLossImproved}). A summary of the additional ideas needed to solve the case with spin $J$ is given at the start of Section~\ref{SectionLowerBoundProof}.

\subsection{\texorpdfstring{The ground state energy of the Lieb--Liniger--Heisenberg model}{The ground state energy of the Lieb--Liniger--Heisenberg model}}
\label{resultsLLH}
The Lieb--Liniger--Heisenberg (LLH) model is an effective model for spin-$1/2$ fermions in a tight waveguide introduced by Girardeau \cite{girardeau2004theory,girardeau2006ground}. It has energy functional 
\begin{equation}
\label{functionalLLH}
     \mathcal{E}^{\textnormal{LLH}}(\Psi)=\int_{[0,L]^N} \sum_{i=1}^{N}\braket{\partial_i\Psi\vert \partial_i\Psi}+\sum_{i<j}\braket{\Psi \big| 2c'P_A^{i,j}\delta(x_i-x_j)+2cP_S^{i,j}\delta(x_i-x_j) \big|\Psi},
 \end{equation} 
 with $c',c>0$, and its spatially symmetric domain is
  \begin{equation}
  \label{domainLLH}
     \left\{\Psi\in \big(L^2[0,L]\big)^{\vee N}\otimes\big(\C^{2}\big)^{\otimes N}\big\vert \Psi\in H^1([0,L]^N; (\C^{2})^{\otimes N} )\textnormal{ and } \mathcal{E}^{\textnormal{LLH}}(\Psi)<\infty \right\}.
 \end{equation}
Here, as in \cite{girardeau2006ground}, the spins are labelled by position: the $k$-th factor of $(\C^2)^{\otimes N}$ is the spin of the $k$-th particle from the left. Two particles that meet at $x_i=x_j$ are neighbours on the line, so $P_A^{i,j}$ and $P_S^{i,j}$ act on two neighbouring spin factors.\footnote{This is essential. If the spins were attached to the particle labels instead, spatial symmetry would allow us to replace $c'P_A^{i,j}+cP_S^{i,j}$ in \eqref{functionalLLH} by its average over all pairs, and the model would reduce to a Lieb--Liniger model with an averaged coupling constant, without any spin chain.} Since $\Psi$ is spatially symmetric, it is determined by its restriction $\psi$ to the sector $\{0<x_1<\dots<x_N<L\}$, and
\[
\frac{1}{N!}\mathcal{E}^{\textnormal{LLH}}(\Psi)=\int_{\{0<x_1<\dots<x_N<L\}}\sum_{i=1}^{N}\abs{\partial_i\psi}^2+\sum_{k=1}^{N-1}\int_{\{x_k=x_{k+1}\}}\braket{\psi\big|\big(c'P_A^{k,k+1}+cP_S^{k,k+1}\big)\psi},
\]
where the integrals over the faces $\{x_k=x_{k+1}\}$ are with respect to the coordinates $x_j$, $j\neq k+1$. In this form, the model describes a contact interaction between neighbouring particles whose strength depends on their spins.

 Note that the energy \eqref{functionalLLH} depends explicitly on spin, and this is the main motivation for studying the upper bound in the context of matrix-valued potentials in Section~\ref{secupperbound}.
 The trial state argument is as sketched in the previous section, but of course we need to patch with a scattering solution that allows for matrix-valued potentials and ensure correct boundary conditions (see Definition~\ref{DefinitionScatteringLengthMatrix} and Remarks~\ref{remarkbcA} and~\ref{finalremark}).
\begin{theorem}[Ground state energy of the LLH model]
\label{TheoremLLHUpperBound}
 Consider a gas of spin-1/2 particles with energy functional \eqref{functionalLLH} with $c,c'>0$ and domain \eqref{domainLLH}, with Dirichlet, Neumann or Robin boundary conditions with non-negative parameter. Let $\rho=N/L$. Then, for $\rho/c$, $\rho/c'$ sufficiently small, the ground state energy $E^{\textnormal{LLH}}(N,L)$ satisfies
\[
	E^{\textnormal{LLH}}(N,L)\mathrel{=} N\frac{\pi^2}{3}\rho^2\left(1+2\rho \epsilon^{\textnormal{LLH}}(c,c')+\mathcal{O}\left((\rho/\min(c,c'))^{6/5}+N^{-1}\right)\right),
\]
where $\epsilon^{\textnormal{LLH}}(c,c')$ is the ground state energy of the spin chain Hamiltonian  
\[
H^{\textnormal{LLH}}=-\frac1N\sum_{i=1}^{N}\left(\frac{2}{c'}P_A^{i,i+1}+\frac{2}{c}P_S^{i,i+1}\right)=-\frac{2}{c'}+\frac1N\left(\frac{2}{c'}-\frac{2}{c}\right)H_{\textnormal{Heis}},
\]
and $H_{\textnormal{Heis}}=\sum_{i=1}^{N}P_S^{i,i+1}$ is the Heisenberg antiferromagnet.
\end{theorem}
The upper bound in Theorem~\ref{TheoremLLHUpperBound} is proved for Dirichlet boundary conditions, by adapting the proof of Theorem~\ref{TheoremUpperBoundSpinJFermi} as explained in Remark~\ref{finalremark}. The lower bound is proved for Neumann boundary conditions in Section~\ref{SecLLHLower}.

    In \cite{girardeau2006ground}, Girardeau used a trial state argument, with a product of a spatial and a spin state, to obtain an upper bound on the Lieb--Liniger--Heisenberg ground state energy in the case $c>c'>0$,
    \begin{equation}
    \label{EqGirardeauUpperBound}
        E^{\textnormal{LLH}}(N,L)\leq N\frac{\pi^2}{3}\rho^2\left(1-\frac{4\rho}{\ln(2)c'+(1-\ln(2))c}+o\left(\frac{\rho}{c'}\right)+\mathcal{O}(1/N)\right),
    \end{equation}
and remarked that, since this product state is exact in several limiting cases, it may be a good approximation throughout the $(c,c')$ plane. 
However, $\epsilon^{\textnormal{LLH}}(c,c')=-2\left(\frac{\ln(2)}{c'}+\frac{1-\ln(2)}{c}\right)+\mathcal{O}(1/N)$ for $c>c'>0$ by the ground state energy of the Heisenberg model  (\eqref{some310} with $J=1/2$), so that it follows from the weighted harmonic-arithmetic mean inequality that Theorem~\ref{TheoremLLHUpperBound} improves \eqref{EqGirardeauUpperBound}. Hence Girardeau's product state is not optimal to first order in the dilute limit, and Theorem~\ref{TheoremLLHUpperBound} shows that $2\rho\epsilon^{\textnormal{LLH}}(c,c')$ is the correct first-order term.

\subsection{Open problems}
\label{openproblems}
Besides extending the lower bound to general matrix-valued potentials, there are several other questions one could study. 

For example, in \cite[Section~1.5]{agerskov2022ground}, we conjectured that the energy expansions for dilute, spinless bosons and fermions \eqref{someeq1276} and \eqref{energ2} depend only on the scattering length to second order. In particular, for dilute bosons, we expect that the ground state energy is
\[
N\frac{\pi^2}{3}\rho^2\left(1+2\rho a_e+3(\rho a_e)^2+\dots\right)
\]
for a wide class of potentials. We can ask if this is also the case for the dilute spin-$1/2$ fermion result \eqref{energy3}. Can we simply follow the bosonic analogue and conjecture that the right expression is
\[
N\frac{\pi^2}{3}\rho^2\Big(1+2\big[\ln(2)\rho a_{e}+(1-\ln(2))\rho a_{o}\big]+3\big[\ln(2)\rho a_{e}+(1-\ln(2))\rho a_{o}\big]^2+\dots\Big)?
\]
In general, does $a_e$ get replaced by $a_e+(a_o-a_e)\epsilon^{\textnormal{LS}}_J$ like in Theorem~\ref{TheoremDiluteFermiGroundStateEnergy},
\[
N\frac{\pi^2}{3}\rho^2\Big(1+2\rho[a_e+(a_o-a_e)\epsilon^{\textnormal{LS}}_J]+3\rho^2[a_e+(a_o-a_e)\epsilon^{\textnormal{LS}}_J]^2+\dots\Big),
\]
and likewise for the matrix-valued case?

Of course, the other open problems for dilute, spinless 1D bosons mentioned in \cite[Section~1.5]{agerskov2022ground}---studying momentum distributions and positive temperature for general potentials---extend to dilute gases with spin as well (see e.g.\ \cite{deuretzbacher2016momentum,yang2017one,kerr2024analytic} in the physics literature). It would be particularly interesting to see to what extent the physics is still described by the Lai--Sutherland model \eqref{spinchainmt}.

\section{Matrix-valued potentials}
\label{matrixpotentials}
In the previous sections, we mostly presented our results for scalar-valued potentials, but the upper bound we prove in the next section also applies to matrix-valued potentials. To state that result, we need a definition of the scattering length that is more general than the approach taken in \eqref{dyson1} and \eqref{dyson2} in the introduction (explained in more detail in \cite[Lemma~2 and Appendix~A]{agerskov2022ground}, also see \cite[Appendix~C]{lieb2006mathematics} for a 3D gas). Background facts needed for the following definition have been collected in Appendix~\ref{AppendixA}.

\begin{definition}[The scattering length matrix for spin-$J$ fermions]
\label{DefinitionScatteringLengthMatrix}
 Consider an even, positive semidefinite measure $V$ that takes values in the space of complex $(2J+1)^2\times (2J+1)^2$ matrices. Assume that $V$ commutes with $P_S$ and $P_A$ (meaning $V$ is block diagonal, leaving the symmetric and antisymmetric spin spaces invariant), and assume that $V$ is supported in $[-R_0,R_0]$ for some $R_0\geq 0$.
 
 Given $R>R_0$, let $F_0\in H^1([-R,R];\operatorname{Mat}((2J+1)^2,(2J+1)^2))$ be the unique solution (in the distributional sense) of the matrix equation 
 \begin{equation}
 \label{EqEvenScatteringEquation}
	-F_0''(x)+\frac{1}{2}V(x)F_0(x)=0,
	\end{equation}
on $[-R,R]$ with boundary conditions $ F_0(R)=I $ and $F_0(-R)=P_A-P_S$ (see Proposition~\ref{CorollaryUniqueF0} in Appendix~\ref{AppendixA}). The matrix-valued function $F_0$ is called the \textbf{scattering solution} of $V$, and it satisfies (see Proposition~\ref{CorollaryUniqueF0} for proofs)
    	\begin{equation}
        \label{EqScatteringSolution}
		F_0(x)=\begin{cases}
(R-\mathsf{A})^{-1}(x-\mathsf{A}) &\ \ \  \text{if }\ \ R_0\leq x\leq R\\
(P_A-P_S)(R-\mathsf{A})^{-1}(-x-\mathsf{A})&\ \ \  \text{if } -R\leq x\leq -R_0
\end{cases}
\end{equation}
    for some matrix $\mathsf{A}$ called the \textbf{scattering length matrix} of $V$, and
	\begin{equation}\label{EqScatteringEnergy}
	\int_{[-R,R]} 2(F_0')^\ast F_0'+F_0^\ast VF_0 =4(R-\mathsf{A})^{-1}.
\end{equation} 
It follows from the above that $\mathsf{A}$ is independent of $R>R_0$, that it is Hermitian and (by letting $R\to R_0$) that $\mathsf{A}\leq R_0$ (since $V$ is positive semidefinite).
\end{definition}

\begin{remark}
\label{remarkbcA}
For spin-$J$ bosons, all we have to change in the definition above is the boundary condition. This becomes $ F_0(R)=I $ and $F_0(-R)=P_S-P_A$, so that there is symmetry upon mapping $R\to -R$ and exchanging the spins, rather than antisymmetry. 

In the Lieb--Liniger--Heisenberg model (Section~\ref{resultsLLH}), the wave function is spatially symmetric, with the spins labelled by position, and the correct boundary condition to define the matrix scattering solution is $ F_0(R)=F_0(-R)=I $. Since the interaction is $V=(2c'P_A+2cP_S)\delta$, this will mean that $F_0$ acts as a multiple of the identity on the spin singlet subspace and on the spin triplet subspace, with the spatial behaviour of the scattering solution of the Lieb--Liniger model (with constant $c'$ for the singlet and $c$ for the triplet).
\end{remark}

For the scalar potentials $V=vI$ that were discussed in the previous sections, it is easily seen that the scattering length matrix $\mathsf{A}$ splits in two diagonal blocks and may be written 
\begin{equation}\label{EqScalarPotentialScatteringLengthMatrix}
    \mathsf{A}=a_e P_A + a_o P_S,
\end{equation}
 where $a_e$ and $a_o$ are the even- and odd-wave scattering lengths discussed in \eqref{dyson1} and \eqref{dyson2} in the introduction (see \cite[Lemma~2 and Appendix~A]{agerskov2022ground} for more properties). 

In the general case, the scattering length matrix $\mathsf{A}$ has two blocks corresponding to the symmetric and antisymmetric spin spaces. These can be diagonalized to give specific scattering lengths with corresponding (anti)symmetric eigenvectors. These might be linear combinations of different total spin. For example, if $J=1$, total spin 0 and 2 are both symmetric, and a linear combination of these might give an eigenvector of $\mathsf{A}$.

\begin{remark}
\label{remarkinfinity}
Generalizing the case $v=0$ discussed in Remark~\ref{remarknopotential}, whenever $F_0$ has a constant eigenvector $\chi(x)=\chi$ (necessarily spin antisymmetric since it is spatially symmetric), we make this fit with \eqref{EqScatteringSolution} using the convention $\mathsf{A}\chi=-\infty \chi$. In that case, $\norm{\mathsf{A}}=\infty$ and the error terms in Theorems~\ref{TheoremDiluteFermiGroundStateEnergy}, \ref{TheoremUpperBoundSpinJFermi}, and~\ref{TheoremLowerBoundSpin} are infinite. 
\end{remark}

 \section{Upper bound for matrix-valued potentials}
 \label{secupperbound}
  To state our upper bound in the most general way, we consider an even (repulsive) potential $V$ that is a positive semidefinite matrix-valued measure acting on $\C^{2J+1}\otimes \C^{2J+1}$. We study the Hamiltonian
  \begin{equation}
  \label{EqHamiltonianupperbound}
      H=-\sum_{i=1}^{N}\partial_i^2+\sum_{i<j}V_{ij},
  \end{equation}
  with corresponding energy functional \begin{equation}
  \label{functionalupperbound}
     \mathcal{E}(\Psi)=\int_{[0,L]^N} \sum_{i=1}^{N}\braket{\partial_i\Psi\vert \partial_i\Psi}+\sum_{i<j}\braket{\Psi \vert V_{ij} \vert \Psi},
 \end{equation}
 where $V_{ij}$ denotes the matrix potential $V(x_i-x_j)$ acting on the spins of particles $i$ and $j$. The associated domain of $\mathcal{E}$ is 
 \begin{equation}
 \label{domainupperbound}
\left\{\Psi\in \wedge^NL^2([0,L];\C^{2J+1}) \big\vert \Psi\in H^1([0,L]^N; (\C^{2J+1})^{\otimes N}) \textnormal{ and } \mathcal{E}(\Psi)<\infty \right\}.
 \end{equation}
We will decompose the potential as $V=vI+\tilde{V}$, where $v$ is a positive scalar potential as in the previous sections. For technical reasons, $\tilde{V}$ is assumed to be of bounded variation, for which we recall the following definition.
 \begin{definition}
 \label{boundedvariation}
	We say that an $M\times M$ Hermitian matrix-valued measure $\tilde{V}$ has bounded variation if $\langle\chi|\tilde{V}|\chi\rangle$ is a finite signed measure for any $\chi\in \C^{M}$. In this case, we define the total variation of $\tilde{V}$ to be the finite\footnote{Let $\{e_\sigma\}_{\sigma=1}^M$ be the standard basis of $\C^M$. Note that $\|\tilde{V}\|(X)$ is finite by writing out $\xi_n$ in this basis
    \[
    \sup_{\substack{\xi_n\\ |\xi_n|=1}}|\langle\xi_n| \tilde{V}(X_n)\xi_n\rangle|\leq \sum_{1\leq\sigma,\sigma'\leq M}|\langle e_\sigma| \tilde{V}(X_n)|e_{\sigma'}\rangle|,
    \]
    and using polarization identities to reduce to the diagonal entries.
   } (positive) measure 
    \[
    \|\tilde{V}\|(X)\coloneqq\sup_{\{X_n\}}\sum_{n}\|\tilde{V}(X_n)\|=\sup_{\{X_n\}}\sum_{n}\sup_{\substack{\xi_n\\ |\xi_n|=1}}|\langle\xi_n| \tilde{V}(X_n)\xi_n\rangle|,
    \] where $X$ is a Borel-measurable set and the first supremum is over partitions of the set $X$ into a countable number of disjoint Borel-measurable subsets $X_n$.
\end{definition}
Recall that we abuse notation and write $V(x)dx$ when integrating with respect to the measure $V$, even if the measure does not allow for this. If we let $X$ be all of $\mathbb{R}$, the above definition says that $\int\|\tilde{V}\|<\infty$, and we will use this constant in some of our estimates. The precise assumptions on $V$ are as follows.

\begin{assumption}\label{AssumptionUpperBound}
$V$ is assumed to be an even, positive semidefinite, $(2J+1)^2\times (2J+1)^2$ matrix-valued measure that commutes with $P_S $ (and hence $P_A$), compactly supported in $[-R_0,R_0]$. It is also assumed that $V=vI+\tilde{V}$, where $v$ is any (positive) Borel measure and $\tilde{V}$ is a Hermitian matrix-valued Borel measure of bounded variation (with $\tilde{V}\geq -vI$), both supported in $[-R_0,R_0]$. Denote by $\mathsf{A}$ the scattering length matrix of $V$ (see Definition~\ref{DefinitionScatteringLengthMatrix}).
\end{assumption}

Our upper bound is given by the following theorem.
\begin{theorem}\label{TheoremUpperBoundSpinJFermi}
Consider a spin-$J$ Fermi gas with Hamiltonian \eqref{EqHamiltonianupperbound} and potential $V$ satisfying Assumption~\ref{AssumptionUpperBound}. Let $\rho=N/L$. Then, for $\rho R_0$ and $\rho\|\mathsf{A}\|$ sufficiently small (where $\|\mathsf{A}\|$ denotes the operator norm of $\mathsf{A}$ and the threshold may depend on $\tilde V$), there exists a constant $C_U>0$ such that the ground state energy with Dirichlet boundary conditions satisfies 
\begin{equation}
\label{bounddirichlet}
E^{\textnormal{Dirichlet}}_{J}(N,L)\leq N\frac{\pi^2}{3}\rho^2\left(1+2\rho \epsilon_{J}(\mathsf{A})+C_U\left((\rho \max(R_0,\norm{\mathsf{A}}))^{6/5}+N^{-1}\right)\right),
\end{equation}
where  $\epsilon_{J}(\mathsf{A})$ is the ground state energy of the spin chain Hamiltonian\footnote{As indicated in Footnote~\ref{lastfootnote}, the boundary conditions of \eqref{lastspinchain} are irrelevant for the statement. We will work with periodic boundary conditions for the spin chain in the proof.}
\begin{equation}
\label{lastspinchain}
\frac{1}{N}\sum_{i=1}^{N}\mathsf{A}^{i,i+1},
\end{equation}
 with $\mathsf{A}^{i,i+1}$ acting as the scattering length matrix $\mathsf{A}$ on spins $i$ and $i+1$ and as the identity on the rest.
\end{theorem}

\begin{remark}
\label{finalremark}
 A key feature of the proof of Theorem~\ref{TheoremUpperBoundSpinJFermi} is that the  trial state is initially defined on the ordered sector $\{0\leq x_1<x_2<\dots<x_N<L\}$, and then extended to other sectors by symmetry. Hence, the proof can easily be adapted to other symmetries---such as spin-$J$ bosons or the set-up studied in the Lieb--Liniger--Heisenberg model (Section~\ref{resultsLLH})---by extending in a different way to other orderings of the particles. For the Lieb--Liniger--Heisenberg model, this extension is spatially symmetric, with the spins labelled by position as in Section~\ref{resultsLLH}, and $V=(2c'P_A+2cP_S)\delta$ satisfies Assumption~\ref{AssumptionUpperBound} with $R_0=0$. The only change needed (to guarantee continuity of the trial state) is the choice of boundary conditions in the Definition~\ref{DefinitionScatteringLengthMatrix} of $\mathsf{A}$---also see Remark~\ref{remarkbcA}. For the Lieb--Liniger--Heisenberg model, the error terms involving $\int\|\tilde{V}\|$ in Lemmas~\ref{extralemma} and~\ref{LemmaSpin1/2EtaContribution1} are not needed, so the smallness condition in Theorem~\ref{TheoremLLHUpperBound} involves only $\rho/c$ and $\rho/c'$, with a universal threshold. Indeed, the term in Lemma~\ref{extralemma} vanishes, since contact interactions between other pairs act on a null subset of $B_{12}$, and in Lemma~\ref{LemmaSpin1/2EtaContribution1}, $W_{\textnormal{hc},12}$ vanishes at contact since $R_0=0$, so the interaction term is bounded by the scattering energy directly.
 \end{remark}

 As a special case, for $V=vI$, we find the following corollary to Theorem~\ref{TheoremUpperBoundSpinJFermi} by \eqref{EqScalarPotentialScatteringLengthMatrix}.
\begin{corollary}[Upper bound in Theorem~\ref{TheoremDiluteFermiGroundStateEnergy}]
\label{TheoremUpperBound2}
 Let $v$ be an even, positive interaction compactly supported in $[-R_0,R_0]$, with $a_e$ and $a_o$ its even- and odd-wave scattering lengths. Consider a spin-$J$ Fermi gas with energy functional \eqref{functional} and domain \eqref{domainfermions} with Dirichlet boundary conditions. Then, for $\rho R_0$, $\rho\abs{a_e}$ and $\rho a_o$ sufficiently small, there exists a constant $C_U>0$ such that the ground state energy $E^{\textnormal{Dirichlet}}_{J}(N,L)$ satisfies
\begin{equation}
\label{upperboundequation}
E^{\textnormal{Dirichlet}}_{J}(N,L)\leq N\frac{\pi^2}{3}\rho^2\left(1+2\rho\left[a_e+(a_o-a_e)\epsilon^{\textnormal{LS}}_J\right]+C_U\big((\rho \max(R_0,a_o,|a_e|))^{6/5}+N^{-1}\big)\right),
\end{equation}
	where $\epsilon^{\textnormal{LS}}_J$ is the minimal energy per site of the spin-$J$ Lai--Sutherland model, i.e.\ the minimal expectation value of
    \[
h_{\textnormal{LS}}=\frac{1}{N}\sum_{i=1}^{N}P_S^{i,i+1}.
\]
\end{corollary}

\subsection{Preparatory facts about the free Fermi ground state}
\label{freefermi}
Let $\Psi_F\in L^2([0,L]^N)$ be the spinless Dirichlet free Fermi ground state. Its $k$-particle reduced density matrix is given by (note the choice of normalization)
\[
\begin{aligned}
	\gamma^{(k)}(x_1,...,x_k;y_1,...,y_k)&=\tfrac{N!}{(N-k)!}\int_{[0,L]^{N-k}}\overline{\Psi_F(x_1,...,x_N)}\\
	&\qquad\times\Psi_F(y_1,...,y_k,x_{k+1},\dots,x_N)d x_{k+1}\ldots d x_N.
\end{aligned}
\]
Similarly, we define the $k$-particle reduced density of $\Psi_F$ by
\[
	\rho^{(k)}(x_1,...,x_k)=\gamma^{(k)}(x_1,...,x_k;x_1,...,x_k).
\]
We shall, given a function $f(x)$, in the following use the notation $f\vert_{x_1}=f(x)\vert_{x_1}=f(x_1)$ and $f\vert_{x_1}^{x_2}=f(x)\vert_{x_1}^{x_2}=f(x_2)-f(x_1)$.

We will need some facts about $\gamma^{(k)}$ and $\rho^{(k)}$ from \cite{agerskov2022ground} (Lemma~\ref{LemmaDensityBounds}), a precise description of $\rho^{(2)}$ near $x_1=x_2$ (Lemma~\ref{Lemma rho2 bound}), and some further bounds of the same kind (Lemma~\ref{LemmaDensityBounds2}). The functions $\gamma^{(k)}$ and $\rho^{(k)}$ are trigonometric polynomials, and we extend them to all of $\R^k\times\R^k$ and $\R^k$, respectively. The main term in the upper bound comes from pairs of particles at short distance, where the pair density $\rho^{(2)}(x_1,x_2)$ vanishes quadratically. The following lemma describes this behaviour with an error that is uniform in the position, and gives the exact average of the quadratic coefficient.
\begin{lemma}[Pair density near contact]\label{Lemma rho2 bound}
Let $\phi_n(x)=\sqrt{2/L}\sin(n\pi x/L)$, $n=1,\dots,N$, be the orbitals of $\Psi_F$, and let
\[
Q(X):=\sum_{1\leq m<n\leq N}\left(\phi_m'(X)\phi_n(X)-\phi_m(X)\phi_n'(X)\right)^2,\qquad X\in\R,
\]
where the prime denotes the derivative. Then $0\leq Q\leq8\pi^2\rho^4$, and for all $X,r\in\R$,
\[
\abs{\rho^{(2)}(X+\tfrac r2,X-\tfrac r2)-Q(X)r^2}\leq\textnormal{const. }\rho^6r^4.
\]
Moreover,
\[
\int_0^LQ(X)dX=\frac{\pi^2}{3}\rho^4L(1+\delta_N),\qquad\abs{\delta_N}\leq\frac1N.
\]
\end{lemma}
\begin{proof}
Write $g_{mn}(x_1,x_2):=\phi_m(x_1)\phi_n(x_2)-\phi_n(x_1)\phi_m(x_2)$, so that $\rho^{(2)}=\sum_{m<n}g_{mn}^2$. Fix $X$ and let $p(r):=\rho^{(2)}(X+\tfrac r2,X-\tfrac r2)$. Since each $g_{mn}$ is antisymmetric, $p$ is even and $p(0)=0$. Moreover, the derivative of $g_{mn}(X+\tfrac r2,X-\tfrac r2)$ at $r=0$ is $\phi_m'(X)\phi_n(X)-\phi_m(X)\phi_n'(X)$, so $p''(0)=2Q(X)$. The fourth derivative of $p$ is bounded by const.\ $\rho^6$ (see the proof of Lemma~\ref{LemmaDensityBounds}), and Taylor's theorem gives the bound on $\abs{p(r)-Q(X)r^2}$. The bound on $Q$ follows from this and $\rho^{(2)}(x_1,x_2)\leq8\pi^2\rho^4(x_1-x_2)^2$ (Lemma~\ref{LemmaDensityBounds}).

For the integral, let $k_n:=n\pi/L$. For $m\neq n$, we expand $(\phi_m'\phi_n-\phi_m\phi_n')^2$ with product-to-sum formulas. All non-constant terms are cosines with full periods on $[0,L]$, so they integrate to zero, and we find $\int_0^L(\phi_m'\phi_n-\phi_m\phi_n')^2=(k_m^2+k_n^2)/L$. Since each $n$ belongs to $N-1$ pairs,
\[
\int_0^LQ=\frac{N-1}{L}\sum_{n=1}^Nk_n^2=\frac{\pi^2(N-1)N(N+1)(2N+1)}{6L^3},
\]
which is the claim with $1+\delta_N=(1-N^{-2})(1+\frac{1}{2N})$.
\end{proof}
\begin{lemma}[from Lemma 10 in \cite{agerskov2022ground}]\label{LemmaDensityBounds}
	\[
		\begin{aligned}
            \rho^{(2)}(x_1,x_2)&\leq 8\pi^2\rho^4(x_1-x_2)^2\\
			\rho^{(3)}(x_1,x_2,x_3)&\leq \textnormal{const. }\rho^7(x_1-x_2)^2(x_2-x_3)^2\\
			\rho^{(4)}(x_1,x_2,x_3,x_4)&\leq \textnormal{const. }\rho^8(x_1-x_2)^2(x_3-x_4)^2\\
			\abs{\sum_{i=1}^{2}\partial_{y_i}^2\gamma^{(2)}(x_1,x_2;y_1,y_2)\rvert_{y=x}}&\leq \textnormal{const. } \rho^{6}(x_1-x_2)^2\\
			\abs{\partial_{y_1}^2\frac{\gamma^{(2)}(x_1,x_2;y_1,y_2)}{y_1-y_2}\Bigg\rvert_{y=x}}&\leq \textnormal{const. } \rho^{6}\abs{x_1-x_2}\\
			\abs{\sum_{i=1}^{3}\partial_{y_i}^2\gamma^{(3)}(x_1,x_2,x_3;y_1,y_2,y_3)\Bigg\vert_{y=x}}&\leq\textnormal{const. }\rho^9(x_1-x_2)^2(x_2-x_3)^2
		\end{aligned}
	\]
\end{lemma}
\begin{proof}
We recall the argument of \cite[Lemma~10]{agerskov2022ground}. By Wick's theorem, $\gamma^{(k)}$ is a determinant of copies of $\gamma^{(1)}$, and all derivatives of $\gamma^{(1)}$ are bounded uniformly by powers of $\rho$. Hence the same holds for $\gamma^{(k)}$, where the power follows from dimensional analysis. The factors $(x_i-x_j)^2$ come from the zeros of $\gamma^{(k)}(x;y)$, which by antisymmetry vanishes when two of the $x_i$, or two of the $y_i$, coincide. If a smooth function $G$ vanishes on $\{u_1=u_2\}$, then
\[
G(u)=(u_1-u_2)\tilde G(u),\qquad \tilde G(u):=\int_0^1\partial_1G(u_2+s(u_1-u_2),u_2,u_3,\dots)\,ds,
\]
where the derivatives of $\tilde G$ are bounded by those of $\partial_1G$. Moreover, $\tilde G$ still vanishes on each hyperplane $\{u_j=u_l\}$ on which $G$ vanishes, provided that $\{j,l\}$ is disjoint from $\{1,2\}$ or equal to $\{2,3\}$. We divide out $x_1-x_2$ and $y_1-y_2$, and then $x_2-x_3$ and $y_2-y_3$ for $\gamma^{(3)}$, or $x_3-x_4$ and $y_3-y_4$ for $\gamma^{(4)}$. The estimates then follow by Taylor expansion and the quotient rule, simplified by the symmetry of the expressions above.
\end{proof}

Similarly, we obtain the following estimates. 
\begin{lemma}[Further density bounds]
\label{LemmaDensityBounds2}
\[
\begin{aligned}
\partial_{y_i}\partial_{x_i}\frac{\gamma^{(4)}(y_1,y_2,y_3,y_4;x_1,x_2,x_3,x_4)}{\abs{x_2-x_1}\abs{y_2-y_1}}\Bigg \lvert_{y=x}&\leq \textnormal{const. } \rho^{8}\\
\abs{\partial_{y_i}\frac{\gamma^{(4)}(y_1,y_2,y_3,y_4;x_1,x_2,x_3,x_4)}{\abs{x_2-x_1}\abs{y_2-y_1}}\Bigg \lvert_{y=x}}&\leq \textnormal{const. } \rho^{8}\abs{x_3-x_4}\\
\abs{\sum_{i=1}^{4}\partial_{y_i}^2\gamma^{(4)}(y_1,y_2,y_3,y_4;x_1,x_2,x_3,x_4) \Bigg \lvert_{y=x}}&\leq \textnormal{const. } \rho^{10} (x_1-x_2)^2(x_3-x_4)^2\\
\partial_{y_i}\partial_{x_i}\frac{\gamma^{(3)}(y_1,y_2,y_3;x_1,x_2,x_3)}{\abs{x_2-x_1}\abs{y_2-y_1}}\Bigg \lvert_{y=x}&\leq \textnormal{const. } \rho^{7}\\
\abs{\partial_{y_i}\frac{\gamma^{(3)}(y_1,y_2,y_3;x_1,x_2,x_3)}{\abs{x_2-x_1}\abs{y_2-y_1}}\Bigg \lvert_{y=x}}&\leq \textnormal{const. } \rho^{7}\abs{x_2-x_3}
\end{aligned}
\]
\end{lemma}
\begin{proof}
As for Lemma~\ref{LemmaDensityBounds}. After dividing out $x_2-x_1$ and $y_2-y_1$, the factor $x_3-x_4$ (for $\gamma^{(4)}$) or $x_2-x_3$ (for $\gamma^{(3)}$) remains in the bounds with a single derivative, since these only differentiate in $y$. The third bound follows from symmetry, as the bounds on $\sum_i\partial_{y_i}^2\gamma^{(k)}$ in Lemma~\ref{LemmaDensityBounds}, and the powers of $\rho$ follow from dimensional analysis.
\end{proof}

\subsection{Constructing a trial state}
\label{sectrialstate}
\subsubsection{Notation}
\label{notation}
We will use $|.|$ for absolute values and the vector 2-norm in $\otimes^N\mathbb{C}^{2J+1}$, and $\langle.|.\rangle$ for its inner product. Let $\|\mathsf{A}\|$ denote the corresponding matrix norm of the scattering length matrix $\mathsf{A}$, which is equal to the spectral radius since $\mathsf{A}$ is Hermitian. Definition~\ref{boundedvariation} and the paragraph below it discuss the meaning of $\int\|\tilde{V}\|$.

As explained in the introduction and Section~\ref{proofideas}, the ground state energy of a spinless gas in the dilute limit can be approximated by making use of the free Fermi gas when particles are far apart, adjusted with the 2-body scattering solution when they are close \cite{agerskov2022ground,agerskov2023one}. We denote by $b>R_0$ the range within which we will alter the free Fermi state to account for scattering effects. This will be chosen appropriately in the proof of Theorem~\ref{TheoremUpperBoundSpinJFermi} in Section~\ref{proofupperbound}. Let
\[
\begin{aligned}
\mathcal{R}(x)&:=\min_{i<j}\abs{x_i-x_j}\\
\mathcal{R}_2(x)&:=\min_{\{i,j\}\neq \{k,l\}}\max(\abs{x_i-x_j},\abs{x_k-x_l})
\end{aligned}
\]
denote the closest and second closest distance between two particles, respectively. 

For particles $i$ and $j$, we define\footnote{Note the absolute values: this is convenient for our method, and it makes no difference since we will define the trial state on $0\leq x_1<\dots<x_N<L$ only and extend antisymmetrically. We will often use the scattering energy \eqref{EqScatteringEnergy}, but this is unaffected by the sign differences created by the absolute values. 
}
\begin{equation}
\label{Wij}
W_{ij}(x)\coloneqq bF_{ij}(\abs{x_i-x_j}), 
\end{equation}
where $F_{ij}$ is the matrix scattering solution from Definition~\ref{DefinitionScatteringLengthMatrix} with $R$ replaced by $b>R_0$, acting on the spins of particles $i$ and $j$. We would also like to have this definition for the closest pair: say particles $i$ and $j$ are the closest pair in a given configuration $x$, i.e.\ $\mathcal{R}(x)=\abs{x_i-x_j}$, define 
\begin{equation}
\label{WR}
W^\mathcal{R}(x)\coloneqq W_{ij}(\abs{x_i-x_j})=bF_{ij}(\abs{x_i-x_j}).
\end{equation}
Similarly, let
 \begin{equation} 
\label{hcij}
 W_{\text{hc},ij}(x):=\begin{cases} 
 \ 0&\textnormal{   if }|x_j-x_i|<R_0\\
 \ b\frac{|x_j-x_i|-R_0}{b-R_0}\ I&\textnormal{   if }b\geq|x_j-x_i|\geq R_0
 \end{cases},
 \end{equation}
 which is $b$ times the scattering solution from Definition~\ref{DefinitionScatteringLengthMatrix} for the hard-core potential with radius $R_0$ (with $R$ replaced by $b$). We again extend this definition to the closest pair,\footnote{For our proof to work, we could also choose the hard-core radius to be the spectral radius of the positive part of $\mathsf{A}$, $\|\mathsf{A}_{+}\|$. This possibility is important to apply the proof to non-compactly supported potentials.}
 \begin{equation} 
\label{hc}
 W_{\text{hc}}^\mathcal{R}(x):=\begin{cases} 
 \ 0&\textnormal{   if }\mathcal{R}(x)<R_0\\
 \ b\frac{\mathcal{R}(x)-R_0}{b-R_0}\ I&\textnormal{   if }b\geq\mathcal{R}(x)\geq R_0
 \end{cases}.
 \end{equation}
 Note that there is an important difference between \eqref{WR} and \eqref{hc}: the latter matrix acts as the identity on the spins, which means that \eqref{hc} is a continuous function in $x$ (given that $\mathcal{R}(x)$ is). However, \eqref{WR} acts non-trivially on the closest pair of spins, and so it may not be continuous whenever $\mathcal{R}(x)=\mathcal{R}_2(x)$, because which pair of particles is the closest may change at this point and the action of the scattering solution matrix will jump with it. 
 
 All this means that our preferred strategy---patching with the scattering solution \eqref{WR} for the closest pair, which was done for spinless bosons in \cite{agerskov2022ground}---leaves discontinuities. To solve this technical challenge, we introduce the continuous function \eqref{hc}, and the continuous and almost everywhere differentiable interpolation function 
\begin{equation}
\label{eta}
\eta(x)\coloneqq\begin{cases}
0&\text{ if } \mathcal{R}_2(x)\leq b\\
\frac{\mathcal{R}_2(x)}{b}-1 &\text{ if } b<\mathcal{R}_2(x)<2b\\
1 &\text{ if } \mathcal{R}_2(x)\geq 2b
\end{cases}.
\end{equation}
We use $\eta$ to define a continuous trial state \eqref{EqTrialStateSpin1/2Fermi} below, interpolating between the preferred, but discontinuous \eqref{WR} and the continuous \eqref{hc} at a negligible length scale (i.e.\ at the cost of an acceptable error). 

\subsubsection{The trial state} 
Let $\Psi_F$ be the spinless free Fermi state with Dirichlet boundary conditions (depending only on spatial coordinates). Using the notation introduced in the previous subsection, we define the (unnormalized) trial state to be as follows on the set $\{0\leq x_1<x_2<\ldots<x_N<L\}$,
\begin{equation}
\label{EqTrialStateSpin1/2Fermi}
\Psi_\chi(x):=\begin{cases}
\frac{\Psi_F(x)}{\mathcal{R}(x)}\left(\eta(x) W^{\mathcal{R}}(x)+(1-\eta(x))W_{\text{hc}}^\mathcal{R}(x)\right)\chi&\textnormal{   if }\mathcal{R}(x)<b\\
\Psi_F(x)\chi &\textnormal{   if }\mathcal{R}(x)\geq b
\end{cases},
\end{equation}
 where $\chi\in\bigotimes^N\C^{2J+1}$ is a fixed spin state to be chosen below in Section~\ref{proofupperbound}.  Note that $W^{\mathcal{R}}$ and $W_{\text{hc}}^\mathcal{R}$ are matrices acting on $\chi$. The extension of $\Psi_\chi$ to other orderings of the coordinates $x_i$ is defined by antisymmetry in the combined space-spin variables, and we denote the extension by $\Psi_\chi$ as well. We will assume that $ \chi $ is translation symmetric (with periodic boundary conditions), meaning that it is an eigenvector of the spin translation operator.

This trial state is in the relevant domain \eqref{domainupperbound}: note that $ P_S^{i,j}\Psi_\chi\rvert_{x_i=x_j}=0 $ due to the boundary condition at $0$ satisfied by $W^{\mathcal{R}}$ and $W_{\text{hc}}^\mathcal{R}$, so that the antisymmetric extension of the part with symmetric spin is not discontinuous at $x_i=x_j$; also note that the discontinuity in $W^{\mathcal{R}}$ at $\mathcal{R}(x)=\mathcal{R}_2(x)$ discussed in the previous subsection is avoided because of the interpolation with $\eta$ and the continuity of $W_{\text{hc}}^\mathcal{R}$.

As was the case in \cite{agerskov2022ground,agerskov2023one}, the trial state given in \eqref{EqTrialStateSpin1/2Fermi} produces an error that grows with the particle number. This is not sufficient to prove Theorem~\ref{TheoremUpperBoundSpinJFermi}. However, we can construct a full trial state on the set $\{0\leq x_1<x_2<\ldots<x_N<L\}$ by taking a product of wave functions of the form \eqref{EqTrialStateSpin1/2Fermi} defined on small, non-overlapping intervals (leaving some distance to avoid interactions between the intervals), and extending to $[0,L]^N$ by antisymmetry. The details are worked out in Section~\ref{proofupperbound}.

\subsection{Proof of Theorem~\ref{TheoremUpperBoundSpinJFermi}}
\subsubsection{Ingredients in the estimates}
Given the fact that the definition \eqref{EqTrialStateSpin1/2Fermi} of $\Psi_\chi$ involves $\min_{i<j}\abs{x_i-x_j}$, it is useful to define 
\[
\begin{aligned}
B&:=\{x\in [0,L]^N\ \vert\ \mathcal{R}(x)<b\}\\
B_{12}&:=\{x\in [0,L]^N\ \vert\ \mathcal{R}(x)=\abs{x_1-x_2}<b\},
\end{aligned}
\]
where $B_{12}$ will be convenient to use because of antisymmetry of $\Psi_\chi$.
The following domains of integration (all subsets of $[0,L]^N$) will be used in the proof.
\begin{equation}
\label{intdomains}
\begin{aligned}
B^{\geq}_{12}&:=B_{12}\cap \{\mathcal{R}_2(x)\geq 2b\}\\
B^{23}_{12}&:=B_{12}\cap\{\mathcal{R}_2(x)=\abs{x_2-x_3}<2b\}\\
B^{34}_{12}&:=B_{12}\cap\{\mathcal{R}_2(x)=\abs{x_3-x_4}<2b\}\\
A_{ij}&:=\left\{\abs{x_i-x_j}<b\right\}\\
A_{12}^{23}&:=A_{12}\cap \{\abs{x_2-x_3}<2b\}\\
A_{12}^{34}&:=A_{12}\cap \{\abs{x_3-x_4}<2b\}
\end{aligned}
\end{equation} 
The distinction between the $B$'s is made because of the interpolation with $\eta$ in \eqref{EqTrialStateSpin1/2Fermi}. The $A$'s are convenient for calculations because they relax the requirements about being the (next) closest pair of particles. The energy density is positive so we often extend e.g.\ $B_{12}\subset A_{12}$ at the cost of an acceptable error.

Using the definition \eqref{Wij} of $W_{12}$, we also define
\begin{equation}
\label{psiij}
(\Psi)_{12}(x):=\frac{\Psi_F(x)}{\abs{x_2-x_1}}W_{12}(x_2-x_1)=\frac{\Psi_F(x)}{\abs{x_2-x_1}}b F_{12}(|x_2-x_1|),
\end{equation}
where $(\Psi)_{12}$ is now a matrix-valued function, related to $F_{12}$ (the matrix scattering solution of the potential acting on the spins of particles $1$ and $2$). At this point, it may be good to contrast this with the spinless free Fermi ground state $\Psi_F$, which is scalar valued, and our trial state $\Psi_\chi$, which is vector valued. 

Given that $\Psi_\chi$ was only made explicit on the set  $\{0\leq x_1<x_2<\ldots<x_N<L\}$ in \eqref{EqTrialStateSpin1/2Fermi}, we introduce notation for other `sectors'---sets that order the coordinates in a different way, 
\[
\{\sigma\}\coloneqq\{0\leq x_{\sigma_1}<x_{\sigma_2}<\ldots<x_{\sigma_N}<L\},
\]
with $\sigma=(\sigma_1,\dots,\sigma_N)$ a permutation that sends $(1,...,N)\mapsto (\sigma_1,...,\sigma_N)$. It will be useful to find an explicit expression for $\Psi_\chi$ for other particle orders as well. We work out the case of $B_{12}^{\geq}\cap\{\sigma\}$ here, and refer to these steps in other cases later. 

Note that by the definition  \eqref{EqTrialStateSpin1/2Fermi} of $\Psi_\chi$ and \eqref{WR} and \eqref{eta}, we have that $\Psi_\chi=(\Psi)_{12}\chi$ on the set $B_{12}^{\geq}\cap\{0\leq x_1<x_2<\ldots<x_N<L\}$ (this is the reason for using the symbol $(\Psi)_{12}$). Can we find a similar expression when the positions are ordered in a different way? It turns out that, on $B_{12}^{\geq}\cap\{\sigma\}$ (assuming $\sigma$ is such that the intersection is not empty\footnote{Note that since particles 1 and 2 are closest on $B_{12}^{\geq}$, they must be next to each other on the line, and only permutations that respect this will give a non-zero intersection $B_{12}^{\geq}\cap \{\sigma\}$.}), we have that 
\begin{equation}
\label{psionsectors}
    \Psi_\chi(x)=(\Psi)_{12}\chi_\sigma(x)=\frac{\Psi_F(x)}{\abs{x_2-x_1}}W_{12}(x_2-x_1)\chi_\sigma,
\end{equation}
where $\chi_\sigma$ denotes the spin state $\chi$ with spins permuted by $\sigma$. To understand this, it is useful to consider an example for $N=4$. On $B_{12}^{\geq}\cap\{0\leq x_3\leq x_4\leq x_2\leq x_1<L\}$, the trial state is defined by ordering the spins and positions such that the positions are ascending (up to a sign that is $-$ in this case). Denoting by $\braket{s_1,s_2,s_3,s_4\vert \cdot}$ the spin coordinate function $\braket {s_1\otimes s_2\otimes s_3\otimes s_4 \vert \cdot}$ (where $s_1,...,s_4\in \C^{2J+1}$ are any set of basis states of the spin space), we find that, given that the particles with coordinates $x_1$ and $x_2$ form the closest pair (which, for the purpose of the second equality and using the definition  \eqref{EqTrialStateSpin1/2Fermi} of $\Psi_\chi$, are the third and fourth spins),
\begin{equation}
\label{examplepsi}
\begin{aligned}
\braket{s_1,s_2,s_3,s_4\vert\Psi_\chi(x_1,x_2,x_3,x_4)}&=-\braket{s_3,s_4,s_2,s_1\vert \Psi_\chi(x_3,x_4,x_2,x_1)}\\
&=-\frac{\Psi_F(x_3,x_4,x_2,x_1)}{\abs{x_1-x_2}}\braket{s_3,s_4,s_2,s_1\vert W_{34}(x_2-x_1)\chi}\\
&=\frac{\Psi_F(x_1,x_2,x_3,x_4)}{\abs{x_1-x_2}}\braket{s_3,s_4,s_2,s_1\vert W_{34}(x_2-x_1)\chi}\\
&=\frac{\Psi_F(x_1,x_2,x_3,x_4)}{\abs{x_1-x_2}}\braket{s_1,s_2,s_3,s_4\vert W_{12}(x_2-x_1)\chi_{\sigma=(3,4,2,1)}}.
\end{aligned}
\end{equation}
A similar, but general argument gives \eqref{psionsectors}. From the third line, we can obtain an additional fact about inner products: on $B_{12}^{\geq}$, the particles with coordinates $x_1$ and $x_2$ are always next to each other because they form the closest pair. This means that in an inner product in $(\mathbb{C}^{2J+1})^{\otimes N}$, we can translate the corresponding spins (3 and 4 in the example) back to 1 and 2 (without doing anything with the spatial coordinates). Of course, we also need to translate the spins in $\chi$, but $\chi$ was assumed to be translation symmetric. Moreover, the order of the two spins of the pair does not matter, since $W_{12}$ commutes with their exchange $P_S^{1,2}-P_A^{1,2}$. Thus, for example, a general version of this argument gives for all $x\in B^{\geq}_{12}\cap\{\sigma\}$ and any $1\leq i\leq N$,
\begin{equation}
\label{translation}
        \abs{\partial_i\Psi_{\chi}(x)}^2=\braket{\partial_i (\Psi)_{12}(x)\chi_{\sigma} \vert \partial_i (\Psi)_{12}(x)\chi_\sigma}=\braket{\partial_i (\Psi)_{12}(x)\chi \vert \partial_i (\Psi)_{12}(x)\chi}.
\end{equation}
Since the latter is independent of $\sigma$, the equality between the first and last expression actually holds on all of $B^{\geq}_{12}$. We will apply this idea to various inner products in the proof, such as the expression above. A similar argument can for example be used to show that for all $x\in B^{\geq}_{12}$,
\begin{equation}
\label{potenergyeq}
\braket{\Psi_\chi(x)|V_{12}(x)\Psi_\chi(x)}=\braket{(\Psi)_{12}(x)\chi|V_{12}(x)(\Psi)_{12}(x)\chi}.
\end{equation}

\subsubsection{Splitting the energy into parts and estimating them}
We can now start estimating. The quantity of interest is the energy \eqref{functionalupperbound} of the trial state $\Psi_\chi$ (we will later account for the norm of $\Psi_\chi$),
\[
\mathcal{E}(\Psi_\chi)=\int_{[0,L]^N}\sum_{i=1}^{N} \abs{\partial_i\Psi_\chi}^2+\sum_{1\leq i<j\leq N}\braket{\Psi_\chi \vert V_{ij}  \Psi_\chi}.
\]
Let $E_F=N\frac{\pi^2}{3}\rho^2(1+\mathcal{O}(1/N))$ denote the ground state energy of the spinless Dirichlet free Fermi
gas. By adding and subtracting $\int_B\sum_i\abs{\partial_i\Psi_F}^2$ and using antisymmetry in the second step, we find that
\begin{equation}\label{EqEnergy1Spin1/2}
\begin{aligned}
\mathcal{E}(\Psi_\chi)&= E_F + \int_{B}\sum_{i=1}^{N}\abs{\partial_i\Psi_{\chi}}^2+\sum_{1\leq i<j\leq N}\braket{\Psi_{\chi}\vert V_{ij}  \Psi_{\chi}}-\sum_{i=1}^{N}\abs{\partial_i \Psi_F}^2\\
&=E_F+\tbinom{N}{2}\int_{B_{12}}\sum_{i=1}^{N}\abs{\partial_i\Psi_{\chi}}^2+\sum_{1\leq i<j\leq N}\braket{\Psi_{\chi}\vert V_{ij} \Psi_{\chi}}-\sum_{i=1}^{N}\abs{\partial_i \Psi_F}^2,
\end{aligned}
\end{equation}
where we used that there are no interactions outside of $B$ because $b$ is bigger than the range of the potential $R_0$ by assumption, as well as the fact that $|\partial_i\Psi_\chi(x)|^2=|\partial_i\Psi_F(x)|^2$ outside of $B$ by the definition \eqref{EqTrialStateSpin1/2Fermi}.

We wish to integrate the last term in \eqref{EqEnergy1Spin1/2} over $A_{12}$ rather than $B_{12}$: note that $x\in A_{12}\setminus B_{12}$ implies $x\in A_{ij}$ for some $(i,j)\neq (1,2)$, so we can estimate
\[
\begin{aligned}
\int_{A_{12}\setminus B_{12}}\sum_{i=1}^{N}\abs{\partial_i\Psi_F}^2&\leq \tbinom{N-2}{2}\int_{(A_{12}\setminus B_{12})\cap A_{34}}\sum_{i=1}^{N}\abs{\partial_i\Psi_F}^2+2(N-2)\int_{(A_{12}\setminus B_{12})\cap A_{13}}\sum_{i=1}^{N}\abs{\partial_i\Psi_F}^2\\
&\leq \tbinom{N-2}{2}\int_{A_{12}\cap A_{34}}\sum_{i=1}^{N}\abs{\partial_i\Psi_F}^2+2N\int_{A_{12}\cap A_{13}}\sum_{i=1}^{N}\abs{\partial_i\Psi_F}^2.
\end{aligned}
\]
We use this estimate on \eqref{EqEnergy1Spin1/2} to rewrite the third term as integration over $A_{12}$. In the other two terms, we write the terms involving $\partial_1$, $\partial_2$ and $V_{12}$ separately and split some of the integrals, with equality, using the integration domains \eqref{intdomains} and antisymmetry to collect similar terms. This gives
\begin{equation}
\label{EqTrialStateEnergyUpperBound1}
\begin{aligned}
\mathcal{E}(\Psi_{\chi})\leq E_F+\tbinom{N}{2}\Bigg[&\int_{B^{\geq}_{12}}2\abs{\partial_1\Psi_{\chi}}^2+\braket{\Psi_\chi|V_{12}\Psi_\chi}-\int_{A_{12}}\sum_{i=1}^{N}\abs{\partial_i \Psi_F}^2\\
+&\int_{B^{\geq}_{12}}\sum^N_{i=3}\abs{\partial_i\Psi_{\chi}}^2+\int_{B_{12}}\sum_{\substack{1\leq i<j\leq N\\(i,j)\neq (1,2)}}\braket{\Psi_\chi|V_{ij}\Psi_\chi}\\
+&\tbinom{N-2}{2}\int_{B_{12}^{34}}\Big(\sum^N_{i=1}\abs{\partial_i\Psi_{\chi}}^2\Big)+\braket{\Psi_\chi|V_{12}\Psi_\chi}\\
+&2(N-2)\int_{B_{12}^{23}}\Big(\sum^N_{i=1}\abs{\partial_i\Psi_{\chi}}^2\Big)+\braket{\Psi_\chi|V_{12}\Psi_\chi}\\
+&\tbinom{N-2}{2}\int_{A_{12}\cap A_{34}}\sum_{i=1}^{N}\abs{\partial_i\Psi_F}^2+2N\int_{A_{12}\cap A_{13}}\sum_{i=1}^{N}\abs{\partial_i\Psi_F}^2\Bigg].
\end{aligned}
\end{equation} 
The following lemmas show that the last four lines all end up in the error term. The first three are proved in Section~\ref{lemmas}, the last (concerning the free Fermi gas) was proved in \cite[Lemma~12]{agerskov2022ground}. 
\begin{lemma}
\label{extralemma0}
\[
\tbinom{N}{2}\int_{B^\geq_{12}}\sum^N_{i=3}\abs{\partial_i\Psi_{\chi}}^2\leq\textnormal{const. } E_F N(\rho b)^3
\]
\end{lemma}
\begin{lemma}
\label{extralemma}
Assuming $ b>2\max(\norm{\mathsf{A}},R_0)$, 
\[
\tbinom{N}{2}\int_{B_{12}}\sum_{\substack{1\leq i<j\leq N\\(i,j)\neq (1,2)}}\braket{\Psi_\chi|V_{ij}\Psi_\chi}\leq \textnormal{const. }E_FN\left(\rho b\right)^3\Big[1+\rho b^2\int \|\tilde{V}\|\Big]
\]
\end{lemma}
\begin{lemma}
\label{LemmaSpin1/2EtaContribution1}
Assuming $N(\rho b)^3<1$ and $b>2\max(\norm{\mathsf{A}},R_0)$,
\[
\begin{aligned}
\tbinom{N}{2}\tbinom{N-2}{2}\int_{B_{12}^{34}}\Big(\sum^N_{i=1}\abs{\partial_i\Psi_{\chi}}^2\Big)+\braket{\Psi_\chi|V_{12}\Psi_\chi}&\leq\textnormal{const. }E_FN\left(\rho b\right)^3\Big[1+\rho b^2\int \|\tilde{V}\|\Big]\\
\tbinom{N}{2}2(N-2)\int_{B_{12}^{23}}\Big(\sum^N_{i=1}\abs{\partial_i\Psi_{\chi}}^2\Big)+\braket{\Psi_\chi|V_{12}\Psi_\chi}&\leq\textnormal{const. }E_F\left(\rho b\right)^3\Big[1+\rho b^2\int \|\tilde{V}\|\Big]
\end{aligned}
\]
\end{lemma}
\begin{lemma}[Lemma 12 in \cite{agerskov2022ground}]
\label{LemmaE2BoundSpin1/2}
Assuming $N(\rho b)^3<1$,
\[
\tbinom{N}{2}\Bigg[\tbinom{N-2}{2}\int_{A_{12}\cap A_{34}}\sum_{i=1}^{N}\abs{\partial_i\Psi_F}^2+2N\int_{A_{12}\cap A_{13}}\sum_{i=1}^{N}\abs{\partial_i\Psi_F}^2\Bigg]\leq \textnormal{const. }E_FN (\rho b)^3
\]
\end{lemma}
For the integral over $B^\geq_{12}$ in the first line of \eqref{EqTrialStateEnergyUpperBound1}, we use \eqref{translation} and \eqref{potenergyeq}, and extend to $A_{12}\supset B^\geq_{12}$ by positivity. The four lemmas above then say that
\begin{equation}
\label{EqTrialStateEnergyUpperBound2}
\begin{aligned}
\mathcal{E}(\Psi_{\chi})\leq E_F&+\tbinom{N}{2}\Bigg[\int_{A_{12}}2\abs{\partial_1(\Psi)_{12}\chi}^2+\braket{(\Psi)_{12}\chi|V_{12}(\Psi)_{12}\chi}-\int_{A_{12}}\sum_{i=1}^{N}\abs{\partial_i \Psi_F}^2\Bigg]\\
&+\textnormal{const. }E_FN\left(\rho b\right)^3\Big[1+\rho b^2\int \|\tilde{V}\|\Big].
\end{aligned}
\end{equation}
Besides the proof of Lemmas~\ref{extralemma0}, \ref{extralemma} and~\ref{LemmaSpin1/2EtaContribution1} (postponed to Section~\ref{lemmas}) there are now two steps left in the proof of Theorem~\ref{TheoremUpperBoundSpinJFermi}: 1.\ Estimating the remaining terms to find the spin chain.  We will do this now in Lemma~\ref{LemmaE1BoundSpin1/2}, and it gives the upper bound on $\mathcal{E}(\Psi_\chi)$ formulated in Proposition~\ref{LemmaUpperBoundFewParticles}.
2.\ Formulating a full trial state built from copies of $\Psi_\chi$ and finishing the proof in Section~\ref{proofupperbound}.
\begin{lemma}
\label{LemmaE1BoundSpin1/2}
Assuming $\rho b<1$ and $ b>2\max(\norm{\mathsf{A}},R_0) $,
\begin{equation}
\label{042094}
\begin{aligned}
\tbinom{N}{2}&\int_{A_{12}}2\abs{\partial_1(\Psi)_{12}\chi}^2+\braket{(\Psi)_{12}\chi|V_{12}(\Psi)_{12}\chi}-\sum_{i=1}^{N}\abs{\partial_i \Psi_F}^2\\
&\leq E_F\left\langle\chi\left\lvert2\rho\frac{1}{N}\sum_{i=1}^{N} {\mathsf{A}}^{i,i+1}\left[1+\frac{\mathsf{A}^{i,i+1}}{bI-\mathsf{A}^{i,i+1}}\right]\right\rvert\chi\right\rangle\\&\quad+\textnormal{const. }E_F\left(N(\rho b)^3+\big(\rho \norm{\mathsf{A}}+(\rho b)^2\big)N^{-1}\right),
\end{aligned}
\end{equation}
with $\mathsf{A}^{i,i+1}$ acting as the scattering length matrix $\mathsf{A}$ from Definition~\ref{DefinitionScatteringLengthMatrix} on spins $i$ and $i+1$ and as the identity on the rest, and $\mathsf{A}^{N,N+1}=\mathsf{A}^{N,1}$.
\end{lemma}
\begin{proof}
 Inserting the expression for $(\Psi)_{12}$ from \eqref{psiij}, we find that 
 \begin{equation}\label{EqE1(1)}
\begin{aligned}
2\tbinom{N}{2}\int_{A_{12}}&\abs{\partial_1(\Psi)_{12}\chi}^2\\
&=2\binom{N}{2}\int_{A_{12}}\Bigg[\braket{\partial_{1}\left(W_{12}(x_1-x_2)\right)\frac{\Psi_F(x)}{x_1-x_2}\chi\Bigg\vert \partial_{1}\left(W_{12}(x_1-x_2)\right)\frac{\Psi_F(x)}{x_1-x_2}\chi}\\
&\qquad\qquad +2\operatorname{Re}\braket{W_{12}(x_1-x_2)\partial_{1}\left(\frac{\Psi_F(x)}{x_1-x_2}\right)\chi\Bigg\vert\partial_1(W_{12}(x_1-x_2))\frac{\Psi_F(x)}{x_1-x_2}\chi}\\
&\qquad\qquad +\braket{W_{12}(x_1-x_2)\partial_{1}\left(\frac{\Psi_F(x)}{x_1-x_2}\right)\chi\Bigg\vert W_{12}(x_1-x_2)\partial_{1}\left(\frac{\Psi_F(x)}{x_1-x_2}\right)\chi}\Bigg].
		\end{aligned}
	\end{equation} 
	Applying partial integration to the second $\partial_1$ in the following integral, using notation introduced above Lemma~\ref{LemmaDensityBounds}, gives
	\begin{equation}
    \label{someeq12}
		\begin{aligned}
			\int_{A_{12}}&\braket{W_{12}(x_1-x_2)\partial_{1}\left(\frac{\Psi_F(x)}{x_1-x_2}\right)\chi\Bigg\vert\partial_1(W_{12}(x_1-x_2))\frac{\Psi_F(x)}{x_1-x_2}\chi}\\
            =
        &\int_{A_{12}}\Bigg[-\braket{\partial_1(W_{12}(x_1-x_2))\partial_{1}\left(\frac{\Psi_F(x)}{x_1-x_2}\right)\chi\Bigg\vert W_{12}(x_1-x_2)\frac{\Psi_F(x)}{x_1-x_2}\chi}\\
            &\qquad\ \ \ -\braket{W_{12}(x_1-x_2)\partial_{1}\left(\frac{\Psi_F(x)}{x_1-x_2}\right)\chi\Bigg\vert W_{12}(x_1-x_2)\partial_1\left(\frac{\Psi_F(x)}{x_1-x_2}\right)\chi} \\
            &\qquad\ \ \ -\braket{W_{12}(x_1-x_2)\partial_{1}^2\left(\frac{\Psi_F(x)}{x_1-x_2}\right)\chi\Bigg\vert W_{12}(x_1-x_2)\frac{\Psi_F(x)}{x_1-x_2}\chi}\Bigg]\\
			&+\int \left[\braket{W_{12}(x_1-x_2)\partial_{1}\left(\frac{\Psi_F(x)}{x_1-x_2}\right)\chi\Bigg\vert W_{12}(x_1-x_2)\frac{\Psi_F(x)}{x_1-x_2}\chi}\right]_{x_1=x_2-b}^{x_2+b} d x_2\dots dx_N.
		\end{aligned}
	\end{equation}
	 Here, $x_1$ ranges over $(x_2-b,x_2+b)\cap(0,L)$. The boundary term at $x_1=x_2+b$ is therefore only present if $x_2+b<L$, and the one at $x_1=x_2-b$ only if $x_2-b>0$. There are no boundary terms at $x_1=0$ or $x_1=L$, because $\Psi_F$ vanishes there. We use the same convention for the boundary terms below.\footnote{ The proof of \cite[Lemma~11]{agerskov2022ground}, on which this proof is modelled, needs the same convention. With it, the bound in \cite[Lemma~11]{agerskov2022ground} holds with an additional error of relative order $(\rho b)^2N^{-1}$, which is absorbed in the error term of \cite[Proposition~6]{agerskov2022ground}.} Since $\Psi_F:[0,L]^N\to \R$ is real valued, this equation allows us to treat the second term in \eqref{EqE1(1)} and cancel the third. If we also apply the product rule to the final term in \eqref{someeq12}, integrate out the variables $x_3,\dots,x_N$ and use that $\abs{W_{12}(\pm b)\chi}^2=b^2\abs{F_{12}(b)\chi}^2=b^2$ by the boundary condition of $F_{12}$, we find that \eqref{EqE1(1)} equals
	\[
		\begin{aligned}
			2\tbinom{N}{2}\int_{A_{12}}&\abs{\partial_1(\Psi)_{12}\chi}^2=\int\left[\partial_{1}\left(\gamma^{(2)}(x_1,x_2;y,x_2)\right)\bigg\vert_{y=x_1}-\frac{\rho^{(2)}(x_1,x_2)}{(x_1-x_2)}\right]_{x_1=x_2-b}^{x_2+b} dx_2\\
			&+2\binom{N}{2}\int_{A_{12}}\Bigg[\braket{\partial_{1}(W_{12}(x_1-x_2))\frac{\Psi_F(x)}{x_1-x_2}\chi\Bigg\vert\partial_1(W_{12}(x_1-x_2))\frac{\Psi_F(x)}{x_1-x_2}\chi}\\
            &\qquad\qquad\qquad-\braket{W_{12}(x_1-x_2)\partial_{1}^2\left(\frac{\Psi_F(x)}{x_1-x_2}\right)\chi\Bigg\vert W_{12}(x_1-x_2)\frac{\Psi_F(x)}{x_1-x_2}\chi}\Bigg].
            \end{aligned}
	\] 
    Also, by antisymmetry of $\Psi_F$ and partial integration (notice only $x_1$ gives a boundary term), and using that $\Psi_F$ is an eigenfunction of $-\sum_{i=1}^{N}\partial_i^2$ with eigenvalue $E_F$,
    \[
		\begin{aligned}
			-\tbinom{N}{2}\int_{A_{12}}\sum_{i=1}^{N}\abs{\partial_i\Psi_F}^2&=-\tbinom{N}{2}\int_{A_{12}}\left(2\abs{\partial_1\Psi_F}^2+\sum_{i=3}^{N}\abs{\partial_i\Psi_F}^2\right)\\&=-\tbinom{N}{2}\int_{A_{12}}\sum_{i=1}^{N}\overline{\Psi_F}(-\partial^2_i\Psi_F)-2\binom{N}{2}\int\left[\overline{\Psi_F}\partial_1\Psi_F\right]_{x_1=x_2-b}^{x_1=x_2+b}\\
			&=-E_F\tbinom{N}{2}\int_{A_{12}}\abs{\Psi_F}^2-\int\left[\partial_1\gamma^{(2)}(x_1,x_2;y,x_2)\vert_{y=x_1}\right]_{x_2-b}^{x_2+b} d x_2.
		\end{aligned}
	\]
    The first term is negative and can be thrown out for an upper bound, while the second cancels a similar term in the previous equation. Putting these equations together and integrating out the coordinates $x_3,...,x_N$ for the first inequality, we find that the quantity of interest satisfies
\[
\begin{aligned}
&\tbinom{N}{2}\int_{A_{12}}2\abs{\partial_1(\Psi)_{12}\chi}^2+\braket{(\Psi)_{12}\chi|V_{12}(\Psi)_{12}\chi}-\sum_{i=1}^{N}\abs{\partial_i \Psi_F}^2\\
&\mathrel{\leq}\int\left[-\frac{\rho^{(2)}(x_1,x_2)}{(x_1-x_2)}\right]_{x_1=x_2-b}^{x_2+b} dx_2\\ 
&\ +\int_{\abs{x_1-x_2}<b}\Bigg[\left(|\partial_{1}W_{12}(x_1-x_2)\chi|^2+ \frac{1}{2}\braket{W_{12}(x_1-x_2)\chi|V_{12}W_{12}(x_1-x_2)\chi}\right)\frac{\rho^{(2)}(x_1,x_2)}{(x_1-x_2)^2}\\
&\hspace{1cm}-|W_{12}(x_1-x_2)\chi|^2\frac{1}{x_1-x_2}\partial_{y_1}^2\left(\frac{\gamma^{(2)}(x_1,x_2;y_1,y_2)}{y_1-y_2}\right)\Bigg\rvert_{y=x}\Bigg]dx_1dx_2.
\end{aligned}
\]
We estimate the three terms on the right side separately, using Lemma~\ref{Lemma rho2 bound} and the function $Q$ defined there. Throughout, let $\tau:=\braket{\chi|(b-\mathsf{A}^{1,2})^{-1}\chi}$, and note that $\tau\leq2/b$ since $b>2\norm{\mathsf{A}}$.

\emph{The first term.} By the convention above, the first term equals
\[
\begin{aligned}
&-\frac1b\int_0^{L-b}\rho^{(2)}(x_2+b,x_2)dx_2-\frac1b\int_b^L\rho^{(2)}(x_2-b,x_2)dx_2\\
&\qquad\leq-2b\int_0^LQ+\textnormal{const. }\left(\rho^4b^2+\rho^6b^3L\right).
\end{aligned}
\]
Here, we used $\rho^{(2)}(x_2\pm b,x_2)\geq Q(x_2\pm\frac b2)b^2-\textnormal{const. }\rho^6b^4$, and $Q\leq8\pi^2\rho^4$ to replace the integrals of $Q(x_2\pm\frac b2)$ over $(0,L-b)$ and $(b,L)$ by $\int_0^LQ$.

\emph{The second term.} Its integrand is non-negative, since $\rho^{(2)}\geq0$ on $\R^2$. We may therefore extend the $x_1$-integration to all of $(x_2-b,x_2+b)$, and we substitute $x_1=x_2+t$. We use $\rho^{(2)}(x_2+t,x_2)\leq Q(x_2+\frac t2)t^2+\textnormal{const. }\rho^6t^4$, $\int_0^LQ(x_2+\frac t2)dx_2\leq\int_0^LQ+4\pi^2\rho^4b$, and the identity
\[
\int_{-b}^b\abs{\partial_tW_{12}(t)\chi}^2+\frac12\braket{W_{12}(t)\chi|V_{12}(t)W_{12}(t)\chi}dt=2b^2\tau,
\]
which follows from \eqref{EqScatteringEnergy} with $R=b$. This bounds the second term by $2b^2\tau\int_0^LQ+\textnormal{const. }(\rho^4b^2+\rho^6b^3L)$.

\emph{The third term.} By Lemma~\ref{LemmaDensityBounds} and the uniform bound $|W_{12}(x_1-x_2)\chi|^2\leq b^2$,\footnote{\label{footconvexbound}This bound follows from convexity of $|W_{12}(x_1-x_2)\chi|^2$ with boundary conditions $|W_{12}(\pm b)\chi|^2=b^2$, which in turn follows from the fact that $W_{12}(x_1-x_2)=bF_{12}(|x_1-x_2|)$ and $F_{12}$ solves \eqref{EqEvenScatteringEquation} with positive semidefinite $V$ (see \eqref{convexity} for the straightforward calculation of the second derivative).} the third term is bounded by const.\ $\rho^6b^3L$.

\emph{Conclusion.} The main parts of the first two terms add up to $(2b^2\tau-2b)\int_0^LQ$, where $2b^2\tau-2b=2b\braket{\chi|\mathsf{A}^{1,2}(b-\mathsf{A}^{1,2})^{-1}\chi}$ is bounded by $4\norm{\mathsf{A}}$ in absolute value. The identity for $\int_0^LQ$ in Lemma~\ref{Lemma rho2 bound} and $\rho^4L=N\rho^3$ therefore give
\[
\begin{aligned}
&\tbinom{N}{2}\int_{A_{12}}2\abs{\partial_1(\Psi)_{12}\chi}^2+\braket{(\Psi)_{12}\chi|V_{12}(\Psi)_{12}\chi}-\sum_{i=1}^{N}\abs{\partial_i \Psi_F}^2\\
&\leq N\frac{\pi^2}{3}\rho^2\cdot2\rho\braket{\chi\bigg\vert \mathsf{A}^{1,2}\frac{b}{b-\mathsf{A}^{1,2}}\chi}+\textnormal{const. }\left(\rho^3\norm{\mathsf{A}}+\rho^4b^2+N\rho^5b^3\right).
\end{aligned}
\]
Since $E_F=N\frac{\pi^2}{3}\rho^2(1+\mathcal{O}(1/N))$, the right side is bounded by the right side of \eqref{042094} with $\frac1N\sum_i\mathsf{A}^{i,i+1}$ replaced by $\mathsf{A}^{1,2}$. Finally, to find the desired result \eqref{042094}, we rewrite the operator in a translation-independent way (with periodic boundary conditions) using the translation symmetry of $\chi$.
\end{proof}

Combining \eqref{EqTrialStateEnergyUpperBound2} with Lemma~\ref{LemmaE1BoundSpin1/2}, we arrive at the following estimate of $\mathcal{E}(\Psi_\chi)$, which will be used to estimate the energy of the full trial state in the next section.
\begin{proposition}
\label{LemmaUpperBoundFewParticles}
	Assuming $N(\rho b)^3< 1$ and $b>2\max(\norm{\mathsf{A}},R_0)$, the energy of the (unnormalized) trial state $\Psi_\chi$ defined in \eqref{EqTrialStateSpin1/2Fermi} is bounded by 
\[
\begin{aligned}
\mathcal{E}(\Psi_\chi)\leq E_F\Bigg[1&+\braket{\chi \Bigg\vert \frac{1}{N}\sum_{i=1}^{N}2\rho \mathsf{A}^{i,i+1}\frac{b}{b-\mathsf{A}^{i,i+1}}\chi}\Bigg]\\
&+\textnormal{const. } E_F\left(N(\rho b)^3\left(1+\rho b^2 \int \|\tilde{V}\|\right)+\big(\rho\norm{\mathsf{A}}+(\rho b)^2\big)N^{-1}\right).
\end{aligned}
\]
\end{proposition}

\subsubsection{A full trial state (proof of Theorem~\ref{TheoremUpperBoundSpinJFermi})}
\label{proofupperbound}
This proof is similar to the strategy used for bosons \cite{agerskov2022ground}: concatenating copies of the (Dirichlet) trial state $\Psi_\chi$ on smaller intervals, or `boxes'. 
\begin{proof}[Proof of Theorem~\ref{TheoremUpperBoundSpinJFermi}]
 Let $\rho\norm{\mathsf{A}}<\frac{1}{32}$ and $\rho R_0\leq\frac{1}{48}$, and fix $b$ to satisfy 
 \begin{equation}
 \label{choiceb}
 b=\max\left(\frac{(\rho\norm{\mathsf{A}})^{4/5}}{\rho},3R_0\right).
 \end{equation}
 Then, $(\rho\norm{\mathsf{A}})^{4/5}<\frac{1}{16}$, which implies $\frac{1}{16}\geq \rho b\mathrel{\geq}(\rho \norm{\mathsf{A}})^{-1/5}\rho \norm{\mathsf{A}}\geq 2\rho \norm{\mathsf{A}}$, so that $b>2\max(\norm{\mathsf{A}},R_0)$, as assumed in Lemmas~\ref{extralemma}--\ref{LemmaE1BoundSpin1/2}. Therefore, we find
	\begin{equation}
 \label{someeq8}
		\mathsf{A}\frac{b}{b-\mathsf{A}}=\mathsf{A}+\mathsf{A}\frac{\mathsf{A}}{b-\mathsf{A}}\leq\mathsf{A}+2\norm{\mathsf{A}}^2/b\leq  \mathsf{A}+2\norm{\mathsf{A}}(\rho\norm{\mathsf{A}})^{1/5}.
	\end{equation}
 Notice first that one box suffices for $N<\frac{1}{2}(\rho b)^{-3/2}$: the result follows directly from Proposition~\ref{LemmaUpperBoundFewParticles} and Lemma~\ref{LemmaDensityBounds} with the observation that
  \begin{equation}
  \label{someeq9}
  \int_{[0,L]^N}|\Psi_\chi|^2\geq 1-\int_B\abs{\Psi_F}^2\geq 1-\int_{\abs{x_1-x_2}<b}\rho^{(2)}(x_1,x_2)\geq 1-\frac{16\pi^2}{3}N(\rho b)^3,  
\end{equation}
as well as $\frac{16\pi^2}{3}N(\rho b)^3\leq \frac{8\pi^2}{3}(\rho b)^{3/2}< 1/2$, so that $\frac{1}{1-\frac{16\pi^2}{3}N(\rho b)^3}\leq 1+\frac{32\pi^2}{3} N(\rho b)^3$. The error term of order $N(\rho b)^3\rho b^2\int\|\tilde{V}\|$ is of the same order as the others, since $\rho b^2\int\|\tilde{V}\|\leq1$. Indeed, \eqref{choiceb} gives $\rho b^2\leq9m(\rho m)^{3/5}$ with $m:=\max(R_0,\norm{\mathsf{A}})$, and $\rho m$ is small enough, depending on $m\int\|\tilde{V}\|$, by the assumptions of Theorem~\ref{TheoremUpperBoundSpinJFermi}.

For $N\geq \frac{1}{2}(\rho b)^{-3/2}$, we consider $M$ disjoint `boxes' (intervals) $I_1,I_2,\dots,I_M$, with $\bigcup_{k=1}^{M}I_k\subset [0,L]$. We then set the number of particles in each box to be either $\ceil{N/M}$ or $\ceil{N/M}-1$, adding up to $N$ in total. We choose
\begin{equation}
\label{choiceM}   
M=\ceil{2N(\rho b)^{3/2}},
\end{equation}
 so that
  $$ \frac18(\rho b)^{-3/2}\leq \ceil{N/M}-1 <\ceil{N/M}\leq \frac{1}{2}(\rho b)^{-3/2}+1\leq \frac{33}{64}(\rho b)^{-3/2}.$$
  Hence, the number of particles in each box lies between $\frac{1}{8}(\rho b)^{-3/2}$ and $\frac{33}{64}(\rho b)^{-3/2}$ (for the lower bound, note that $N/M>\frac14(\rho b)^{-3/2}$ and $\frac18(\rho b)^{-3/2}\geq 8$).
We write $\tilde{N}=\ceil{N/M}$ and fix the length of each box so that the density is $\tilde{\rho}=\frac{N}{L-(M-1)b}$. In that way, the sum of the lengths is $L-(M-1)b$, leaving a distance $ b $ between each box. Also,  \begin{equation}\label{EqTilderhoEstimate}
    \tilde{\rho}=\frac{\rho}{1-\frac{b(M-1)}{L}}\leq \rho(1+2b(M-1)/L)\leq \rho(1+4(\rho b)^{5/2}), 
\end{equation} which follows from $ b(M-1)/L\leq 2(\rho b)^{5/2}< 1/2 $. 

We now consider the glued trial state $\Psi_{\chi,\text{full}}=\prod_{k=1}^{M}\Psi_{\chi,k}$, where $\Psi_{\chi,k}$ denotes the state defined in \eqref{EqTrialStateSpin1/2Fermi}, on the box $I_k$.
Note that there is no interaction energy between the boxes, as the distance $b$ between the boxes \eqref{choiceb} is greater than the range of the potential $R_0$.

 To study the energy within each box, let $ e_{0,I_k}=\frac{\pi^2}{3}\abs{I_k}\tilde{\rho}^3(1+\text{const. }\tilde{N}^{-1}) $ be the free Fermi energy in box $I_k$ (where $\abs{I_k}$ denotes the length of the box), let $n_k\in\{\tilde N-1,\tilde N\}$ be the number of particles in $I_k$, and let $\chi_k$ be the spin state used in $\Psi_{\chi,k}$.
 From Proposition~\ref{LemmaUpperBoundFewParticles}, the trial state constructed using the small boxes gives the upper bound 
 \begin{equation}\label{EqUpperBoundSmallN}
		\begin{aligned}
			E^{\textnormal{Dirichlet}}_{J}(N,L)\leq \sum_{k=1}^{M}\Big(\int_{I_k^{n_k}}|\Psi_{\chi,k}|^2&\Big)^{-1}e_{0,I_k}\Bigg[1+\braket{\chi_k \Bigg\vert \frac{1}{n_k}\sum_{i=1}^{n_k}2\tilde{\rho} \mathsf{A}^{i,i+1}\frac{bI}{bI-\mathsf{A}^{i,i+1}}\chi_k}\\
            &+\textnormal{const. } \left(\tilde{N}(\tilde{\rho} b)^3\left(1+\tilde{\rho} b^2 \int \|\tilde{V}\|\right)+n_k^{-1}\right)\Bigg].
		\end{aligned}
	\end{equation}
 Here, we used Proposition~\ref{LemmaUpperBoundFewParticles} and $\tilde{\rho}\norm{\mathsf{A}}+(\tilde\rho b)^2\leq1$. Notice that $n_k^{-1}\leq(\tilde N-1)^{-1}\leq8(\rho b)^{3/2}$. Hence, this term is subleading, and we absorb it into other error terms. We now choose $\chi_k$ to be a translation-symmetric ground state of the periodic chain $\sum_{i=1}^{n_k} \mathsf{A}^{i,i+1}$, and we write $E_{\mathsf{A}}(n_k)$ for its ground state energy. The sum $H'$ of the periodic chains on $M$ consecutive blocks of $n_1,\dots,n_M$ spins differs from the periodic chain $\sum_{i=1}^{N}\mathsf{A}^{i,i+1}$ by $M$ bonds removed and $M$ bonds added. Its ground state energy $\sum_kE_{\mathsf{A}}(n_k)$ is therefore at most its expectation in a ground state of $\sum_{i=1}^{N}\mathsf{A}^{i,i+1}$, so we find 
  $$
\sum_{k=1}^{M}E_{\mathsf{A}}(n_k)\leq N\epsilon_{J}(\mathsf{A})+2M\norm{\mathsf{A}}, 
 $$
 where $\epsilon_{J}(\mathsf{A})$ was defined to be the ground state energy of $\frac{1}{N}\sum_{i=1}^{N}\mathsf{A}^{i,i+1}$. Therefore, using \eqref{someeq8}, $e_{0,I_k}=\frac{\pi^2}{3}n_k\tilde\rho^2(1+\textnormal{const. }\tilde N^{-1})$ and $\abs{E_{\mathsf{A}}(n_k)}\leq n_k\norm{\mathsf{A}}$, we find
 $$
 \begin{aligned}
 &\sum_{k=1}^{M}\frac{e_{0,I_k}}{n_k}\braket{\chi_k \Bigg\vert \sum_{i=1}^{n_k}2\tilde{\rho} \mathsf{A}^{i,i+1}\frac{bI}{bI-\mathsf{A}^{i,i+1}}\chi_k}\\
 &\qquad\leq N\frac{\pi^2}{3}\tilde\rho^2\Big(2\tilde{\rho} \epsilon_{J}(\mathsf{A})+4\tilde\rho\norm{\mathsf{A}}(\rho\norm{\mathsf{A}})^{1/5} + \textnormal{const. }\frac{M}{N}\tilde{\rho} \norm{\mathsf{A}}\Big).
 \end{aligned}
 $$
  The last term is of order $(\rho b)^{3/2}\rho \norm{\mathsf{A}}$ by \eqref{choiceM} and is therefore subleading.
Using \eqref{someeq9} and \eqref{EqTilderhoEstimate} and the observations above, we see that \eqref{EqUpperBoundSmallN} is upper bounded by
\begin{equation}
\label{lasteq}
\begin{aligned}
N\frac{\pi^2}{3}\rho^2\frac{1+2\rho \epsilon_J(\mathsf{A}) +4(\rho \norm{\mathsf{A}})^{6/5}+\text{const. }\left[\frac{M}{N}+\tilde{N}(\rho b)^3\left(1+\rho b^2 \int\|\tilde{V}\|\right)\right]}{1-\frac{16\pi^2}3 \tilde{N}(\tilde\rho b)^3}.
\end{aligned}
\end{equation}
Using $M/N\leq 2(\rho b)^{3/2}+\frac{1}{N}$ and the fact that $ \frac{16\pi^2}{3}\tilde{N}(\tilde\rho b)^3\leq 1/2 $, so that $$ 1\Big/\Big(1-\frac{16\pi^2}{3}\tilde{N}(\tilde\rho b)^3\Big)\leq 1+\frac{32\pi^2}{3}\tilde{N}(\tilde\rho b)^3, $$
 as well as $\tilde{N}(\tilde\rho b)^3\leq \textnormal{const. }(\rho b)^{3/2}$,
 the desired result \eqref{bounddirichlet} follows (recall the remark about $\int\|\tilde{V}\|$ below \eqref{someeq9}). Note the error term $N\frac{\pi^2}{3}\rho^2\mathcal{O}(N^{-1})$ in \eqref{bounddirichlet} is only important when $N<(\rho b)^{-3/2}$, so it does not appear in \eqref{lasteq}.
\end{proof}

This concludes the proof of Theorem~\ref{TheoremUpperBoundSpinJFermi}. It remains to prove the lemmas that we used.

\subsubsection{Proof of Lemmas~\ref{extralemma0}, \ref{extralemma} and~\ref{LemmaSpin1/2EtaContribution1}}
\label{lemmas}
\begin{proof}[Proof of Lemma~\ref{extralemma0}]  
    Given that the integration domain is $B^\geq_{12}$, we use \eqref{translation}, and extend to $A_{12}\supset B^\geq_{12}$ by positivity. As $\Psi_F$ is an eigenfunction of $-\sum_{i=1}^{N}\partial_i^2$ with eigenvalue $E_F$, by partial integration, we see that
\[
	\begin{aligned}
		\tbinom{N}{2}\int_{B^\geq_{12}}\sum^N_{i=3}\abs{\partial_i\Psi_{\chi}}^2&\leq \tbinom{N}{2}\int_{A_{12}}\sum_{i=3}^{N}\braket{\partial_i(\Psi)_{12}\chi \vert \partial_i(\Psi)_{12}\chi}\\
        &=\tbinom{N}{2}\int_{A_{12}}\sum_{i=3}^{N} \braket{W_{12}(x_2-x_1)\chi\vert W_{12}(x_2-x_1)\chi }\overline{\Psi_F}\frac{1}{(x_1-x_2)^2}(-\partial^2_i)\Psi_F\\\quad&=E_F\tbinom{N}{2}\int_{A_{12}}\braket{W_{12}(x_2-x_1)\chi \vert W_{12}(x_2-x_1)\chi} \frac{1}{\abs{x_1-x_2}^2}\left\lvert \Psi_F\right\rvert^2\\
        &\quad  -2\tbinom{N}{2}\int_{A_{12}} \overline{\Psi_F}\frac{\braket{W_{12}(x_2-x_1)\chi \vert W_{12}(x_2-x_1)\chi}}{(x_1-x_2)^2}(-\partial^2_1)\Psi_F\\
        &\leq   b^2 \int_{\abs{x_1-x_2}\leq b} \left[E_F\frac{\rho^{(2)}(x_1,x_2)}{(x_1-x_2)^2}+\frac{1}{2}\frac{\abs{\sum_{i=1}^{2}\partial_{y_i}^2\gamma^{(2)}(x_1,x_2;y_1,y_2)\vert_{y=x}}}{(x_1-x_2)^2}\right]\\
        &\leq \textnormal{const. } E_F N(\rho b)^3,
	\end{aligned}
\]
where we used $|W_{12}(x_1-x_2)\chi|^2\leq b^2$ (see Footnote~\ref{footconvexbound}) and Lemma~\ref{LemmaDensityBounds} in the last two steps. 
\end{proof}

\begin{proof}[Proof of Lemma~\ref{extralemma}]
Recalling the definitions \eqref{Wij}, \eqref{hcij} and \eqref{eta}, define
\[
(\tilde{\Psi})_{12}(x):=\frac{\Psi_F(x)}{\abs{x_2-x_1}} \tilde{W}_{12}(x),
\]
with $\tilde{W}_{12}(x)=\eta(x)W_{12}(x_2-x_1)+(1-\eta(x))W_{\textnormal{hc},12}(x_2-x_1)$. Note that $\Psi_\chi(x)=(\tilde{\Psi})_{12}(x)\chi$ on the set $B_{12}\cap\{0\leq x_1<\dots<x_N<L\}$ by the definition \eqref{EqTrialStateSpin1/2Fermi}. Similar to \eqref{psionsectors}, we find that on $B_{12}\cap \{\sigma\}$ (assuming the intersection is not empty),
\begin{equation}
\label{psionsectors2}
    \Psi_\chi(x)=(\tilde{\Psi})_{12}\chi_\sigma(x)=\frac{\Psi_F(x)}{\abs{x_2-x_1}}\tilde{W}_{12}(x)\chi_\sigma.
\end{equation}
Note that we can repeat the argument that led to \eqref{translation} and \eqref{potenergyeq} to find that on all of $B_{12}$,
\begin{equation}
\label{lastlabel}
\braket{\Psi_\chi(x)|\Psi_\chi(x)}=\langle(\tilde{\Psi})_{12}(x)\chi|(\tilde{\Psi})_{12}(x)\chi\rangle.
\end{equation}
Before we start estimating, it is useful to note that 
\begin{equation}
\label{convexity1}
\langle\tilde{W}_{12}(x)\chi\vert \tilde{W}_{12}(x)\chi\rangle\leq \braket{W_{12}(x_2-x_1)\chi\vert W_{12}(x_2-x_1)\chi}\leq b^2.
\end{equation}
The first inequality follows from the convexity of squaring and the fact that $\braket{W_{\textnormal{hc},12}\chi\vert W_{\textnormal{hc},12}\chi}\leq \braket{W_{12}\chi\vert W_{12}\chi}$ pointwise (see Lemma~\ref{LemmaScatteringSolution hard core pointwise bound} in Appendix~\ref{AppendixA} for a proof). The second follows from convexity of $\braket{W_{12}\chi\vert W_{12}\chi}$ (see Footnote~\ref{footconvexbound} and the proof just referenced) and the boundary condition $|W_{12}(\pm b)\chi|^2=b^2$. 

Recall that by Assumption~\ref{AssumptionUpperBound}, we can write $V_{ij}=v_{ij}+\tilde{V}_{ij}$, with $v_{ij}$ positive and scalar, and $\tilde{V}_{ij}$ a Hermitian matrix-valued measure of bounded variation, both with support contained in $[-b,b]$. A useful fact for the estimates is that, for a bounded Borel function $g$ on $\R$ with values in $(\C^{2J+1})^{\otimes N}$,\footnote{To see this, it suffices to study a simple function $g(x)=\sum_{k=1}^{n}\xi_k \mathbbm{1}_{X_k}$. Then $\int \langle g(x)|\tilde{V}(x) g(x)\rangle dx=\sum_{k=1}^{n}\int_{X_k} \langle\xi_k|\tilde{V}\xi_k\rangle\leq\sum_{k=1}^{n}\int_{X_k} |\xi_k|^2\|\tilde{V}\|= \int|g(x)|^2\|\tilde{V}\|(x)dx$. We can then generalize to bounded Borel functions by uniform approximation with simple functions.} 
\begin{equation}
\label{estimate20023}
\int \langle g(x)|\tilde{V}(x) g(x)\rangle dx\leq \int|g(x)|^2\|\tilde{V}\|(x)dx.
\end{equation}
This gives 
\begin{equation}
\label{someeq12910}
\tbinom{N}{2}\int_{B_{12}}\sum_{\substack{1\leq i<j\leq N\\(i,j)\neq (1,2)}}\braket{\Psi_\chi|V_{ij}\Psi_\chi}\leq \tbinom{N}{2}\int _{B_{12}}\sum_{\substack{1\leq i<j\leq N\\(i,j)\neq (1,2)}}\left(v_{ij}|\Psi_\chi|^2+\|\tilde{V}_{ij}\||\Psi_\chi|^2\right).
\end{equation}
By \eqref{lastlabel} and antisymmetry to collect similar terms, we find that the $\tilde{V}_{ij}$-part of \eqref{someeq12910} satisfies
	\begin{equation}\label{EqE3bound0}
		\begin{aligned}
			&\tbinom{N}{2}\int_{B_{12}} \Bigg(\sum_{2\leq i<j\leq N}\|\tilde{V}_{ij}\||(\tilde{\Psi})_{12}\chi|^2+\sum_{j=3}^{N}\|\tilde{V}_{1j}\||(\tilde{\Psi})_{12}\chi|^2\Bigg)\\&\leq \text{const. }b^2\int_{\{\abs{x_1-x_2}<b\}\cap\{\abs{x_3-x_4}<b\}}\|\tilde{V}\|(\abs{x_3-x_4})\frac{1}{(x_1-x_2)^2}\rho^{(4)}(x_1,x_2,x_3,x_4)\\
			&\ +\text{const. }b^2\int_{\{\abs{x_1-x_2}<b\}\cap\{\abs{x_2-x_3}<b\}}\|\tilde{V}\|(\abs{x_2-x_3})\frac{1}{(x_1-x_2)^2}\rho^{(3)}(x_1,x_2,x_3)\\
            &\leq \textnormal{const. }(N^2+N)\rho^2\cdot (\rho b)^3\cdot \rho b^2
            \int \|\tilde{V}\|\leq \textnormal{const. }E_FN(\rho b)^3\cdot \rho b^2
            \int \|\tilde{V}\|,
		\end{aligned}
	\end{equation}
    where we used Lemma~\ref{LemmaDensityBounds} for the estimates.
For the $v_{ij}$-part of \eqref{someeq12910}, we note that for fixed $\abs{x_i-x_j}$ and $(i,j)\neq (1,2)$, by positivity of $v_{ij}$,
\begin{equation}
\label{someeq128}
v(x_i-x_j)|W_{12}(x_1-x_2)\chi|^2\leq v(x_i-x_j)|W_{12}(x_i-x_j)\chi|^2
\end{equation}
since the function $|W_{12}(x_1-x_2)\chi|^2$ is increasing in $|x_1-x_2|$ because it is convex and even, and $\abs{x_i-x_j}\geq\abs{x_1-x_2}$ on $B_{12}$. Thus, using \eqref{lastlabel}, followed by \eqref{convexity1} and \eqref{someeq128} and antisymmetry, and integrating out certain coordinates, we find that
\begin{equation}
\label{EqE3bound}
\begin{aligned}
		&\tbinom{N}{2}\int_{B_{12}} \left(\sum_{2\leq i<j\leq N}v_{ij}|(\tilde{\Psi})_{12}\chi|^2+\sum_{j=3}^{N}v_{1j}|(\tilde{\Psi})_{12}\chi|^2\right)\\
            &\leq \text{const. }\left(\int_{\{\abs{x_1-x_2}<b\}\cap\{\abs{x_3-x_4}<b\}}v(\abs{x_3-x_4})\abs{W_{12}(x_3-x_4)\chi}^2\frac{1}{(x_1-x_2)^2}\rho^{(4)}(x_1,x_2,x_3,x_4)\right.\\
			&\qquad\qquad\left.+\int_{\{\abs{x_1-x_2}<b\}\cap\{\abs{x_2-x_3}<b\}}v(\abs{x_2-x_3})\abs{W_{12}(x_2-x_3)\chi}^2\frac{1}{(x_1-x_2)^2}\rho^{(3)}(x_1,x_2,x_3)\right).
		\end{aligned}
	\end{equation}
	Hence, combining \eqref{someeq12910},  \eqref{EqE3bound0} and \eqref{EqE3bound}, and by adding and subtracting interaction terms with $\tilde{V}_{12}(x_2-x_3)$ and $\tilde{V}_{12}(x_3-x_4)$ instead of $v(|x_2-x_3|)$ and $v(|x_3-x_4|)$ and estimating the subtraction as in \eqref{someeq12910} and \eqref{EqE3bound0}, and using Lemma~\ref{LemmaDensityBounds}, we find that
\[
\begin{aligned}
&\tbinom{N}{2}\int_{B_{12}}\sum_{\substack{1\leq i<j\leq N\\(i,j)\neq (1,2)}}\braket{\Psi_\chi|V_{ij}\Psi_\chi}\\
&\leq \text{const. }\Bigg(\rho^8\int_{\{\abs{x_1-x_2}<b\}\cap\{\abs{x_3-x_4}<b\}}(x_3-x_4)^2\braket{W_{12}(x_3-x_4)\chi| V_{12}(x_3-x_4)W_{12}(x_3-x_4)\chi}\\
&\qquad\qquad\quad\left.+\rho^7\int_{\{\abs{x_1-x_2}<b\}\cap\{\abs{x_2-x_3}<b\}}(x_2-x_3)^2\braket{W_{12}(x_2-x_3)\chi| V_{12}(x_2-x_3)W_{12}(x_2-x_3)\chi}\right.\\
&\qquad\qquad\quad
+ E_FN(\rho b)^3\cdot\rho b^2\int \|\tilde{V}\|
\Bigg)\\
 &\leq\text{const. } \left(N^2+N\right)\rho^2(\rho b)^4\cdot b\braket{\chi\vert(b-\mathsf{A}^{1,2})^{-1}\chi}+E_FN(\rho b)^3\cdot\rho b^2\int \|\tilde{V}\|\\
&\leq\text{const. }E_F N (\rho b)^3\Big[\rho b+\rho b^2\int \|\tilde{V}\|\Big],
		\end{aligned}
\]
where we used the scattering energy \eqref{EqScatteringEnergy} and the assumption  $b>2\|\mathsf{A}\|$ in the last steps. 
\end{proof}

\begin{proof}[Proof of Lemma~\ref{LemmaSpin1/2EtaContribution1}]
Consider the trial state $\Psi_\chi$  \eqref{EqTrialStateSpin1/2Fermi} on the set $B_{12}^{34}$ \eqref{intdomains}, which means that the particles with positions $x_1$ and $x_2$ are the closest pair, while those with positions $x_3$ and $x_4$ are the next closest pair. Similar to \eqref{psionsectors2} (also see \eqref{psionsectors} and \eqref{examplepsi}), we see that on $B_{12}^{34}\cap\{\sigma\}$,
\begin{equation}
\label{someeq1013039}
\begin{aligned}
\Psi_{\chi}(x)&= \frac{\Psi_F(x)}{\abs{x_2-x_1}} W_{12}^{34}(x)\chi_\sigma\\
&\coloneqq  \frac{\Psi_F(x)}{\abs{x_2-x_1}}\Big[\eta_{34}(x) W_{12}(x_1-x_2)+(1-\eta_{34}(x))W_{\text{hc},12}(x_1-x_2)\Big]\chi_\sigma,
\end{aligned}
\end{equation}
where we defined $W_{12}^{34}(x)=W_{12}^{34}(x_1,x_2,x_3,x_4)$ in terms of $W_{12}$ \eqref{Wij} and $W_{\textnormal{hc},12}$ \eqref{hcij}, 
\[
\begin{aligned}
\eta_{34}(x):=\begin{cases}
0&\text{ if } \abs{x_3-x_4}\leq b\\
\frac{\abs{x_3-x_4}}{b}-1 &\text{ if } b<\abs{x_3-x_4}<2b\\
1 &\text{ if } \abs{x_3-x_4}\geq 2b
\end{cases}.
\end{aligned}
\]
Also, because of reasoning similar to \eqref{lastlabel}, $ B_{12}^{34}\subset A_{12}^{34} $, and partial integration in $x_5,\dots, x_N$, 
\[
\begin{aligned}
\tbinom{N}{2}\tbinom{N-2}{2}\int_{B_{12}^{34}}\sum_{i=1}^{N}\abs{\partial_i\Psi_{\chi}}^2&=\tbinom{N}{2}\tbinom{N-2}{2}\int_{B_{12}^{34}}\sum_{i=1}^{N}\abs{\partial_i\left( \frac{\Psi_F}{\abs{x_2-x_1}} W_{12}^{34}\chi\right)}^2\\
&\leq \tbinom{N}{2}\tbinom{N-2}{2}\int_{A_{12}^{34}}\sum_{i=1}^{N}\abs{\partial_i\left( \frac{\Psi_F}{\abs{x_2-x_1}} W_{12}^{34}\chi\right)}^2\\
&= \tbinom{N}{2}\tbinom{N-2}{2}\int_{A_{12}^{34}}\sum_{i=1}^{4}\abs{\partial_i\left( \frac{\Psi_F}{\abs{x_2-x_1}} W_{12}^{34}\chi\right)}^2\\
&\ \ \ + \tbinom{N}{2}\tbinom{N-2}{2}\int_{A_{12}^{34}}\sum_{i=5}^{N}\braket{ \frac{\Psi_F}{\abs{x_2-x_1}} W_{12}^{34}\chi \Bigg\vert \frac{-\partial_i^2\Psi_F}{\abs{x_2-x_1}} W_{12}^{34}\chi }.
\end{aligned}
\]
 Using that $ \Psi_F $ is an eigenfunction of $ -\Delta $ with eigenvalue $ E_F $, we find that
\begin{equation}
\label{someeq1320931}
\begin{aligned}
\tbinom{N}{2}\tbinom{N-2}{2}&\int_{B_{12}^{34}}\sum_{i=1}^{N}\abs{\partial_i\Psi_{\chi}}^2\leq \tbinom{N}{2}\tbinom{N-2}{2}\Bigg[\int_{A_{12}^{34}}\sum_{i=1}^{4}\abs{\partial_i\left( \frac{\Psi_F}{\abs{x_2-x_1}} W_{12}^{34}\chi\right)}^2\\
&\hspace{1cm}-\int_{A_{12}^{34}}\sum_{i=1}^{4}\braket{ \frac{\Psi_F}{\abs{x_2-x_1}} W_{12}^{34}\chi \Bigg\vert \frac{-\partial_i^2\Psi_F}{\abs{x_2-x_1}} W_{12}^{34}\chi }+ E_F\int_{A_{12}^{34}}\abs{\frac{\Psi_F}{\abs{x_2-x_1}} W_{12}^{34}\chi}^2 \Bigg].
\end{aligned}
\end{equation}
By convexity of squaring, the scattering energy \eqref{EqScatteringEnergy} and the assumption $b>2\|\mathsf{A}\|$, we find that 
\[
    \int_{x_2-b}^{x_2+b}\braket{\partial_1  W_{12}^{34}\chi\vert \partial_1  W_{12}^{34}\chi}dx_1 \leq b^2\max\left(\frac{2}{b-\norm{\mathsf{A}}},\frac{2}{b-R_0}\right)\leq 4b.
\]
Also, similar to \eqref{convexity1}, we find that $\braket{W_{12}^{34}\chi\vert  W_{12}^{34}\chi}\leq \braket{ W_{12}\chi\vert  W_{12}\chi}\leq b^2$.
Using the dependence on $x_3$ in \eqref{someeq1013039} for the first inequality and Cauchy--Schwarz for the other two, we find that
\[
\begin{aligned}
    \int_{x_2-b}^{x_2+b}\braket{\partial_3  W_{12}^{34}\chi\vert \partial_3  W_{12}^{34}\chi}dx_1 &\leq 8b\\
        \int_{x_2-b}^{x_2+b}\abs{\braket{\partial_1 W_{12}^{34}\chi\vert  W_{12}^{34}\chi} }dx_1&\leq 4b^2\\
        \int_{x_2-b}^{x_2+b}\abs{\braket{\partial_3 W_{12}^{34}\chi\vert  W_{12}^{34}\chi} }dx_1&\leq 4b^2.
\end{aligned}
\]
Thus, using Lemmas~\ref{LemmaDensityBounds} and~\ref{LemmaDensityBounds2} and the assumption $N(\rho b)^3<1$ (and the notation from Section~\ref{freefermi}), we find that \eqref{someeq1320931} (the kinetic energy part of the object we wish to estimate) is upper bounded by 
\[
    \begin{aligned}
        \textnormal{const. }&\Bigg(\sum_{i=1}^{4} \Bigg[\int_{\substack{\{\abs{x_1-x_2}<b\}\\\{\abs{x_3-x_4}<2b\}}}b^2\partial_{y_i}\partial_{x_i}\frac{\gamma^{(4)}(y_1,y_2,y_3,y_4;x_1,x_2,x_3,x_4)}{\abs{x_2-x_1}\abs{y_2-y_1}}\Bigg \lvert_{y=x}dx_1dx_2dx_3dx_4
        \\
&+b^2\int_{\{\abs{x_3-x_4}<2b\}}\sup_{x_1\in\{\abs{x_1-x_2}<b\}}\left(\abs{\partial_{y_i}\frac{\gamma^{(4)}(y_1,y_2,y_3,y_4;x_1,x_2,x_3,x_4)}{\abs{x_2-x_1}\abs{y_2-y_1}}\Bigg \lvert_{y=x}}\right) dx_2dx_3dx_4\Bigg]\\
&+b\int_{\{\abs{x_3-x_4}<2b\}}\sup_{x_1\in\{\abs{x_1-x_2}<b\}}\frac{\rho^{(4)}(x_1,x_2,x_3,x_4)}{(x_1-x_2)^2}dx_2dx_3dx_4\\
&+b^2\int_{\substack{\{\abs{x_1-x_2}<b\}\\\{\abs{x_3-x_4}<2b\}}} \frac{1}{(x_1-x_2)^2}\abs{\sum_{i=1}^{4}\partial_{y_i}^2\gamma^{(4)}(y_1,y_2,y_3,y_4;x_1,x_2,x_3,x_4) \Bigg \lvert_{y=x}}dx_1dx_2dx_3dx_4\\
&+ E_F b^2\int_{\substack{\{\abs{x_1-x_2}<b\}\\\{\abs{x_3-x_4}<2b\}}}\frac{\rho^{(4)}(x_1,x_2,x_3,x_4)}{(x_1-x_2)^2}dx_1dx_2dx_3dx_4\Bigg)\leq \text{ const. } E_F N\left(\rho b\right)^3.
\end{aligned}
\]
We continue with the potential energy part. The strategy will be similar to what was done in Lemma~\ref{extralemma}, but simplified because we are only dealing with the $V_{12}$-term. By Assumption~\ref{AssumptionUpperBound}, $V_{ij}=v_{ij}+\tilde{V}_{ij}$, with $v_{ij}$ positive and scalar, and $\tilde{V}_{ij}$ a Hermitian matrix-valued measure of bounded variation, both with support contained in $[-b,b]$.
Applying \eqref{estimate20023} and using an argument similar to the one that gave \eqref{lastlabel} (also see \eqref{translation} and \eqref{potenergyeq}), followed by $\braket{W_{12}^{34}\chi\vert  W_{12}^{34}\chi}\leq \braket{ W_{12}\chi\vert  W_{12}\chi}\leq b^2$ already used before, we find that
\[
\begin{aligned}
&\tbinom{N}{2}\tbinom{N-2}{2}\int_{B_{12}^{34}}\braket{\Psi_\chi|V_{12}|\Psi_\chi}\leq \tbinom{N}{2}\tbinom{N-2}{2}\int _{B_{12}^{34}}v_{12}|\Psi_\chi|^2+\|\tilde{V}_{12}\||\Psi_\chi|^2\\
        &= \tbinom{N}{2}\tbinom{N-2}{2}\int _{B_{12}^{34} }v_{12}\frac{\abs{\Psi_F}^2}{\abs{x_2-x_1}^2}\braket{ W_{12}^{34}\chi| W_{12}^{34}\chi}+\|\tilde{V}_{12}\|\frac{\abs{\Psi_F}^2}{\abs{x_2-x_1}^2}\braket{ W_{12}^{34}\chi| W_{12}^{34}\chi}\\
&\leq \tbinom{N}{2}\tbinom{N-2}{2}\int _{B_{12}^{34} }v_{12}\frac{\abs{\Psi_F}^2}{\abs{x_2-x_1}^2}\braket{W_{12}\chi|W_{12}\chi}+\|\tilde{V}_{12}\|\frac{\abs{\Psi_F}^2}{\abs{x_2-x_1}^2}b^2.
    \end{aligned}
\]
Adding and subtracting $\frac{\abs{\Psi_F}^2}{\abs{x_2-x_1}^2}\langle W_{12}\chi|\tilde{V}_{12}W_{12}\chi\rangle$ and estimating the subtraction as in the previous step, followed by $B^{34}_{12}\subset A^{34}_{12}$ and integrating out the coordinates $x_5,...,x_N$, and using Lemma~\ref{LemmaDensityBounds} and the scattering energy \eqref{EqScatteringEnergy} and the assumption $b>2\|\mathsf{A}\|$ gives
\[
    \begin{aligned}
        \tbinom{N}{2}\tbinom{N-2}{2}&\int_{B_{12}^{34}}\braket{\Psi_\chi|V_{12}|\Psi_\chi}\leq \tbinom{N}{2}\tbinom{N-2}{2}\int _{A_{12}^{34} }\frac{\abs{\Psi_F}^2}{\abs{x_2-x_1}^2}\braket{W_{12}\chi|V_{12} W_{12}\chi}+2\|\tilde{V}_{12}\|\frac{\abs{\Psi_F}^2}{\abs{x_2-x_1}^2}b^2\\
        &\leq \int _{\substack{\abs{x_1-x_2}<b\\
        \abs{x_3-x_4 }<2b}}\frac{\rho^{(4)}(x_1,x_2,x_3,x_4)}{\abs{x_2-x_1}^2}\braket{W_{12}\chi|V_{12} W_{12}\chi}+2\|\tilde{V}_{12}\|\frac{\rho^{(4)}(x_1,x_2,x_3,x_4)}{\abs{x_2-x_1}^2}b^2\\
        &\leq \textnormal{const. } N\rho^2 N(\rho b)^4 \left(\braket{\chi\Big\vert\frac{b}{b-\mathsf{A}^{1,2}}\chi}+2b \int \|\tilde{V}\|\right)\\ &\leq \textnormal{const. }E_FN(\rho b)^3\Big(\rho b+\rho b^2\int \|\tilde{V}\|\Big).
    \end{aligned}
\]
This concludes the proof of the first estimate of the lemma. The second estimate can be found by straightforwardly adapting all the steps above to the $B^{23}_{12}$ case. 
\end{proof}

\section{\texorpdfstring{Lower bound for scalar potentials and the LLH model}{Lower bound for scalar potentials and the LLH model}}
\label{seclowerbound}
The following lower bound is used to prove Theorem~\ref{TheoremDiluteFermiGroundStateEnergy}. Our method does not allow us to extend to the matrix-valued potentials that we considered for the upper bound---see Footnote~\ref{Footnote:lower bound for matrix potentials} for the technical obstruction. It does, however, apply to the Lieb--Liniger--Heisenberg model, see Section~\ref{SecLLHLower}.
\begin{theorem}[Lower bound in Theorem~\ref{TheoremDiluteFermiGroundStateEnergy}]
	\label{TheoremLowerBoundSpin}
    Let $v$ be an even, positive interaction compactly supported in $[-R_0,R_0]$, with $a_e$ and $a_o$ its even- and odd-wave scattering lengths. Consider a spin-$J$ Fermi gas with energy functional \eqref{functional} and domain \eqref{domainfermions} with Neumann boundary conditions. Then, for $\rho R_0$, $\rho\abs{a_e}$ and $\rho a_o$ sufficiently small, there exists a constant $C_L>0$ such that the ground state energy $E_{J}^{\textnormal{Neumann}}(N,L)$ satisfies
	\begin{equation}
    \label{lowerboundequation}
		E_{J}^{\textnormal{Neumann}}(N,L)\geq N\frac{\pi^2}{3}\rho^2\left(1+2\rho \left[a_e+(a_o-a_e)\epsilon^{\textnormal{LS}}_J\right]-C_L\big((\rho \max(R_0,a_o,\abs{a_e}))^{6/5}+N^{-1}\big)\right),
	\end{equation}
where $\epsilon^{\textnormal{LS}}_J$ is the minimal energy per site of the spin-$J$ Lai--Sutherland model,\footnote{We will prove the bound for a chain with open boundary conditions, but this is irrelevant, see Footnote~\ref{lastfootnote}.} i.e.\ the minimal expectation value of
    \begin{equation}
    \label{spinchainlb}
h_{\textnormal{LS}}=\frac{1}{N}\sum_{i=1}^{N}P_S^{i,i+1}.
\end{equation}
\end{theorem}

\subsection{\texorpdfstring{Preparatory facts: Dyson's lemma and contact interactions}{Preparatory facts: Dyson's lemma and contact interactions}}
\label{secprep}
For fermions with spin, we will use the following generalization of Dyson's lemma, which was proved for spinless bosons (and fermions) in \cite{agerskov2022ground}.
\begin{lemma}[Dyson for spin-$J$ fermions, see Lemmas 2 and 21 in \cite{agerskov2022ground} and Appendix C in \cite{lieb2006mathematics}]
	\label{LemmaDysonSpin1/2Fermi}
	Let $ R>R_0=\textnormal{range}(v) $ and let
    \[
    \varphi\in \left(H_{\textnormal{even}}^1(\R)\otimes P_A \left(\C^{2J+1}\otimes \C^{2J+1}\right)\right)\oplus\left(H_{\textnormal{odd}}^1(\R)\otimes P_S \left(\C^{2J+1}\otimes\C^{2J+1}\right)\right).
    \]
    Then, for any interval $ \mathcal{I}\ni 0 $, 
	\begin{equation}
    \begin{aligned}
    \label{19283198}
		\int_{\mathcal{I}} \abs{\partial \varphi}^2+\frac12 v\abs{\varphi}^2&\geq \int_{\mathcal{I}}\overline{\varphi}\left(\frac{1}{R-a_e}P_A+\frac{1}{R-a_o}P_S\right)\left(\delta_R+\delta_{-R}\right)\varphi\\
        &\geq \int_{\mathcal{I}}\overline{\varphi}\left(\frac{1}{R-a_e}\right)\left(\delta_R+\delta_{-R}\right)\varphi,
    \end{aligned}
	\end{equation}
	where $ a_{e}$ and $a_o$ are the even- and odd-wave scattering lengths of $v$.
\end{lemma}
The proof is simply that of the bosonic case, or alternatively the spin-polarized fermionic case, both discussed in \cite{agerskov2022ground}, together with the fact that two spin-$J$ fermions have a symmetric spatial wave function only when their spin is antisymmetric and vice versa. The second inequality in \eqref{19283198} follows from positivity and $a_o\geq a_e$.

We will use Lemma~\ref{LemmaDysonSpin1/2Fermi} in the following form, on a half-line. Let $\varphi\in H^1([0,R];\C^{2J+1}\otimes\C^{2J+1})$ with $P_S\varphi(0)=0$. We extend $P_A\varphi$ to an even function and $P_S\varphi$ to an odd function on $[-R,R]$ (and further to $\R$, preserving these symmetries), and apply \eqref{19283198} with $\mathcal{I}=[-R,R]$. By symmetry, both sides of \eqref{19283198} are then twice the corresponding expressions on $[0,R]$, so
\begin{equation}
\label{EqDysonHalfLine}
\int_0^R\abs{\partial\varphi}^2+\frac12v\abs{\varphi}^2\geq\braket{\varphi(R)\bigg|\left(\frac{1}{R-a_e}P_A+\frac{1}{R-a_o}P_S\right)\varphi(R)},
\end{equation}
where an atom of $v$ at $0$ is counted with half its weight on the left side.

\subsubsection{\texorpdfstring{Contact interactions}{Contact interactions}}
\label{SecLLNeumann}
The Lieb--Liniger model describes bosons in one dimension interacting through a contact potential. The formal Hamiltonian (to be understood in the quadratic form sense) is
\begin{equation}
\label{LLhamil}
	H_{\textnormal{LL}}=-\sum_{i=1}^{N}\partial^2_{i}+2c\sum_{1\leq i<j\leq N}\delta(x_i-x_j),
\end{equation}
acting on the Hilbert space $L^2_{\text{sym}}([0,L]^N)$. We write $E_{\textnormal{LL}}^{\textnormal{Neumann}}(N,L,c)$ for its ground state energy with Neumann boundary conditions, at interaction strength $c$ and interval length $L$. In \cite[Lemma~17]{agerskov2022ground}, we bounded this energy from below using the exact solution of the model. Here, we use a simpler bound that does not rely on the Bethe ansatz, and that also applies when the strength of the contact interaction varies between pairs of neighbouring particles, as we will need below. 

Since a contact interaction only acts between neighbouring particles, it is convenient to work on a single sector. For $\ell>0$, let
\[
\Omega_\ell:=\{y\in\R^N\,\vert\,0<y_1<\dots<y_N<\ell\},
\]
and for $1\leq k\leq N-1$, integrals over the face $\{y_k=y_{k+1}\}$ of $\Omega_\ell$ are taken with respect to the coordinates $y_j$, $j\neq k+1$. Given a Hilbert space $\mathcal{K}$ and Hermitian operators $C_1,\dots,C_{N-1}\geq0$ on $\mathcal{K}$, we define, for $\varphi\in H^1(\Omega_\ell;\mathcal{K})$,
\begin{equation}
\label{EqContactForm}
\mathcal{E}_C(\varphi):=\int_{\Omega_\ell}\sum_{i=1}^{N}\abs{\partial_i\varphi}^2+\sum_{k=1}^{N-1}\int_{\{y_k=y_{k+1}\}}\braket{\varphi|C_k\varphi}.
\end{equation}
This describes a contact interaction of strength $C_k$ between the $k$-th and the $(k+1)$-th particle (counted from the left). There are no boundary conditions on the faces or at the ends of $[0,\ell]$. If all $C_k=c_k$ are positive numbers, we write $\mathcal{E}_c$. For example, the Lieb--Liniger energy of a symmetric function $\Psi$ is $N!\,\mathcal{E}_c(\Psi|_{\Omega_L})$, with $c_k=c$ for all $k$.

We will compare $\mathcal{E}_c$ with the ground state energy of $N$ spinless free fermions on an interval $[0,\Lambda]$ with Neumann boundary conditions,
\[
e_N(\Lambda):=\left(\frac{\pi}{\Lambda}\right)^2\sum_{n=0}^{N-1}n^2.
\]
This function is decreasing and convex in $\Lambda>0$, and
\begin{equation}
\label{EqeNbounds}
\frac{\pi^2}{3}\frac{N^3}{\Lambda^2}\left(1-\frac{3}{2N}\right)\leq e_N(\Lambda)\leq\frac{\pi^2}{3}\frac{N^3}{\Lambda^2}.
\end{equation}
If all $c_k$ were infinite, the faces would act as Dirichlet boundary conditions, and the ground state energy of $\mathcal{E}_c$ would be $e_N(\ell)$. For finite $c_k$, the following lemma bounds the energy from below by $e_N$ on a longer interval, where the $k$-th pair of neighbours adds a length $2/c_k$. Since $-2/c_k$ is the scattering length of the potential $2c_k\delta$, this matches the heuristic calculation \eqref{heurist}, in which each pair of neighbours shortens the interval by its scattering length, as for the hard-core potential in Remark~\ref{remarkexamples}. For the Lieb--Liniger model, whose scattering length is $-2/c$, the resulting bound agrees with \eqref{someeq1276} to first order.

\begin{lemma}
\label{LemmaLiebLinigerNeumannLowerBound}
Let $c_1,\dots,c_{N-1}>0$. Then, for all $\varphi\in H^1(\Omega_\ell;\mathcal{K})$,
\[
\mathcal{E}_c(\varphi)\geq e_N\left(\ell+\sum_{k=1}^{N-1}\frac{2}{c_k}\right)\int_{\Omega_\ell}\abs{\varphi}^2 .
\]
In particular, $E^{\textnormal{Neumann}}_{\textnormal{LL}}(N,L,c)\geq e_N(L+2(N-1)/c)\geq\frac{\pi^2}{3}N\rho^2\left(1-4\rho/c-\frac{3}{2N}\right)$.
\end{lemma}
\begin{proof}
The proof is a ground state substitution. We compare $\varphi$ with an explicit positive function $\tilde u$ on $\Omega_\ell$, namely the Neumann free Fermi ground state on $[0,\Lambda]$, with a gap of length $2/c_k$ inserted between the $k$-th and the $(k+1)$-th particle.

Put $\ell_k:=2/c_k$ and $\Lambda:=\ell+\sum_k\ell_k$. For $y\in\overline{\Omega_\ell}$, let $z_i:=y_i+\sum_{k<i}\ell_k$ and $\theta_i:=\pi z_i/\Lambda$, and define
\[
\tilde u(y):=\prod_{1\leq i<j\leq N}(\cos\theta_i-\cos\theta_j).
\]
On $\overline{\Omega_\ell}$, we have $0\leq\theta_1<\dots<\theta_N\leq\pi$, so $\tilde u$ is smooth and strictly positive there. Moreover, $\tilde u$ is (up to sign) the Vandermonde determinant of the numbers $\cos\theta_i$. Since $\cos(n\theta)$ is a polynomial of degree $n$ in $\cos\theta$ with leading coefficient $2^{n-1}$ for $n\geq1$, column operations show that $\tilde u$ is a constant multiple of the Slater determinant $\det[\cos(n\pi z_i/\Lambda)]_{1\leq i\leq N,\,0\leq n\leq N-1}$, which is the Neumann free Fermi ground state on $[0,\Lambda]$ in the variables $z$. Hence
\[
-\Delta\tilde u=e_N(\Lambda)\tilde u\ \ \textnormal{in }\Omega_\ell,\qquad\partial_{y_1}\tilde u=0\ \ \textnormal{on }\{y_1=0\},\qquad\partial_{y_N}\tilde u=0\ \ \textnormal{on }\{y_N=\ell\}.
\]

\emph{The faces.} Let $\mathfrak{d}_k:=\frac12(\partial_{y_{k+1}}-\partial_{y_k})$. Up to a factor $\sqrt2$, this is the derivative in the direction pointing into $\Omega_\ell$ at the face $\{y_k=y_{k+1}\}$. In the ground state substitution below, the contact term on this face has to absorb a boundary term containing $\mathfrak{d}_k\log\tilde u$, and we claim that
\begin{equation}
\label{EqFaceBound}
\mathfrak{d}_k\log\tilde u\leq\frac{1}{\ell_k}\qquad\textnormal{on the face }\{y_k=y_{k+1}\}.
\end{equation}
Indeed, $\mathfrak{d}_k$ is the derivative with respect to $s:=z_{k+1}-z_k$ at fixed $z_k+z_{k+1}$ and fixed $z_j$, $j\notin\{k,k+1\}$, and $s=\ell_k$ on the face. Since $\log\tilde u$ is a sum over the factors of $\tilde u$, we can consider them one by one. Writing $\theta_k=\bar\theta-t$ and $\theta_{k+1}=\bar\theta+t$ with $t=\pi s/(2\Lambda)$, the factor $\cos\theta_k-\cos\theta_{k+1}=2\sin\bar\theta\sin t$ has logarithmic derivative $\frac{\pi}{2\Lambda}\cot t\leq1/s$, since $\cot t\leq1/t$. For $i<k$, the two factors involving $\theta_i$ contribute $\frac{\pi}{2\Lambda}(f(\theta_{k+1})-f(\theta_k))\leq0$, where $f(\theta)=\sin\theta/(\cos\theta_i-\cos\theta)$ is decreasing. Similarly, for $j>k+1$, the two factors involving $\theta_j$ contribute $\frac{\pi}{2\Lambda}(g(\theta_k)-g(\theta_{k+1}))\leq0$, where $g(\theta)=\sin\theta/(\cos\theta-\cos\theta_j)$ is increasing. The remaining factors do not depend on $s$.

\emph{Ground state substitution.} Now write $\varphi=\tilde u w$. By density, we may assume that $w\in C^1(\overline{\Omega_\ell};\mathcal{K})$. We integrate the identity
\[
\abs{\nabla(\tilde u w)}^2=\tilde u^2\abs{\nabla w}^2-\tilde u\Delta\tilde u\abs{w}^2+\nabla\cdot(\tilde u\nabla\tilde u\abs{w}^2)
\]
over $\Omega_\ell$. The boundary terms at $\{y_1=0\}$ and $\{y_N=\ell\}$ vanish by the Neumann conditions. On the face $\{y_k=y_{k+1}\}$, the outward normal derivative is $-\sqrt2\,\mathfrak{d}_k$, and the surface measure is $\sqrt2$ times the one used in \eqref{EqContactForm}. Adding the contact terms, we find
\[
\mathcal{E}_c(\varphi)=e_N(\Lambda)\int_{\Omega_\ell}\abs{\varphi}^2+\int_{\Omega_\ell}\tilde u^2\abs{\nabla w}^2+\sum_{k=1}^{N-1}\int_{\{y_k=y_{k+1}\}}\tilde u^2\left(c_k-2\,\mathfrak{d}_k\log\tilde u\right)\abs{w}^2.
\]
The last sum is non-negative by \eqref{EqFaceBound}, since $c_k=2/\ell_k$, and this proves the first claim. The statement about the Lieb--Liniger model follows by taking $c_k=c$ and $\ell=L$, together with \eqref{EqeNbounds}.
\end{proof}

\subsection{Proof of the lower bound (Theorem~\ref{TheoremLowerBoundSpin})}
\label{SectionLowerBoundProof}
We will adapt the strategy used in the bosonic case \cite{agerskov2022ground}, outlined heuristically in Section~\ref{proofideas}. Summarizing, this strategy used the following steps.
\begin{enumerate}
\item Suppose that the particle configuration is such that $x_{i+1}$ and $x_i$ are neighbouring particles. We fix some radius $R$ bigger than the range $R_0$ of the potential, to be decided later, and apply Dyson's lemma to lower bound the energy contained in the region $|x_{i+1}-x_i|\leq R$ by a Hamiltonian involving $\delta$-potentials at $x_{i+1}-x_i=\pm R$.
\item We throw away the regions with $|x_{i+1}-x_i|<R$ for a lower bound (note the energy density is positive), which causes the $\delta$-potentials to collapse into a single delta (of twice the strength) at $x_{i+1}=x_i$. This gives the Lieb--Liniger model on a smaller interval. Using the known ground state energy, we find that it produces the desired expression. 
\item To really obtain the result of this calculation as a lower bound, we demonstrate that the norm of the wave function in the region that was thrown out is negligible. 
\item However, this strategy only works for a limited number of particles, and so an argument to extend to any particle number is needed. We use boxes and superadditivity caused by the positive potential as in \cite{lieb1998ground}.
\end{enumerate}
What is different for particles with spin? As we can see from Lemma~\ref{LemmaDysonSpin1/2Fermi} (Dyson's lemma for spin-$J$ fermions), the strength of the delta function that we obtain at radius $R$ depends on the spin symmetry of the two particles: antisymmetric spin goes with the even-wave scattering length $a_e$, whereas symmetric spin goes with the odd-wave $a_o$. This means we do not find Lieb and Liniger's spin-independent delta-function by using the above strategy and Lemma~\ref{LemmaDysonSpin1/2Fermi}. This is a problem because the use of the exact expression for the ground state energy of the Lieb--Liniger model was crucial to the proof in \cite{agerskov2022ground}.   

Instead, we proceed as follows. Applying Lemma~\ref{LemmaDysonSpin1/2Fermi} with a single radius $R$ and removing the regions where $|x_{i+1}-x_i|<R$, we obtain a contact interaction between neighbouring particles whose strength depends on their spins: it is $2/(R-a_e)$ on $P^{i,i+1}_A$ and $2/(R-a_o)$ on $P^{i,i+1}_S$ (Lemma~\ref{LemmaLowerBoundDyson+SpaceRemoved}). Two observations allow us to bound the energy of this interaction without an exact solution.
\begin{itemize}
\item If the strengths of the contact interactions do not depend on spin (but possibly on the pair of neighbours), the energy is bounded below by that of free fermions on a longer interval (Lemma~\ref{LemmaLiebLinigerNeumannLowerBound}). This replaces the exact solution of the Lieb--Liniger model.
\item For a fixed wave function, the energy is an affine function of the strengths of the contact interactions. It is therefore the average of two energies of the same wave function: in one, the spin-dependent part of the interaction is doubled on the odd pairs of neighbours and removed from the even ones, and in the other, the other way around. In each of these, the relevant spin projections act on disjoint pairs of spins, so they commute, and the first observation applies to each of their joint eigenspaces. Convexity of $e_N$ then recombines the results into the Lai--Sutherland chain \eqref{spinchainlb} (Lemma~\ref{LemmaDoubling}).
\end{itemize}
The price is an error of relative size $(a_o-a_e)/R$ in the first-order term, which is small if $R$ is large. These ideas are made precise in Section~\ref{lbspn}.

For the third step of the strategy above, we use a simpler argument than in \cite{agerskov2022ground}: a part of the energy pays for the weight of the removed regions, so that no a priori upper bound on the energy is needed (Lemma~\ref{LemmaNormLossImproved} in Section~\ref{secnormloss}). The last step is as in the bosonic case (Section~\ref{lbapn}).

\subsubsection{Lower bound for small particle numbers}
\label{lbspn}
Recall that $v$ is an even, repulsive (positive) interaction, compactly supported in $[-R_0,R_0]$, and that its even- and odd-wave scattering lengths satisfy $a_e\leq a_o\leq R_0$ and $a_o\geq0$. Let $\Psi$ be a function in the domain \eqref{domainfermions} of the energy functional \eqref{functional}, with values in $(\C^{2J+1})^{\otimes N}$. By antisymmetry, it suffices to consider $\Psi$ on the sector $\Omega_L$ of Section~\ref{SecLLNeumann}, where $x_1<\dots<x_N$.

Fix a radius $R>R_0$ such that $\tilde L:=L-(N-1)R>0$. We remove the regions where neighbouring particles are closer than $R$ to each other, and move the particles together to close the gaps. That is, we define $\tilde\psi$ on $\Omega_{\tilde L}$ by
\begin{equation}
\label{defpsi}
\tilde{\psi}(y_1,y_2,\dots,y_N):=\Psi(y_1,y_2+R,y_3+2R,\dots,y_N+(N-1)R).
\end{equation}
We write $\langle\tilde\psi|\tilde\psi\rangle_{L^2}:=N!\int_{\Omega_{\tilde L}}\abs{\tilde\psi}^2$. This is the integral of $\abs{\Psi}^2$ over the set $\{\min_{i<j}\abs{x_i-x_j}>R\}$, so it is at most $1$ if $\Psi$ is normalized. We also let $\mathfrak{r}_i(x):=\min_{j\neq i}\abs{x_i-x_j}$ be the distance from $x_i$ to the nearest other particle. Finally, in integrals over $\Omega_L$, an atom of $v$ at $0$ is counted with half its weight, so that $N!$ times the left side of \eqref{EqDysonLowerBound+SpaceRemoved} below is the energy \eqref{functional} of $\Psi$.

The first step is to apply Dyson's lemma to each pair of neighbours, which gives a contact interaction for $\tilde\psi$. We also keep a part of the energy of $\Psi$ in the removed regions, namely $(1-\alpha)\mathcal{N}(\Psi)$ below, to pay for their weight later (cf.\ \cite{lieb1998ground}, where a small part of the kinetic energy is kept aside when Dyson's lemma is applied).

\begin{lemma}
\label{LemmaLowerBoundDyson+SpaceRemoved}
Let
\begin{equation}
\label{EqDefCkN}
C_k:=\frac{2}{R-a_e}P_A^{k,k+1}+\frac{2}{R-a_o}P_S^{k,k+1},\qquad \mathcal{N}(\Psi):=\int_{\Omega_L}\sum_{i}\mathbbm{1}_{\{\mathfrak{r}_i<R\}}\abs{\partial_i\Psi}^2+\sum_{i<j}v_{ij}\abs{\Psi}^2.
\end{equation}
Then, for every $\alpha\in[0,1]$,
\begin{equation}
\label{EqDysonLowerBound+SpaceRemoved}
\int_{\Omega_L}\sum_{i}\abs{\partial_i\Psi}^2+\sum_{i<j}v_{ij}\abs{\Psi}^2\geq\mathcal{E}_{\alpha C}(\tilde\psi)+(1-\alpha)\mathcal{N}(\Psi).
\end{equation}
\end{lemma}
\begin{proof}
The idea is simple: in the removed region of each pair of neighbours, Dyson's lemma along a line turns the energy into a contact interaction at distance $R$. The bookkeeping below makes sure that each part of the removed regions is used only once.

The left side of \eqref{EqDysonLowerBound+SpaceRemoved} equals $\mathcal{T}+\mathcal{N}(\Psi)$, where $\mathcal{T}:=\int_{\Omega_L}\sum_i\mathbbm{1}_{\{\mathfrak{r}_i\geq R\}}\abs{\partial_i\Psi}^2$ is the kinetic energy of the particles whose nearest neighbour is at distance at least $R$.

\emph{The removed regions.} We sort the removed regions according to the first pair of neighbours (from the left) that is closer than $R$. Let $\Omega^{(1)}:=\Omega_L$ and, for $k=1,\dots,N-1$,
\[
Z_k:=\Omega^{(k)}\cap\{x_{k+1}-x_k<R\},\qquad\Omega^{(k+1)}:=\Omega^{(k)}\cap\{x_{k+1}-x_k>R\}.
\]
The sets $Z_k$ are disjoint, and $\Omega^{(N)}$ is the set where all neighbouring particles are more than $R$ apart. Since $\mathfrak{r}_i\geq R$ for all $i$ on $\Omega^{(N)}$, we have $\mathcal{T}\geq\int_{\Omega^{(N)}}\sum_i\abs{\partial_i\Psi}^2$. On $Z_k$, both $\mathfrak{r}_k$ and $\mathfrak{r}_{k+1}$ are less than $R$, so $\mathcal{N}(\Psi)\geq\sum_k\mathcal{U}_k$, where
\[
\mathcal{U}_k:=\int_{Z_k}\abs{\partial_k\Psi}^2+\abs{\partial_{k+1}\Psi}^2+v_{k,k+1}\abs{\Psi}^2.
\]
As $\alpha\sum_k\mathcal{U}_k+(1-\alpha)\mathcal{N}(\Psi)\leq\mathcal{N}(\Psi)$, the left side of \eqref{EqDysonLowerBound+SpaceRemoved} is at least
\[
\int_{\Omega^{(N)}}\sum_i\abs{\partial_i\Psi}^2+\alpha\sum_{k=1}^{N-1}\mathcal{U}_k+(1-\alpha)\mathcal{N}(\Psi).
\]
The change of variables $y_i=x_i-(i-1)R$ maps $\Omega^{(N)}$ onto $\Omega_{\tilde L}$, and the part $S_k:=\overline{\Omega^{(N)}}\cap\{x_{k+1}-x_k=R\}$ of its boundary onto the face $\{y_k=y_{k+1}\}$. The first term above is therefore the kinetic energy in $\mathcal{E}_{\alpha C}(\tilde\psi)$, and it suffices to show that $\mathcal{U}_k$ is at least the contact term of the $k$-th pair of neighbours in $\mathcal{E}_C(\tilde\psi)$, that is,
\begin{equation}
\label{EqUk}
\mathcal{U}_k\geq\int_{S_k}\braket{\Psi|C_k\Psi},
\end{equation}
where we integrate over the coordinates $x_j$, $j\neq k+1$.

\emph{Dyson's lemma on slices.} By antisymmetry, $P_S^{k,k+1}\Psi=0$ on $\{x_k=x_{k+1}\}$, so we can apply \eqref{EqDysonHalfLine}. We do this twice, moving either particle $k+1$ or particle $k$, and each time we use half of $v_{k,k+1}$. First, fix all coordinates except $x_{k+1}$, and let $\varphi(t):=\Psi(\dots,x_k,x_k+t,\dots)$. In $\Omega^{(k)}$, only the gaps $x_2-x_1,\dots,x_k-x_{k-1}$ are constrained, so the corresponding slice of $Z_k$ is $0<t<\min(R,x_{k+2}-x_k)$, with $x_{N+1}:=L$. If $x_{k+2}-x_k\geq R$, the slice has length $R$, and \eqref{EqDysonHalfLine} shows that the contribution of $\abs{\partial_{k+1}\Psi}^2+\frac12v_{k,k+1}\abs{\Psi}^2$ on this slice is at least
\[
\braket{\varphi(R)\bigg|\left(\frac{1}{R-a_e}P_A^{k,k+1}+\frac{1}{R-a_o}P_S^{k,k+1}\right)\varphi(R)}=\frac12\braket{\varphi(R)|C_k\varphi(R)}.
\]
Second, fix all coordinates except $x_k$, and let $\varphi(t):=\Psi(\dots,x_{k+1}-t,x_{k+1},\dots)$. Now the slice of $Z_k$ is $0<t<\min(R,x_{k+1}-x_{k-1}-R)$ (or $0<t<\min(R,x_2)$ if $k=1$), and the same bound holds for the contribution of $\abs{\partial_k\Psi}^2+\frac12v_{k,k+1}\abs{\Psi}^2$ if the slice has length $R$. At points of $S_k$, both slices have length $R$, since $x_{k+2}-x_k\geq R$ and $x_{k+1}-x_{k-1}-R\geq R$ there. Both families of slices parametrize $S_k$ with the same measure, and all other contributions to $\mathcal{U}_k$ are non-negative. Adding the two bounds gives \eqref{EqUk}.\footnote{For almost every choice of the fixed coordinates, the slices are in $H^1$, and their values at $t=0$ and $t=R$ are given by the traces.}
\end{proof}

It remains to bound $\mathcal{E}_C(\tilde\psi)$ from below. We write $C_k=c_A+\Delta c\,P_S^{k,k+1}$, with $c_A:=2/(R-a_e)$ and $\Delta c:=2/(R-a_o)-2/(R-a_e)\geq0$.\footnote{\label{Footnote:lower bound for matrix potentials}This is where we use that the potential is scalar. For a matrix-valued potential, Lemma~\ref{LemmaLowerBoundDyson+SpaceRemoved} would give $C_k=2(R-\mathsf{A}^{k,k+1})^{-1}$. The doubling in Lemma~\ref{LemmaDoubling} can be adapted to this case, by writing $C_k=c+D_k$ with $c>0$ and $D_k\geq0$, at a cost of second order in $D_k$. The main obstacle is Lemma~\ref{LemmaNormLoss}, which uses a scalar lower bound on the interaction that a matrix-valued potential need not have.} Neighbouring projections $P_S^{k,k+1}$ and $P_S^{k+1,k+2}$ act on overlapping pairs of spins and do not commute. However, the projections $P_S^{k,k+1}$ with odd $k$ act on disjoint pairs of spins and commute, and so do those with even $k$. Since $\mathcal{E}_C(\varphi)$ is affine in $C$ for fixed $\varphi$, we can move the spin-dependent part of the interaction onto the odd or onto the even pairs of neighbours:
\[
\mathcal{E}_C=\frac12\big(\mathcal{E}_{C^{\textnormal{o}}}+\mathcal{E}_{C^{\textnormal{e}}}\big),\qquad\textnormal{where}\qquad C^{\textnormal{o}}_k:=\begin{cases}c_A+2\Delta c\,P_S^{k,k+1}&\textnormal{for odd }k,\\c_A&\textnormal{for even }k,\end{cases}
\]
and $C^{\textnormal{e}}_k$ is defined in the same way, with odd and even exchanged. To compare with the spin chain, let $\bar\epsilon_J:=\lim_{N\to\infty}\epsilon^{\textnormal{LS}}_J$. By Lemma~\ref{LemmaLaiSutherlandFiniteN}, $\epsilon^{\textnormal{LS}}_J\geq\bar\epsilon_J-1/N$, and since an open chain differs from a periodic one by one projection,
\begin{equation}
\label{EqOpenChain}
\sum_{k=1}^{N-1}P_S^{k,k+1}\geq N\bar\epsilon_J-2.
\end{equation}
The following lemma bounds $\mathcal{E}_C$ for contact interactions of this form. We state it for general $\ell$, $c_A$ and $\Delta c$, since we will apply it to a slightly weakened interaction.
 
\begin{lemma}
\label{LemmaDoubling}
Let $c_A>0$, $\Delta c\geq0$ and $C_k:=c_A+\Delta c\,P_S^{k,k+1}$ for $1\leq k\leq N-1$. Let $\ell_A:=2/c_A$, $\ell_S:=2/(c_A+\Delta c)$ and $\beta:=\ell_A(\ell_A-\ell_S)/(2\ell_A-\ell_S)$. Then, for all $\varphi\in H^1(\Omega_\ell;(\C^{2J+1})^{\otimes N})$,
\[
\mathcal{E}_C(\varphi)\geq e_N\Big(\ell+(N-1)\ell_A-\beta(N\bar\epsilon_J-2)\Big)\int_{\Omega_\ell}\abs{\varphi}^2,
\]
and $0\leq(\ell_A-\ell_S)-\beta\leq(\ell_A-\ell_S)^2/\ell_A$.
\end{lemma}
\begin{proof}
The idea is to split the spin space into joint eigenspaces of commuting projections, apply Lemma~\ref{LemmaLiebLinigerNeumannLowerBound} on each of them, and recombine the results using convexity of $e_N$. We first bound $\mathcal{E}_{C^{\textnormal{o}}}$. The projections $P_S^{k,k+1}$ with odd $k$ commute, so the spin space splits into their joint eigenspaces. For each choice $T=(T_k)_{k\textnormal{ odd}}$ with $T_k\in\{P_A^{k,k+1},P_S^{k,k+1}\}$, let $\Pi^T:=\prod_{k\textnormal{ odd}}T_k$. These are orthogonal projections that add up to the identity, and $C^{\textnormal{o}}_k\Pi^T=c_k(T)\Pi^T$ for numbers $c_k(T)\in\{c_A,c_A+2\Delta c\}$. Since the $\Pi^T$ do not depend on $y$, the energy splits accordingly, and we can apply Lemma~\ref{LemmaLiebLinigerNeumannLowerBound} to each part:
\[
\mathcal{E}_{C^{\textnormal{o}}}(\varphi)=\sum_T\mathcal{E}_{c(T)}(\Pi^T\varphi)\geq\sum_Te_N\Big(\ell+\sum_k\frac{2}{c_k(T)}\Big)\int\abs{\Pi^T\varphi}^2\geq e_N(\Lambda_{\textnormal{o}})\int\abs{\varphi}^2.
\]
In the last step, we used convexity of $e_N$ (Jensen's inequality with weights $\int\abs{\Pi^T\varphi}^2/\int\abs{\varphi}^2$), with the averaged length
\[
\Lambda_{\textnormal{o}}:=\ell+\frac{\int\braket{\varphi|\sum_k2(C^{\textnormal{o}}_k)^{-1}\varphi}}{\int\abs{\varphi}^2},\qquad 2(C^{\textnormal{o}}_k)^{-1}=\begin{cases}\ell_A-(\ell_A-\ell')P_S^{k,k+1}&\textnormal{for odd }k,\\\ell_A&\textnormal{for even }k,\end{cases}
\]
where $\ell':=2/(c_A+2\Delta c)$. In the same way, $\mathcal{E}_{C^{\textnormal{e}}}(\varphi)\geq e_N(\Lambda_{\textnormal{e}})\int\abs{\varphi}^2$, with the roles of odd and even $k$ exchanged. Averaging the two bounds and using convexity of $e_N$ once more, we find
\[
\mathcal{E}_C(\varphi)\geq e_N\Big(\tfrac12(\Lambda_{\textnormal{o}}+\Lambda_{\textnormal{e}})\Big)\int\abs{\varphi}^2,\qquad \tfrac12(\Lambda_{\textnormal{o}}+\Lambda_{\textnormal{e}})=\ell+(N-1)\ell_A-\frac{\ell_A-\ell'}{2}\,\frac{\int\braket{\varphi|\sum_kP_S^{k,k+1}\varphi}}{\int\abs{\varphi}^2}.
\]
Note that the full spin chain $\sum_kP_S^{k,k+1}$ appears here, evaluated in the same state $\varphi$, so splitting the interaction into odd and even pairs does not lose its first-order effect. The only loss is that $\ell_A-\ell_S$ is replaced by $\beta$, which is due to the convexity of $c\mapsto2/c$: since $c_A+2\Delta c=4/\ell_S-2/\ell_A$, we have $\ell'=\ell_A\ell_S/(2\ell_A-\ell_S)$ and $\frac12(\ell_A-\ell')=\beta$. The claim now follows from \eqref{EqOpenChain} and monotonicity of $e_N$. The last statement is elementary, since $2\ell_A-\ell_S\geq\ell_A$.
\end{proof}

Let $r:=\max(R_0,a_o,\abs{a_e})$. Combining Lemmas~\ref{LemmaLowerBoundDyson+SpaceRemoved} and~\ref{LemmaDoubling} with Lemma~\ref{LemmaNormLossImproved} (proved in Section~\ref{secnormloss}), we obtain a lower bound that is useful as long as the particle number is not too large. Lemma~\ref{LemmaNormLossImproved} shows that the removed regions cannot lower the energy once the contact interaction is slightly weakened.

\begin{proposition}[Lower bound for small particle numbers]\label{PropositionLowerBoundSmallParticleNumber}
Assume $R\geq2r$, $R>R_0$, $\rho R\leq\frac14$ and $N(\rho R)^2\leq\tau_0:=\frac{3}{64\pi^2}$, and let
\[
\beta:=(a_o-a_e)\frac{R-a_e}{R-a_e+(a_o-a_e)},\qquad\textnormal{so that}\qquad0\leq(a_o-a_e)-\beta\leq\frac{(a_o-a_e)^2}{R-a_e}.
\]
Then
\[
E^{\textnormal{Neumann}}_J(N,L)\geq N\frac{\pi^2}{3}\rho^2\left(1+2\rho\left[a_e+\beta\bar\epsilon_J\right]-\frac{256\pi^2}{3}N(\rho R)^3-\frac{7}{N}\right).
\]
\end{proposition}
\begin{proof}
For $N=1$, the left side is $0$ and the right side is negative, so let $N\geq2$, and let $\Psi$ be the normalized Neumann ground state. The idea is to apply Lemma~\ref{LemmaLowerBoundDyson+SpaceRemoved} with $\alpha$ slightly smaller than $1$. The remaining part $(1-\alpha)\mathcal{N}(\Psi)$ of the energy then pays for the weight of the removed regions (Lemma~\ref{LemmaNormLossImproved}), and weakening the contact interaction by the factor $\alpha$ only causes a small error.

To make this precise, put $\delta:=\frac{32\pi^2}{3}N(\rho R)^2\leq\frac12$ and $\alpha:=1-\delta$. Lemma~\ref{LemmaDoubling}, applied with $\alpha c_A$ and $\alpha\Delta c$ in place of $c_A$ and $\Delta c$, gives the following bound, since this replaces $\ell_A=R-a_e$, $\ell_S=R-a_o$ and $\beta$ by $\ell_A/\alpha$, $\ell_S/\alpha$ and $\beta/\alpha$:
\[
N!\,\mathcal{E}_{\alpha C}(\tilde\psi)\geq e_N(\Lambda_\alpha)\langle\tilde\psi|\tilde\psi\rangle_{L^2},\qquad\Lambda_\alpha:=\tilde L+\frac{X}{\alpha},\qquad X:=(N-1)(R-a_e)-\beta(N\bar\epsilon_J-2).
\]
Since $\abs{a_e}\leq R/2$, $0\leq\beta\leq a_o-a_e\leq R$ and $0\leq\bar\epsilon_J\leq1$, we have $(N-1)(R-a_o)\leq X\leq2NR$. In particular, $X\geq0$, so $\Lambda_\alpha\geq\tilde L\geq L/2$, and \eqref{EqeNbounds} gives $e_N(\Lambda_\alpha)\leq\frac{4\pi^2}{3}N\rho^2=\delta/(8R^2)$. Hence Lemma~\ref{LemmaNormLossImproved} applies, and $E^{\textnormal{Neumann}}_J(N,L)\geq e_N(\Lambda_\alpha)$.

Finally, we bound $e_N(\Lambda_\alpha)$ from below. Since $1/\alpha\leq1+2\delta$ and $\tilde L+(N-1)(R-a_e)=L-(N-1)a_e$, we have
\[
\Lambda_\alpha\leq\tilde L+X+4\delta NR=L(1-\kappa),\qquad \kappa:=\rho\Big[\big(1-\tfrac1N\big)a_e+\beta\big(\bar\epsilon_J-\tfrac2N\big)\Big]-4\delta\rho R.
\]
Moreover, $\abs{\kappa}\leq\frac32\rho R+4\delta\rho R\leq\frac78$. Monotonicity of $e_N$, \eqref{EqeNbounds} and $(1-\kappa)^{-2}\geq1+2\kappa$ therefore give
\[
\begin{aligned}
E^{\textnormal{Neumann}}_J(N,L)\geq e_N(L(1-\kappa))&\geq N\frac{\pi^2}{3}\rho^2(1+2\kappa)\left(1-\frac{3}{2N}\right)\\
&\geq N\frac{\pi^2}{3}\rho^2\left(1+2\rho\left[a_e+\beta\bar\epsilon_J\right]-8\delta\rho R-\frac{7}{N}\right),
\end{aligned}
\]
which is the claim, since $8\delta\rho R=\frac{256\pi^2}{3}N(\rho R)^3$. (The terms of order $1/N$ are bounded using $\abs{\kappa}\leq\frac78$, $\abs{a_e}+2\beta\leq\frac52R$ and $\rho R\leq\frac14$.)
\end{proof}

\subsubsection{The removed regions (proof of Lemma~\ref{LemmaNormLossImproved})}
\label{secnormloss}
In this section, we prove Lemma~\ref{LemmaNormLossImproved}, which was used in the proof of Proposition~\ref{PropositionLowerBoundSmallParticleNumber}. The first ingredient bounds the weight of the removed regions by the energy in them. It holds for functions with any symmetry, and it corrects \cite[Lemma~20]{agerskov2022ground}: the proof there uses the bound $\int v_{ij}\abs{\Psi}^2\geq\min\abs{\Psi}^2\int v$ on intervals that need not contain the support of $v_{ij}$, for instance near the ends of $[0,\ell]$ or near a third particle. For symmetric functions, Lemma~\ref{LemmaNormLoss} gives \cite[Lemma~20]{agerskov2022ground} with its constants multiplied by $2$. Together with the remark on \cite[Lemma~11]{agerskov2022ground} in the proof of Lemma~\ref{LemmaE1BoundSpin1/2}, this shows that the results of \cite{agerskov2022ground} are unaffected.

\begin{lemma}[cf.\ Lemma 20 in \cite{agerskov2022ground}]
\label{LemmaNormLoss}
Let $\mathcal{K}$ be a Hilbert space, $\ell\geq2R_0$ and $R>R_0$. For all $\Psi\in H^1([0,\ell]^N;\mathcal{K})$,
\begin{equation}
\label{eqlemmanormloss}
\int_{\{\min_{i<j}\abs{x_i-x_j}<R\}}\abs{\Psi}^2\leq8R^2\sum_{i}\int\mathbbm{1}_{\{\mathfrak{r}_i<R\}}\abs{\partial_i\Psi}^2+2R(R-a_e)\sum_{i<j}\int v_{ij}\abs{\Psi}^2.
\end{equation}
\end{lemma}
\begin{proof}
We use a one-dimensional estimate in each variable $x_i$ separately, on intervals around the other particles. To make sure that each interval contains (half of) the support of the potential, we use the full intervals $(x_j-R,x_j+R)$ rather than the parts of them in which $x_j$ is the nearest particle, and we control their overlaps.

Fix $i$ and the coordinates $x_j$, $j\neq i$. Then $\{x_i\,\vert\,\mathfrak{r}_i<R\}$ is the union of the intervals $I_j:=(x_j-R,x_j+R)\cap[0,\ell]$, $j\neq i$. If three of the untruncated intervals share a point, the one with the middle centre lies in the union of the other two. Discarding such intervals one at a time, we obtain a subfamily $(I_j)_{j\in\mathcal{J}}$ with the same union, which covers every point at most twice. Each $I_j$ has length at most $2R$, and since $\ell\geq2R_0$, it contains $[x_j,x_j+R_0]$ or $[x_j-R_0,x_j]$. As $v$ is even, it follows that $\int_{I_j}v(x_i-x_j)dx_i\geq\frac12\int v\geq2/(R-a_e)$, where the last inequality follows by taking $\psi\equiv1$ in \eqref{dyson1}.

Now write $\varphi(t)$ for $\Psi$ with $x_i=t$. For $t,t'\in I_j$, we have $\abs{\varphi(t)}\leq\abs{\varphi(t')}+(2R)^{1/2}(\int_{I_j}\abs{\varphi'}^2)^{1/2}$, so $\max_{I_j}\abs{\varphi}^2\leq2\min_{I_j}\abs{\varphi}^2+4R\int_{I_j}\abs{\varphi'}^2$. Moreover, $\min_{I_j}\abs{\varphi}^2\leq\frac{R-a_e}{2}\int_{I_j}v_{ij}\abs{\varphi}^2$ by positivity of $v$. Therefore,
\[
\int_{I_j}\abs{\varphi}^2\leq2R\max_{I_j}\abs{\varphi}^2\leq8R^2\int_{I_j}\abs{\varphi'}^2+2R(R-a_e)\int_{I_j}v_{ij}\abs{\varphi}^2.
\]
We sum over $j\in\mathcal{J}$, integrate over the other coordinates and sum over $i$. Since the intervals $I_j$ cover $\{x_i\,\vert\,\mathfrak{r}_i<R\}$, and cover each point at most twice, this gives
\[
\sum_i\int\mathbbm{1}_{\{\mathfrak{r}_i<R\}}\abs{\Psi}^2\leq2\times\big(\textnormal{right side of }\eqref{eqlemmanormloss}\big).
\]
Finally, if $\min_{i<j}\abs{x_i-x_j}<R$, then both particles of a closest pair have $\mathfrak{r}_i<R$, so the left side is at least twice the left side of \eqref{eqlemmanormloss}.
\end{proof}

\begin{lemma}
\label{LemmaNormLossImproved}
Let $N\geq2$, $R\geq2r$, $R>R_0$ and $\tilde L>0$, let $\Psi$ be a normalized function in the domain \eqref{domainfermions}, and let $\alpha\in[0,1]$. If $N!\,\mathcal{E}_{\alpha C}(\tilde\psi)\geq e\,\langle\tilde\psi|\tilde\psi\rangle_{L^2}$ for some $e\leq(1-\alpha)/(8R^2)$, then the energy \eqref{functional} of $\Psi$ is at least $e$.
\end{lemma}
\begin{proof}
By antisymmetry, the energy of $\Psi$ is $N!$ times the left side of \eqref{EqDysonLowerBound+SpaceRemoved}. By Lemma~\ref{LemmaLowerBoundDyson+SpaceRemoved}, it is therefore at least $N!\,\mathcal{E}_{\alpha C}(\tilde\psi)+(1-\alpha)N!\,\mathcal{N}(\Psi)$, with $\mathcal{N}$ defined in \eqref{EqDefCkN}. Since $L>(N-1)R\geq2R_0$, and $2R(R-a_e)\leq3R^2$ because $\abs{a_e}\leq R/2$, Lemma~\ref{LemmaNormLoss} gives $N!\,\mathcal{N}(\Psi)\geq(1-\langle\tilde\psi|\tilde\psi\rangle_{L^2})/(8R^2)$. Hence the energy of $\Psi$ is at least
\[
e\,\langle\tilde\psi|\tilde\psi\rangle_{L^2}+\frac{1-\alpha}{8R^2}\big(1-\langle\tilde\psi|\tilde\psi\rangle_{L^2}\big)\geq e,
\]
since $0\leq\langle\tilde\psi|\tilde\psi\rangle_{L^2}\leq1$ and $e\leq(1-\alpha)/(8R^2)$.
\end{proof}
In other words, moving weight into the removed regions costs more energy than it can save. In \cite[Section~3.2]{agerskov2022ground}, an a priori upper bound on the energy was used instead, to show that $\langle\tilde\psi|\tilde\psi\rangle_{L^2}$ is close to $1$. Lemma~\ref{LemmaNormLossImproved} avoids this, so the proof of the lower bound does not use the upper bound.

\subsubsection{Lower bound for arbitrary particle numbers (proof of Theorem~\ref{TheoremLowerBoundSpin})}
\label{lbapn}
In this subsection, we extend Proposition~\ref{PropositionLowerBoundSmallParticleNumber} to arbitrary particle numbers and prove Theorem~\ref{TheoremLowerBoundSpin}. Proposition~\ref{PropositionLowerBoundSmallParticleNumber} is only useful if $N(\rho R)^3$ is small. For larger $N$, we divide $[0,L]$ into boxes that contain about $(\rho R)^{-3/2}$ particles each. Since the potential is positive, we may drop the interactions between particles in different boxes for a lower bound. The particles may, however, distribute themselves unevenly over the boxes. As in \cite{agerskov2022ground}, we handle this with a chemical potential $\mu$ (a Legendre transformation in the particle number), together with superadditivity of the energy.

From now on, let $R:=r(\rho r)^{-1/5}$, so that $\rho R=(\rho r)^{4/5}$ and $(\rho R)^{3/2}=(\rho r)^{6/5}$, and assume that $\rho r$ is sufficiently small. This choice balances the two main errors, namely those of Proposition~\ref{PropositionLowerBoundSmallParticleNumber} for boxes with about $(\rho R)^{-3/2}$ particles, and the cost of replacing $a_o-a_e$ by $\beta$: both are of relative size $(\rho R)^{3/2}$ (see the proof below). With $\beta$ as in Proposition~\ref{PropositionLowerBoundSmallParticleNumber}, put $a_\sharp:=a_e+\beta\bar\epsilon_J$, so that $\abs{a_\sharp}\leq3r$. For $n$ particles in an interval of length $\ell$, the main term in Proposition~\ref{PropositionLowerBoundSmallParticleNumber} is $\ell g(n/\ell)$, where
\[
g(t):=\frac{\pi^2}{3}t^3(1+2ta_\sharp).
\]
For $\rho r$ sufficiently small, $g$ is convex on $[0,5\rho]$. We choose the chemical potential
\[
\mu:=g'(\rho)=\pi^2\rho^2\left(1+\frac83\rho a_\sharp\right),
\]
so that, by convexity, $g(t)-\mu t\geq g(\rho)-\mu\rho$ for all $t\in[0,5\rho]$. The following corollary replaces \cite[Corollary~25]{agerskov2022ground}. It applies Proposition~\ref{PropositionLowerBoundSmallParticleNumber} in a box of the size described above, with at most five times the average density.

\begin{corollary}
\label{CorollaryLowerBoundSpecN}
Let $\frac12(\rho R)^{-3/2}\leq\rho\ell\leq(\rho R)^{-3/2}$. Then, for all integers $0\leq n\leq5\rho\ell$,
\[
E^{\textnormal{Neumann}}_J(n,\ell)\geq\ell g(n/\ell)-\textnormal{const. }\ell\rho^3(\rho R)^{3/2}.
\]
\end{corollary}
\begin{proof}
Let $t:=n/\ell$, so that $t\leq5\rho$. For $n\leq1$, the left side is $0$, and $\ell g(t)\leq\textnormal{const. }\rho^2\leq\textnormal{const. }\ell\rho^3(\rho R)^{3/2}$. For $n\geq2$, we apply Proposition~\ref{PropositionLowerBoundSmallParticleNumber} to $n$ particles in an interval of length $\ell$, at density $t$. This is allowed, since $tR\leq5\rho R\leq\frac14$ and $n(tR)^2\leq125\rho\ell(\rho R)^2\leq125(\rho R)^{1/2}\leq\tau_0$. The main term is $\frac{\pi^2}{3}nt^2(1+2ta_\sharp)=\ell g(t)$, and the error terms satisfy
\[
\begin{aligned}
nt^2\cdot n^{-1}&\leq50\,\ell\rho^3(\rho R)^{3/2},\\
nt^2\cdot n(tR)^3&\leq5^7\ell\rho^3(\rho\ell)(\rho R)^3\leq5^7\ell\rho^3(\rho R)^{3/2}.
\end{aligned}
\]
\end{proof}

\begin{proof}[Proof of Theorem~\ref{TheoremLowerBoundSpin}]
\emph{Reduction.} By Proposition~\ref{PropositionLowerBoundSmallParticleNumber}, $2\rho(a_o-a_e-\beta)\bar\epsilon_J\leq2\rho(a_o-a_e)^2/(R-a_e)\leq16\rho r^2/R=16(\rho R)^{3/2}$, and $\epsilon^{\textnormal{LS}}_J\leq\bar\epsilon_J+1/N$ by Lemma~\ref{LemmaLaiSutherlandFiniteN}. It therefore suffices to prove
\begin{equation}
\label{EqGoal}
E^{\textnormal{Neumann}}_J(N,L)\geq N\frac{\pi^2}{3}\rho^2\left(1+2\rho a_\sharp-\textnormal{const. }\left((\rho R)^{3/2}+N^{-1}\right)\right).
\end{equation}

\emph{Few particles.} If $N\leq(\rho R)^{-3/2}$, then $N(\rho R)^2\leq(\rho R)^{1/2}\leq\tau_0$ and $N(\rho R)^3\leq(\rho R)^{3/2}$, so \eqref{EqGoal} follows from Proposition~\ref{PropositionLowerBoundSmallParticleNumber}.

\emph{Many particles.} Now let $N>(\rho R)^{-3/2}$. We split $[0,L]$ into $M:=\ceil{N(\rho R)^{3/2}}$ intervals of length $\ell:=L/M$, so that $\frac12(\rho R)^{-3/2}\leq\rho\ell\leq(\rho R)^{-3/2}$, and we write $E(n):=E^{\textnormal{Neumann}}_J(n,\ell)$, with $E(0)=0$. Let $\Psi$ be the normalized Neumann ground state for $N$ particles in $[0,L]$. We split $[0,L]^N$ according to which interval each particle lies in. On each piece, the restriction of $\Psi$ is $H^1$ without boundary conditions at the ends of the intervals, interactions between particles in different intervals are non-negative and may be dropped, and antisymmetry within each interval is inherited. Hence,
\begin{equation}
\label{EqBoxes}
E^{\textnormal{Neumann}}_J(N,L)\geq\min_{n_1+\dots+n_M=N}\sum_{m=1}^ME(n_m)=\mu N+\min_{n_1+\dots+n_M=N}\sum_{m=1}^M\big(E(n_m)-\mu n_m\big).
\end{equation}
We claim that, for all $n\geq0$,
\begin{equation}
\label{EqBoxClaim}
E(n)-\mu n\geq\ell\big(g(\rho)-\mu\rho\big)-\textnormal{const. }\ell\rho^3(\rho R)^{3/2}.
\end{equation}
Given this, \eqref{EqBoxes}, $M\ell=L$ and $N=\rho L$ give $E^{\textnormal{Neumann}}_J(N,L)\geq Lg(\rho)-\textnormal{const. }L\rho^3(\rho R)^{3/2}$, which is \eqref{EqGoal}.

For $0\leq n\leq5\rho\ell$, the claim \eqref{EqBoxClaim} follows from Corollary~\ref{CorollaryLowerBoundSpecN}, since $g(t)-\mu t\geq g(\rho)-\mu\rho$ for $t\in[0,5\rho]$. Larger $n$ can be reduced to this case by superadditivity. Since $v\geq0$, we have $E(n_1+n_2)\geq E(n_1)+E(n_2)$: drop the interactions between two groups of particles, and note that a wave function that is antisymmetric in all particles is antisymmetric within each group. Let $K:=\ceil{4\rho\ell}\leq5\rho\ell$. Since $g(t)-\mu t$ is increasing on $[\rho,5\rho]$ by convexity of $g$, Corollary~\ref{CorollaryLowerBoundSpecN} gives
\[
\begin{aligned}
E(K)-\mu K&\geq\ell\big(g(4\rho)-4\rho\mu\big)-\textnormal{const. }\ell\rho^3(\rho R)^{3/2}\\
&=\ell\rho^3\left[\pi^2\left(\frac{52}{3}+160\rho a_\sharp\right)-\textnormal{const. }(\rho R)^{3/2}\right]\geq0.
\end{aligned}
\]
Writing $n=qK+n'$ with integers $q\geq0$ and $0\leq n'<K$, superadditivity then gives
\[
E(n)-\mu n\geq q\big(E(K)-\mu K\big)+E(n')-\mu n'\geq E(n')-\mu n',
\]
so \eqref{EqBoxClaim} holds for all $n\geq0$.
\end{proof}

\subsection{\texorpdfstring{Lower bound for the Lieb--Liniger--Heisenberg model}{Lower bound for the Lieb--Liniger--Heisenberg model}}
\label{SecLLHLower}
In this section, we prove the lower bound in Theorem~\ref{TheoremLLHUpperBound} for Neumann boundary conditions, which implies it for Robin boundary conditions with non-negative parameter. By Section~\ref{resultsLLH}, the energy of $\Psi$ in the domain \eqref{domainLLH} is $N!\,\mathcal{E}_C(\psi)$, where $\psi$ is the restriction of $\Psi$ to $\Omega_L$, $\mathcal{E}_C$ is defined in \eqref{EqContactForm}, and $C_k:=c'P_A^{k,k+1}+cP_S^{k,k+1}$. In other words, the model is a contact interaction of the same type as the one we obtained from Dyson's lemma in Lemma~\ref{LemmaLowerBoundDyson+SpaceRemoved}, and the proof of Theorem~\ref{TheoremLowerBoundSpin} applies with a few changes. We write $E^{\textnormal{LLH}}_{\textnormal{Neumann}}(N,L)$ for the ground state energy.

\emph{The case $c\leq c'$.} Then $C_k\geq c$, and Lemma~\ref{LemmaLiebLinigerNeumannLowerBound} gives
\[
E^{\textnormal{LLH}}_{\textnormal{Neumann}}(N,L)\geq e_N\big(L+2(N-1)/c\big)\geq N\frac{\pi^2}{3}\rho^2\left(1-\frac{4\rho}{c}-\frac{3}{2N}\right).
\]
This is the claim, since $\epsilon^{\textnormal{LLH}}(c,c')=-2/c$ in this case.

\emph{The case $c>c'$.} Put $a_A:=-2/c'$ and $a_S:=-2/c$, which are the scattering lengths of the potentials $2c'\delta$ and $2c\delta$. Then $a_A<a_S<0$ and
\[
\epsilon^{\textnormal{LLH}}(c,c')=a_A+(a_S-a_A)\epsilon^{\textnormal{LS}}_{1/2},
\]
so $a_A$ and $a_S$ play the roles of $a_e$ and $a_o$. The analogue of \eqref{EqDysonHalfLine} is the following bound. For $X\in\{A,S\}$, $1\leq k\leq N-1$ and $\varphi\in H^1([0,R];(\C^2)^{\otimes N})$, with $P_X=P_X^{k,k+1}$,
\begin{equation}
\label{EqRobinHalfLine}
\int_0^R\abs{P_X\varphi'}^2+\frac{\abs{P_X\varphi(0)}^2}{\abs{a_X}}\geq\frac{\abs{P_X\varphi(R)}^2}{R-a_X}.
\end{equation}
This follows from $\abs{P_X\varphi(R)}\leq\abs{P_X\varphi(0)}+\sqrt{R}\,\norm{P_X\varphi'}_{L^2}$ and the Cauchy--Schwarz inequality. Unlike \eqref{EqDysonHalfLine}, it requires no condition at $0$, so the antisymmetry used in the proof of Lemma~\ref{LemmaLowerBoundDyson+SpaceRemoved} is not needed. With $a_e$ and $a_o$ replaced by $a_A$ and $a_S$, and with $R_0=0$ (so that $r=\abs{a_A}$), the proof of Theorem~\ref{TheoremLowerBoundSpin} in Section~\ref{SectionLowerBoundProof} applies with the following three changes.

\emph{The reduction (Lemma~\ref{LemmaLowerBoundDyson+SpaceRemoved}).} The terms $v_{ij}$ are replaced by the contact terms of $\mathcal{E}_C$, both on the left side of \eqref{EqDysonLowerBound+SpaceRemoved} and in the definition \eqref{EqDefCkN} of $\mathcal{N}$. In the proof, the contact term of the $k$-th pair of neighbours is split equally between the two families of slices, and \eqref{EqDysonHalfLine} is replaced by \eqref{EqRobinHalfLine}, applied to $P_A^{k,k+1}\varphi$ and $P_S^{k,k+1}\varphi$. This gives \eqref{EqDysonLowerBound+SpaceRemoved} with $C_k$ replaced by $\frac{2}{R-a_A}P_A^{k,k+1}+\frac{2}{R-a_S}P_S^{k,k+1}$, which has the form required in Lemma~\ref{LemmaDoubling}.

\emph{The removed regions (Lemma~\ref{LemmaNormLossImproved}).} In the proof of Lemma~\ref{LemmaNormLossImproved}, Lemma~\ref{LemmaNormLoss} is replaced by the following simpler estimate. Fix $k$ and all coordinates except $x_{k+1}$, and let $\varphi(t):=\psi(\dots,x_k,x_k+t,\dots)$ for $0<t<s:=\min(R,x_{k+2}-x_k)$. Then
\[
\int_0^s\abs{\varphi}^2\leq2s\abs{\varphi(0)}^2+s^2\int_0^s\abs{\varphi'}^2,\qquad\abs{\varphi(0)}^2\leq\frac{\abs{a_A}}{2}\braket{\varphi(0)|C_k\varphi(0)},
\]
where the second inequality holds since $C_k\geq c'=2/\abs{a_A}$. Integrating and summing over $k$, we find $1-\langle\tilde\psi|\tilde\psi\rangle_{L^2}\leq R^2N!\,\mathcal{N}(\psi)\leq8R^2N!\,\mathcal{N}(\psi)$ for $R\geq\abs{a_A}$. This is the bound that the proof of Lemma~\ref{LemmaNormLossImproved} takes from Lemma~\ref{LemmaNormLoss}, so Lemma~\ref{LemmaNormLossImproved} and Proposition~\ref{PropositionLowerBoundSmallParticleNumber} hold as stated.

\emph{Arbitrary particle numbers.} Dividing into intervals is compatible with the labelling of the spins by position. On the part of $\Omega_L$ where the $m$-th interval contains $n_m$ particles, these particles occupy consecutive positions, so the spin space factorizes over the intervals, and the faces between particles in different intervals have measure zero. Superadditivity is not available, since the particles of two groups may alternate, but in the proof of Theorem~\ref{TheoremLowerBoundSpin} it was only used for $n>5\rho\ell$. For $n\geq K$, we instead use that $C_k\geq c'$. By Lemma~\ref{LemmaLiebLinigerNeumannLowerBound} and \eqref{EqeNbounds},
\[
\frac{E(n)}{n}\geq\frac{e_n\big(\ell+(n-1)\abs{a_A}\big)}{n}\geq\frac{\pi^2}{3}\left(\frac{n}{\ell+n\abs{a_A}}\right)^2\left(1-\frac{3}{2n}\right)\geq\mu,
\]
since, for $n\geq K\geq4\rho\ell$, we have $n\geq3$ and $n/(\ell+n\abs{a_A})\geq3\rho$ (as $\rho\abs{a_A}\leq\frac1{12}$), while $\mu\leq\pi^2\rho^2$ because $a_\sharp\leq a_S<0$. Hence $E(n)-\mu n\geq0\geq\ell(g(\rho)-\mu\rho)$ for $n\geq K$, which gives \eqref{EqBoxClaim}.

With these changes, the proof of Theorem~\ref{TheoremLowerBoundSpin} gives
\[
E^{\textnormal{LLH}}_{\textnormal{Neumann}}(N,L)\geq N\frac{\pi^2}{3}\rho^2\left(1+2\rho\left[a_A+(a_S-a_A)\epsilon^{\textnormal{LS}}_{1/2}\right]-\textnormal{const. }\left((\rho\abs{a_A})^{6/5}+N^{-1}\right)\right),
\]
which is the lower bound in Theorem~\ref{TheoremLLHUpperBound}, since $\abs{a_A}=2/\min(c,c')$.

\section*{Declarations}
 \textbf{Competing interests.} The authors have no competing interests to declare.

\noindent{\textbf{Use of AI.} The AI assistant Claude (Anthropic) was used to correct an error in the proof of the lower bound in an earlier version of this paper. The key elements of the new proof, namely the bound for contact interactions (Lemma~\ref{LemmaLiebLinigerNeumannLowerBound}), the doubling argument (Lemma~\ref{LemmaDoubling}) and the treatment of the removed regions (Lemmas~\ref{LemmaNormLoss} and~\ref{LemmaNormLossImproved}), were found by Claude. Claude also found and corrected smaller errors throughout the paper.}

\section*{Acknowledgments}
This work was supported by the Villum Centre of Excellence for the Mathematics of Quantum Theory (QMATH, grant 10059) and by grant VIL73411 from Villum Fonden. JA was partially supported by the Novo Nordisk Foundation, Grant number NNF22SA0081175, NNF Quantum Computing Programme. RR was partially supported by the European Research Council (ERC) under the European Union's Horizon 2020 research and innovation programme (ERC CoG UniCoSM, Grant Agreement No. 724939). RR thanks the University of Copenhagen and the University of Cambridge for hospitality during visits, and the Royal Society for supporting the research visit to Cambridge (AL$\backslash$24100022).  RR is a member of the \textit{Gruppo Nazionale per la Fisica Matematica} (GNFM) of the \textit{Istituto Nazionale di
Alta Matematica Francesco Severi} (INdAM).

\begin{appendices}
\section{Facts about the scattering length matrix}
\label{AppendixA}
Consider an even, positive semidefinite measure $V$ that takes values in the space of complex $(2J+1)^2\times (2J+1)^2$ matrices. Assume that $V$ commutes with $P_S$ and $P_A$ (meaning $V$ leaves the symmetric and antisymmetric spin spaces invariant), and assume that $V$ is supported in $[-R_0,R_0]$ for some $R_0\geq 0$.
 
Let $R>R_0$ and define the scattering energy functional $\mathcal{E}_{\textnormal{scattering}}:H^1([-R,R];\C^{2J+1}\otimes \C^{2J+1})\to [0,\infty]$ by
	\begin{equation}
	\mathcal{E}_{\textnormal{scattering}}(\psi)=\int^R_{-R} 2\braket{\partial\psi,\partial\psi}+\braket{\psi,V\psi}. 
	\end{equation}
To state Definition~\ref{DefinitionScatteringLengthMatrix}, we need the following facts. 
\begin{lemma}\label{LemmaConvexityScatteringFunctional}
    For $\chi\in\mathbb{C}^{2J+1}\otimes\mathbb{C}^{2J+1}$, denote by $\mathcal{E}_{\chi}$ the functional $\mathcal{E}_{\textnormal{scattering}}$ restricted to functions $\psi$ with boundary conditions $\psi(R)=\chi$ and $\psi(-R)=(P_A-P_S)\chi$. Then, $\mathcal{E}_{\chi}$ is strictly convex.
\end{lemma}
\begin{proof}
    Note that $\mathcal{E}_{\textnormal{scattering}}$ is quadratic, and that it is non-negative by the assumptions on $V$. Hence, given $\psi_1,\psi_2$ in the domain of $\mathcal{E}_\chi$ and $\alpha\in[0,1]$, we can write \begin{equation}\label{EqConvexityScatteringFunctional}
        \begin{aligned}
            \mathcal{E}_\chi(\alpha \psi_1+(1-\alpha)\psi_2)&=\alpha\mathcal{E}_\chi(\psi_1)+(1-\alpha)\mathcal{E}_\chi(\psi_2)-\alpha(1-\alpha)\mathcal{E}_{\textnormal{scattering}}(\psi_1-\psi_2)\\
            &\leq \alpha\mathcal{E}_\chi(\psi_1)+(1-\alpha)\mathcal{E}_\chi(\psi_2).
        \end{aligned}
    \end{equation}
    Now assume that $\psi_1\neq \psi_2$. Since they satisfy the same boundary condition, they cannot differ by a constant, and thus the kinetic energy in $\mathcal{E}_{\textnormal{scattering}}(\psi_1-\psi_2)$ is strictly positive. By $V$ being positive semidefinite, we conclude that $\mathcal{E}_\chi$ is strictly convex,
    \[
        \begin{aligned}
            \mathcal{E}_\chi(\alpha \psi_1+(1-\alpha)\psi_2)< \alpha\mathcal{E}_\chi(\psi_1)+(1-\alpha)\mathcal{E}_\chi(\psi_2).
        \end{aligned}
    \]
\end{proof}
A minimizer of $\mathcal{E}_\chi$ exists by the direct method, since $\mathcal{E}_\chi$ controls the $H^1$-norm and is weakly lower semicontinuous (by Fatou's lemma, as $H^1([-R,R])$ embeds compactly into $C([-R,R])$). Together with Lemma~\ref{LemmaConvexityScatteringFunctional}, this gives the existence and uniqueness of the scattering solution claimed in Definition~\ref{DefinitionScatteringLengthMatrix}.
\begin{corollary}
    There is a unique minimizer of $\mathcal{E}_{\textnormal{scattering}}$ with boundary condition $\psi(R)=\chi$ and $\psi(-R)=(P_A-P_S)\chi$.
\end{corollary}
Given this uniqueness, we can define a map $F_0:\C^{2J+1}\otimes\C^{2J+1}\to H^1([-R,R];\C^{2J+1}\otimes\C^{2J+1})$ that maps a boundary spin-state $\chi$ to the corresponding minimizer of $\mathcal{E}_\chi$. Note that any minimizer is the unique solution of the Euler--Lagrange equation
\begin{equation}
\label{eleqn}
	-\psi''(x)+\frac{1}{2}V(x)\psi(x)=0,
\end{equation}
with boundary conditions $\psi(R)=\chi$ and $\psi(-R)=(P_A-P_S)\chi$. Thus, linearity of \eqref{eleqn} and the boundary condition imply linearity of $F_0$. This means that $F_0$ can be identified with an element in $H^1([-R,R];\operatorname{Mat}((2J+1)^2,(2J+1)^2))$. This proves the first part of the following proposition.
\begin{proposition}\label{CorollaryUniqueF0}
    The scattering solution $F_0\in H^1([-R,R];\operatorname{Mat}((2J+1)^2,(2J+1)^2))$ is the unique solution of \eqref{EqEvenScatteringEquation} on $[-R,R]$ with boundary conditions $F_0(R)=I$ and $F_0(-R)=P_A-P_S$. Moreover, $F_0$ satisfies \eqref{EqScatteringSolution} and \eqref{EqScatteringEnergy}. 
\end{proposition}
\begin{proof}
   The first part of the proposition follows from the analysis above. Since $V$ is even and commutes with $P_A-P_S$, the function $(P_A-P_S)F_0(-x)$ solves \eqref{EqEvenScatteringEquation} with the same boundary conditions as $F_0$, so $F_0(-x)=(P_A-P_S)F_0(x)$ by uniqueness. Integration by parts, \eqref{EqEvenScatteringEquation} and this symmetry give, for $R_0<s\leq R$,
\begin{equation}
\label{EqEnergyIdentityF0}
\int_{-s}^{s}2(F_0')^*F_0'+F_0^*VF_0=2\big[F_0^*F_0'\big]_{-s}^{s}=4F_0(s)^*F_0'(s).
\end{equation}
For $R_0\leq x\leq R$ (outside the range of $V$), $F_0$ is harmonic, so $F_0(x)=I+B(x-R)$ with $B:=F_0'(R)$, and \eqref{EqEnergyIdentityF0} with $s=R$ shows that $B$ is Hermitian and non-negative. On the range of $B$, put $\mathsf{A}:=R-B^{-1}$. Then $F_0(x)=(R-\mathsf{A})^{-1}(x-\mathsf{A})$ for $R_0\leq x\leq R$, the symmetry gives \eqref{EqScatteringSolution} for $-R\leq x\leq-R_0$, and \eqref{EqEnergyIdentityF0} with $s=R$ is \eqref{EqScatteringEnergy}. If $B\chi=0$, then \eqref{EqEnergyIdentityF0} shows that $F_0\chi=\chi$ is constant, which is the convention $\mathsf{A}\chi=-\infty\chi$ of Remark~\ref{remarkinfinity}.

Finally, $\mathsf{A}$ does not depend on $R$. Let $R_0<R<R'$, and let $\tilde F$ be the solution on $[-R',R']$. Then $\tilde F(R)$ is invertible: if $\tilde F(R)\chi=0$, then \eqref{EqEnergyIdentityF0}, applied to $\tilde F$ with $s=R$, gives $\tilde F\chi=0$ on $[-R,R]$, hence on $[R_0,R']$ by linearity, which contradicts $\tilde F(R')\chi=\chi$. By the symmetry of $\tilde F$ and uniqueness, $F_0=\tilde F\,\tilde F(R)^{-1}$ on $[-R,R]$, and inserting \eqref{EqScatteringSolution} for $\tilde F$ shows that $F_0$ has the same matrix $\mathsf{A}$.
\end{proof}

The following result is used in the proof of the upper bound.
\begin{lemma}\label{LemmaScatteringSolution hard core pointwise bound}
    Consider the set-up in Definition~\ref{DefinitionScatteringLengthMatrix}---that is, let $F$ be the scattering solution and $\mathsf{A}$ the scattering length matrix of some potential $V$. Furthermore, let $F_{\rm hc}$ be the scattering solution for the hard-core potential with radius $R_0$. Then, for all $\xi\in \C^{2J+1}\otimes\C^{2J+1}$ and $x\in[-R,R]$,
    \[
    \braket{F_{\rm hc}(x) \xi |F_{\rm hc}(x) \xi }\leq \braket{F(x)\xi|F(x)\xi}.
    \]
\end{lemma}
\begin{proof}
For $\abs{x}<R_0$, we have $F_{\rm hc}(x)=0$, and there is nothing to prove. For $R_0\leq x\leq R$, we have $F(x)=(R-\mathsf{A})^{-1}(x-\mathsf{A})$ by \eqref{EqScatteringSolution}, and $F_{\rm hc}(x)=\frac{x-R_0}{R-R_0}I$. We write $\xi$ in an orthonormal basis of eigenvectors of the Hermitian matrix $\mathsf{A}$, whose eigenvalues $\lambda$ satisfy $\lambda\leq R_0$. Since $\lambda\mapsto\frac{x-\lambda}{R-\lambda}$ is decreasing for $\lambda<R$, we have $\frac{x-\lambda}{R-\lambda}\geq\frac{x-R_0}{R-R_0}\geq0$, and hence
\[
\braket{F(x)\xi|F(x)\xi}=\sum_\lambda\abs{\xi_\lambda}^2\left(\frac{x-\lambda}{R-\lambda}\right)^2\geq\left(\frac{x-R_0}{R-R_0}\right)^2\abs{\xi}^2=\braket{F_{\rm hc}(x)\xi|F_{\rm hc}(x)\xi}.
\]
This includes the convention $\lambda=-\infty$ of Remark~\ref{remarkinfinity}. Negative $x$ follow from $F(-x)=(P_A-P_S)F(x)$ and the same relation for $F_{\rm hc}$.

We also note that $\braket{F(x) \xi | F(x)\xi}$ is convex in $x$, as used in Footnote~\ref{footconvexbound}: by the scattering equation \eqref{EqEvenScatteringEquation},
\begin{equation}
\label{convexity}
        \partial_x^2\langle F \xi | F\xi\rangle=2\braket{F' \xi | F'\xi}+\braket{F'' \xi | F\xi}+\braket{F \xi | F''\xi}=2\braket{F' \xi | F'\xi}+\braket{F \xi |V F\xi}\geq0.
\end{equation}
\end{proof}

We add a remark about scalar potentials, to prove a claim made in Remark~\ref{remarkrangeintro}.

\begin{remark}
\label{remarkrange}
In Remark~\ref{remarkrangeintro}, we claimed that all scattering lengths $a_e$ and $a_o$ with $a_e\leq a_o$ and $a_o\geq0$ can be attained by varying the potential. To see this, consider the double delta potential, $v=2c(\delta_{-R_0}+\delta_{R_0})$, with strength $c> 0$ and range $R_0\geq 0$. This gives $a_e=R_0-\frac{1}{c}$ and $a_o=R_0-\frac{R_0}{1+R_0c}$. For $a_o\geq0$ fixed, we find $R_0=\frac{a_o+\sqrt{a_o^2+4a_o/c}}{2}$. Thus, $a_e(c)=\frac{a_o+\sqrt{a_o^2+4a_o/c}}{2}-\frac{1}{c}$. Since $a_e(c)$ is a continuous function of $c>0$ with $\lim_{c\to\infty}a_e(c)=a_o$ and $\lim_{c\to0^+}a_e(c)=-\infty$, we see that any $a_e<a_o$ can be achieved, while $a_e=a_o$ is obtained in the limit $c\to\infty$, which is the hard core of radius $R_0=a_o$.
\end{remark}

\section{Remarks about the Lai--Sutherland and Yang--Gaudin models}
\label{AppendixB}
\subsection{Yang--Gaudin}
Consider the spin-$1/2$ Yang--Gaudin model, which is the spin-$1/2$ Fermi gas with point interactions, $v(x)=2c \delta(x)$. If the total spin-$z$ is fixed to  $(N-2M)/2$, the Bethe ansatz gives the following equations for the thermodynamic limit in which $ N,M,L\to\infty $ proportionally \cite{yang1967some},
\begin{align}
2\pi\sigma(\Lambda)&=-\int_{-B}^{B}\frac{2c\sigma(\Lambda')d \Lambda'}{c^2+(\Lambda-\Lambda')^2}+\int_{-Q}^{Q}\frac{4c f(k)d k}{c^2+4(k-\Lambda)^2}\label{EqYG1},\\
2\pi f(k) &= 1+\int_{-B}^{B}\frac{4c\sigma(\Lambda')d \Lambda'}{c^2+4(k-\Lambda')^2}\label{EqYG2},\\
\rho=N/L&=\int_{-Q}^{Q} f(k) d k,\quad M/L=\int_{-B}^{B}\sigma(\Lambda)d\Lambda,\label{EqYG3}\\
e_{\textnormal{YG}}=E_{\textnormal{YG}}/L&=\int_{-Q}^{Q}k^2 f(k)d k \label{EqYG4}.
\end{align}
Here, $e_{\textnormal{YG}}$ denotes the energy density of a given solution.
By the Lieb--Mattis result \cite{lieb1962theory}, the ground state is known to have total spin 0, so that it satisfies $M=N/2$. This is achieved for $B\to \infty$, which follows from \eqref{EqYG1} and \eqref{EqYG3}. It is not hard to show (see e.g.\ \cite{agerskov2023one}) that, given a solution to \eqref{EqYG1}--\eqref{EqYG4} with $B=\infty$, one finds
\[
    e_{\textnormal{YG}}=\frac{\pi^2}{3}\rho^3\left(1+2\ln(2)\rho a_e+\mathcal{O}\left((\rho/c)^2\right)\right),
\]
where $a_e=-2/c$ is the even-wave scattering length of the interaction $v(x)=2c \delta(x)$ (note $a_o=0$ for a point interaction).

While the technique used to solve the Yang--Gaudin model \cite{yang1967some,gaudin1967systeme,gaudin2014bethe} is similar to that for the Lieb--Liniger model \cite{lieb1963exact}, the mathematical rigour is somewhat less clear. More specifically, while the algebraic equations obtained with the Bethe ansatz for the Lieb--Liniger model are known to have unique solutions \cite{yang1969thermodynamics}, this is unclear for the ones found for the Yang--Gaudin model. Furthermore, even if existence and uniqueness of solutions to the relevant set of algebraic equations is assumed, one needs to argue that the ground state is among the Bethe ansatz states to begin with. For the Lieb--Liniger model, this follows from uniqueness of the bosonic ground state, and the fact that Girardeau's $c=\infty$ ground state is of the Bethe ansatz shape. More generally, it was proved by Dorlas in \cite{cmp/1104252974} that the Lieb--Liniger Bethe ansatz states form a complete orthogonal set, but none of these arguments are directly valid for the spin-$J$ (or even spin-$1/2$) Fermi gas. Hence, there does not seem to be a proof in the literature that the Bethe ansatz is valid for the ground state. However, Theorem~\ref{TheoremDiluteFermiGroundStateEnergy} and the known expansion of the Heisenberg ground state energy (see \eqref{some310}) prove that the first two terms of \eqref{EqYGGroundStateEnergy} are correct in the dilute limit.

\subsection{Lai--Sutherland}
The following lemma was used in \eqref{some310} and in Section~\ref{SectionLowerBoundProof}. In the notation used there, $\frac{1}{N}E_{\textnormal{LS}}^N=\epsilon^{\textnormal{LS}}_J$ and $e_{\textnormal{LS}}=\bar\epsilon_J$.
\begin{lemma}\label{LemmaLaiSutherlandFiniteN}
    Let $E_{\textnormal{LS}}^N$ denote the ground state energy of the $N$-spin Lai--Sutherland model $H_{\textnormal{LS}}=\sum_{i=1}^{N}P_S^{i,i+1}$ with periodic boundary conditions (for some fixed spin $J$), and let $e_{\textnormal{LS}}=\lim\limits_{N\to\infty}\frac{1}{N}E_{\textnormal{LS}}^N$ be the thermodynamic ground state energy per site. Then, \begin{equation}
e_{\textnormal{LS}}-\frac{1}{N}\leq \frac{1}{N}E_{\textnormal{LS}}^{N}\leq e_{\textnormal{LS}}+\frac{1}{N}.
\end{equation}
\end{lemma}
\begin{proof}
For $ M>N $, we have 
\begin{equation}
\label{some13011}
E_{\textnormal{LS}}^N\leq (N-1)\frac{E_{\textnormal{LS}}^M}{M}+1\leq N\frac{E_{\textnormal{LS}}^M}{M}+1.
\end{equation}
This can be seen by taking a translation-symmetric ground state of a chain of length $ M $, using its reduced density matrix on $N$ consecutive spins as a (mixed) trial state for the chain of length $N$, and using that $N-1$ terms in the Hamiltonian all have the same expectation value by translation symmetry, while the last term is estimated by the maximal eigenvalue 1 of a projection operator. In particular, $E^N_{\textnormal{LS}}/N\leq E^M_{\textnormal{LS}}/M+1/N$ for $M>N$, so that $\limsup_{N\to\infty}E^N_{\textnormal{LS}}/N\leq\liminf_{M\to\infty}E^M_{\textnormal{LS}}/M$, and the limit $e_{\textnormal{LS}}$ exists. 

On the other hand, for $m\in\mathbb{N}$, take a translation-symmetric ground state of the chain of length $N$, concatenate $m$ of them to obtain a trial state for the chain of length $mN$, and again use that all terms in the Hamiltonian corresponding to spins in the interior of the chain of length $N$ give the same energy contribution. This gives
\begin{equation}
\label{some13012}
E_{\textnormal{LS}}^{mN}\leq m(N-1)\frac{E_{\textnormal{LS}}^{N}}{N}+m\leq mN\frac{E_{\textnormal{LS}}^{N}}{N}+m.
\end{equation}
Taking the limits $ m\to \infty $ and $ M\to\infty $ in \eqref{some13011} and \eqref{some13012} gives the desired result.
\end{proof}

\end{appendices}

\bibliographystyle{ieeetr}
\bibliography{bibliography}
\end{document}